\documentclass[a4paper,onecolumn, superscriptaddress,10pt, accepted=2021-02-06, issue=2, volume=3, shorttitle=papers/compositionality-3-2]{compositionalityarticle}
\pdfoutput=1

\usepackage[utf8]{inputenc}
\usepackage[english]{babel}
\usepackage[T1]{fontenc}
\usepackage{amsmath}

\newtheorem{theorem}{Theorem}
\newtheorem{lemma}{Lemma}
\newtheorem{definition}{Definition}

\newtheorem{corollary}{Corollary}
\newtheorem{example}{Example}
\newtheorem{assumption}{Assumption}
\newtheorem{remark}{Remark}

\newtheorem{notConv}{Notational convention}

\newenvironment{proof}{\noindent\textbf{Proof: }}{\hfill$\square$\\}

\usepackage{tikz}

\usepackage{amssymb}
\usepackage{bm}
\usepackage{bbold}
\usepackage{mathtools}
\usepackage{thmtools}
\usepackage{thm-restate}

\usepackage{array}
\usepackage{pdfpages}
\usepackage{wrapfig}
\usepackage{subfigure}
\usepackage{graphicx}
\usepackage{longtable}
\usepackage{booktabs}
\usepackage{pdflscape} 

\makeatletter
\def\bstr{b}
\def\bfstr{bf}
\def\cstr{c}
\def\fstr{f}
\def\strLst{A,B,C,D,d,E,F,G,H,I,J,K,L,M,N,O,P,Q,R,S,T,U,V,W,X,Y,Z}

\newcommand{\MkB}[1]{\expandafter\def\csname\bstr#1\endcsname{\mathbb{#1}}}
  \@for\i:=\strLst\do{%
    \expandafter\MkB \i     }
    
\newcommand{\MkBF}[1]{\expandafter\def\csname\bfstr#1\endcsname{\mathbf{#1}}}
  \@for\i:=\strLst\do{%
    \expandafter\MkBF \i     }

\newcommand{\MkCal}[1]{\expandafter\def\csname\cstr#1\endcsname{\mathcal{#1}}}
  \@for\i:=\strLst\do{%
    \expandafter\MkCal \i     }
    
\newcommand{\MkFrak}[1]{\expandafter\def\csname\fstr#1\endcsname{\mathfrak{#1}}}
  \@for\i:=\strLst\do{%
    \expandafter\MkFrak \i     }
    
\makeatother

\newcommand{\Lin}[1]{\mathop{\mathsf{Lin}}(#1)}

\newcommand{\rSpan}[5]{\left(#1\xleftarrow{#2}#3\xrightarrow{#4}#5\right)}

\newcommand{\rSpanAlt}[5]{\left(#1\xhookleftarrow{#2}#3\xhookrightarrow{#4}#5\right)}

\newcommand{\cSquare}[1]{\Box(#1)}

\newcommand{\rMatch}[2]{\mathsf{M}_{#1}(#2)}

\newcommand{\LinAc}[1]{\overline{\mathsf{Lin}}(#1)}

\newcommand{\ac}[1]{\mathsf{#1}} 

\newcommand{\Shift}{\mathsf{Shift}}

\newcommand{\Trans}{\mathsf{Trans}}

\newcommand{\pB}[1]{\mathsf{PB}(#1)}
\newcommand{\pO}[1]{\mathsf{PO}(#1)}
\newcommand{\FPC}[1]{\mathsf{FPC}(#1)}

\newcommand{\mono}[1]{\mathsf{mono}(#1)}

\newcommand{\epi}[1]{\mathsf{epi}(#1)}

\newcommand{\mor}[1]{\mathsf{mor}(#1)}

\newcommand{\iso}[1]{\mathsf{iso}(#1)}

\newcommand{\obj}[1]{\mathsf{obj}(#1)}

\newcommand{\Match}[2]{\mathsf{M}_{#1}(#2)}

\newcommand{\sqMatch}[2]{\mathsf{M}^{sq}_{#1}(#2)}

\newcommand{\sqRMatch}[2]{\mathbf{M}^{sq}_{#1}(#2)}

\newcommand{\comp}[3]{#1 \stackrel{#2}{\blacktriangleleft} #3}

\newcommand{\sqComp}[3]{#1 \stackrel{#2}{\sphericalangle} #3}

\newcommand{\mIO}{\mathop{\varnothing}}

\newenvironment{romanenumerate}{
\begingroup
\renewcommand\labelenumi{(\roman{enumi})}
\renewcommand\theenumi\labelenumi
\begin{enumerate}}{\end{enumerate}
\endgroup}

\def\eqdef{\mathrel{:=}}
\def\iff{\mathrel{\Leftrightarrow}}
\def\iffeq{\mathrel{:\Leftrightarrow}}

\makeatletter
\renewcommand*{\@opargbegintheorem}[3]{\trivlist
      \item[\hskip \labelsep{\bfseries #1\ #2}] \textbf{(#3)}\ \itshape}
\makeatother

\newcommand{\inputtikz}[1]{%
 \ensuremath{\vcenter{\hbox{\includegraphics{diagrams/#1.pdf}}}}%
}

\newcommand{\inputtikzB}[2]{%
 \raisebox{#2cm}{\includegraphics{diagrams/#1.pdf}}%
}

\colorlet{h1color}{blue!40!black} 
\colorlet{h2color}{orange!90!black} 
\colorlet{h3color}{blue!40!white} 
\colorlet{h4color}{green!40!black} 
\colorlet{h5color}{yellow!80!black} 

\usepackage[numbers,sort&compress]{natbib}

\usepackage{hyperref}

\begin{document}
	
\title{Compositionality of Rewriting Rules with Conditions}

\date{}
\author{Nicolas Behr}
\email{nicolas.behr@irif.fr}
\homepage{http://nicolasbehr.com}
\orcid{0000-0002-8738-5040}
\thanks{corresponding author; the work of N.B.~was supported by a \emph{Marie Sk\l{}odowska-Curie Individual fellowship} (Grant Agreement No.~753750 -- RaSiR)}

\author{Jean Krivine}
\email{jean.krivine@irif.fr}
\homepage{https://www.irif.fr/\string~jkrivine/homepage/Home.html}
\orcid{0000-0001-7261-7462}

\affiliation{Universit\'{e} de Paris, CNRS, IRIF, F-75205, Paris, France}

\maketitle

\begin{abstract}
  We extend the notion of compositional associative rewriting as recently studied in the rule algebra framework literature to the setting of rewriting rules with conditions. Our methodology is category-theoretical in nature, where the definition of rule composition operations encodes the non-deterministic sequential concurrent application of rules in Double-Pushout (DPO) and Sesqui-Pushout (SqPO) rewriting with application conditions based upon $\cM$-adhesive categories. We uncover an intricate interplay between the category-theoretical concepts of conditions on rules and morphisms, the compositionality and compatibility of certain shift and transport constructions for conditions, and thirdly the property of associativity of the composition of rules.
\end{abstract}

\section{Introduction and relation to previous work}

Graph rewriting has emerged as a powerful formalism to represent complex systems whose dynamics can be captured by a finite set of rules. %
The \emph{rule-based modeling} approach, originally introduced by V.~Danos and C.~Laneve in the early 2000s~\cite{Danos:2003ab,Danos:2004aa,Danos:2003aa}, has developed into one of the main frameworks for the study of biochemical reaction systems (in the form of the two main frameworks \textsc{Kappa}~\cite{Danos:ab,Boutillier:2018aa} and \textsc{BioNetGen}\cite{Blinov:2004aa,Harris:2016aa}). The approach proposes to model protein-protein interaction networks using graph rewriting models, in which proteins are the vertices of a graph whose connected components denote molecular complexes. %
A paradigmatic application scenario for rule-based methods in the systems biology community is the study of signaling pathways~\cite{Danos:ab,danos2012graphs}, which are highly intricate biochemical reaction networks ubiquitous in living cells.

While the algorithmic aspects of graph rewriting are well-studied, programming language approaches to modeling with graphs are to date still a comparatively underdeveloped topic. %
Contrary to classical term rewriting, the notion of a \emph{match} of a graph rewriting rule and its \emph{effects} on a term (a graph) is subject to various definitions, allowing more or less control over possible rewrites. %
In addition, the mere nature of the graphs that are being rewritten impacts both the algorithmic design and expressiveness of graph rewriting. %
Category theory is a practical toolkit for equipping graphs with well-defined operational semantics. \emph{Double-Pushout (DPO) rewriting}~\cite{CorradiniMREHL97} is a widespread technique, partly because it does not yield side effects when rules are applied (which makes it amenable to static analysis for instance). %
However, when a graph rewriting rule entails node deletion, DPO semantics will not allow a match of such a rule to trigger if the node that is deleted is connected outside the domain of the match (which would yield side-effects). This has limited the practicality of DPO semantics in the context of biological modeling, where more permissive techniques have been employed. \emph{Sesqui-Pushout (SqPO) rewriting}~\cite{Corradini_2006} in particular is the technique that is used to rewrite \textsc{Kappa} graphs~\cite{Danos:ab}.

Quite orthogonal to the issue of defining rule matches and effects, having access to a fine-grained control over rule triggering is a key issue when graph rewriting is used as a modeling language. To this aim, graph rewriting rules have been equipped with \emph{application conditions}~\cite{habel2009correctness,ehrig2014mathcal}, which can be seen as constraints that need to be checked ``on the fly'' when a rewrite rule is applied. 

This paper presents a \emph{compositional} variant of DPO and SqPO-type rewriting for rules with conditions in a general category-theoretical setting. From a mathematical perspective, while the framework of both types of rewriting developed here relies upon the original definitions of DPO-rewriting  (see e.g.~\cite{ehrig2014mathcal}) and of SqPO-rewriting~\cite{Corradini_2006}, new developments are necessary in order to obtain the desired compositionality properties. For rules without conditions, one of the core technical obstacles has proved to be establishing a \emph{compositional associativity} property for sequential compositions of rewriting rules, which for the DPO-type case has been achieved in~\cite{bp2018,bp2019-ext,bdg2016}, and for the SqPO-type case in~\cite{nbSqPO2019}. The latter work also established a novel \emph{compositional concurrency} property for SqPO-type rewriting theories. Lifting these results to the settings of rules with application conditions is the core contribution of this work. 

The main motivation for our search for compositional rewriting theories is two-fold: in the setting of rewriting without conditions, the notion of \emph{rule algebras}~\cite{bdg2016,bdgh2016,bp2018,nbSqPO2019} has been developed as a new mathematical framework to encode the concurrent and combinatorial interactions of rewriting rules, which in particular allows one to develop principled and novel analysis techniques for stochastic rewriting systems~\cite{bdg2016,bp2018,bdg2019}. Especially the recent results of~\cite{bdg2019} hint at the intimate interplay of design choices in constructing rewriting systems with regards to their dynamical properties, and with the potential of greatly improving the tractability of the analysis of such systems via judiciously chosen \emph{constraints} on objects and rewriting rules (such as implemented in the form of rigidity constraints in the \textsc{Kappa} formalism~\cite{danos2008rule}). 

Our second main motivation originates in the desire to analyze rewriting systems \emph{statically}, rather than via simulation-based techniques. While the traditional rewriting theory generally approaches this problem from the viewpoint of \emph{derivation traces} (i.e.\ sequences of rule applications to a given input graph), we posit that a viable alternative approach may consist in focusing instead on \emph{sequential rule compositions}, which in particular when combined with application conditions is anticipated to yield a powerful framework to study the \emph{causality} of rewriting systems. %
For certain specialized applications of DPO-type graph rewriting, namely those in the well-established field %
of \emph{chemical graph rewriting}~\cite{Benk__2003,banzhaf_et_al:DR:2015:4968,Andersen_2016}, such types of analyses have already proven very fruitful~\cite{andersen2018rule,andersen2018towards,Fagerberg_2018,Andersen_2019}. Therefore, we believe that our compositional refinements described in the present paper can provide a significant contribution to future algorithm developments in this field.

\subsection{Related work}

\subsubsection{Traditional rewriting literature}

First and foremost, the work presented in this paper relies heavily upon the rich literature on categorical rewriting theories with its almost 50 year long history. The introduction of the seminal concept of \emph{adhesive categories} by Lack and Soboci{\'{n}}ski in the early 2000s~\cite{ls2004adhesive,lack2005adhesive} and its refinement to $\cM$-adhesive categories in the work of Ehrig et al.~\cite{ehrig2006adhesive,ehrig2010categorical,Habel:2012aa,Braatz:2010aa,ehrig2014mathcal} resulted in a powerful mathematical framework for \emph{Double-Pushout (DPO)} rewriting, and in particular permitted to unify numerous concepts from the earlier graph rewriting literature~\cite{DBLP:conf/gg/1997handbook}. While $\cM$-adhesive categories have been demonstrated to be capable of describing a broad range of data structures of relevance in applications of rewriting theories (see e.g.\ \cite{ehrig:2006aa} and \cite{bp2019-ext} for curated lists of examples), it is only via a second key refinement of the rewriting theory that more elaborate data structures such as e.g.\ chemical molecules may be encoded: the theory of \emph{constraints} and \emph{application conditions}. This theory had been a well-established component of some of the earliest categorical formulations of graph rewriting theories since the 1980s by Habel et al.~\cite{EhrigHabel1986,Habel1996}, and was later generalized to the $\cM$-adhesive setting as utilized in the present paper in~\cite{habel2009correctness,Pennemann:aa,ehrig2014mathcal}.

Building upon this extensive body of work, and motivated by applications to modeling in the life sciences (cf.\ \cite{BK2020} and comments below), we demonstrate in the present paper that by requiring the underlying $\cM$-adhesive category to satisfy certain additional technical assumptions, one obtains a refined variant of DPO-semantics that is well suited for developing novel static analysis techniques for rewriting systems. Referring to Remarks~\ref{rem:Shift} and~\ref{rem:dpoCC} in the main text for the precise technical details of the relationship between our ``refined'' and the ``traditional'' DPO semantics, our work moreover extends the traditional DPO-theory by providing an \emph{associativity theorem} which had not been previously known in the literature (cf.\ Theorem~\ref{thm:assocAC}).

The second type of categorical rewriting semantics considered in this paper is the \emph{Sesqui-Pushout (SqPO)} rewriting theory introduced in~\cite{Corradini_2006}. Despite the importance of this type of semantics especially in the modeling of biochemical reaction systems (cf.\ Section~\ref{sec:boChem}), the approach had been considerably less well-developed than the DPO-variant, lacking in particular a theory of constraints and application conditions 
as well as notions of \emph{concurrency} and \emph{associativity} theorems. While in~\cite{reversibleSqPO} an attempt had been made to identify a special sub-class of SqPO-type rewriting systems over $\cM$-adhesive categories in which the rewriting semantics coincides with that of the DPO-setting, and thus allowing to reuse the theory of application condition from this setting, this special case of \emph{``reversible''} SqPO-rules constitutes too strong a restriction in order to study the rewriting systems of relevance e.g.\ in biochemistry (Section~\ref{sec:boChem}). Building upon our earlier work on SqPO-rewriting in the setting of rules without conditions~\cite{nbSqPO2019}, and taking advantages of the close analogies with our refined DPO-type framework as presented in comprehensive technical detail in the present paper, we develop these missing elements of SqPO-rewriting theory (cf.\ Section~\ref{sec:SqPO}). Referring to~\cite{Corradini_2006} for an extended discussion, it might finally be worthwhile to note that SqPO-semantics in the setting of ($\cM$-) linear rules and ($\cM$-) monic matches (as is the case in the present paper) is well known to coincide with \emph{single-pushout (SPO)} semantics~\cite{loeweSPO}.

\subsubsection{Bio- and organo-chemistry}\label{sec:boChem}

An intriguing trend in modern rewriting theory over the past 15 years has been the development of rule-based modeling approaches in biochemistry (\textsc{Kappa}~\cite{Boutillier:2018aa}, \textsc{BioNetGen}~\cite{Harris:2016aa}) and in organic chemistry (\textsc{M{\O}D}~\cite{Andersen_2016}), which are based upon SqPO- and DPO-type semantics, respectively. While an detailed discussion of the numerous sophisticated theoretical and algorithmic techniques (including in particular various notions of \emph{static analysis techniques}, cf.\ e.g.\ \cite{Danos:ab,danos2008rule,danos2012graphs,Murphy_2010,Petrov_2012,Abou_Jaoud__2016,Camporesi_2017,Feret2012137} and \cite{Benk__2003,banzhaf_et_al:DR:2015:4968,Fagerberg_2018,andersen2018rule,Andersen_2019}) is beyond the scope of the present paper, suffice it here to comment on a salient technical point: notably, many of the developments in the aforementioned frameworks were not explicitly rooted in categorical rewriting theory itself, despite the foundations of the two fields upon this type of theory. We recently demonstrated in~\cite{BK2020} that based upon the results of the present paper, one may not only reformulate equivalently the current implementations of chemical rewriting theories, but one may also take advantage of our novel compositionality and associativity results in order to formulate rule algebra and tracelet theories (see below) to develop new approaches to the static analysis of complex reaction systems and their dynamics. The encodings of chemical molecules and their reactions as reported in~\cite{BK2020} in both the bio- and the organo-chemical settings are based upon typed variants of \emph{undirected simple graphs} that satisfy suitable sets of constraints, which is why we will take undirected simple graphs and rewriting rules thereof as a running example throughout this paper.

\subsubsection{Rule algebra and tracelet theory}

Rule-based modeling approaches for chemical reaction systems are part of a larger class of theories known as \emph{stochastic rewriting theories}. Over the past five years, we have demonstrated that one may express such systems based upon the general theory of \emph{continuous-time Markov chains (CTMCs)}~\cite{norris}, leading to the so-called \emph{stochastic mechanics} framework~\cite{bdg2016,bp2018,bp2019-ext,nbSqPO2019,bdg2019,BK2020}. At the core of this framework is the notion of \emph{rule algebras}, which encode the non-determinism in the choices of admissible matches in sequential compositions of rules via the construction of a certain \emph{binary operation} (on a vector space whose basis is indexed by equivalence classes of linear rules). As explained in detail in~\cite{bp2019-ext,bdg2019}, since the ``blueprint'' of this type of construction in the general mathematics and in particular combinatorics literature is provided by the operations of derivative and multiplication by formal variables acting upon the space of formal power series, the aforementioned binary operation must satisfy two fundamental technical properties: it must be \emph{associative} (in the standard sense of associativity of binary operations), and it must possess a \emph{neutral element}. Referring to loc.\ cit.\ for the precise technical details, the former property requires an \emph{associativity theorem} for sequential rule compositions, while the latter property requires the underlying category to possess a (strict) $\cM$-initial object. Finally, in order to formulate the stochastic rewriting theories themselves, it is necessary to endow the rule algebras with \emph{canonical representations}, the construction of which in turn hinges on the \emph{concurrency theorems} for the respective sort of rewriting. The technical results presented in the present paper have been recently applied in~\cite{BK2020} to establish precisely the rule algebra theories for linear rules with conditions both in DPO- and in SqPO-semantics. 

A second recent line of developments has been the introduction of the theory of \emph{tracelets} in~\cite{behr2019tracelets}, which also relies heavily on the technical results of the present paper. Intuitively, while the aforementioned rule algebraic constructions aim to reason based upon rule algebra products of linear rules (which, in a sense, encode ``sums over all possible ways to compose rules''), the concept of tracelets is more directly related to the classical rewriting-theoretical notion of \emph{derivation traces}~\cite{DBLP:conf/gg/1997handbook,ehrig:2006aa,Corradini_2006}.  We refer the interested readers to~\cite{behr2019tracelets} for an in-depth review of the relationship of tracelet theory with the traditional concepts of concurrency and related static analysis techniques, mentioning here only the fact that yet again this new theory relies upon the notion of associative compositional rewriting theories.

\subsection{A motivating example: rewriting simple graphs}\label{sec:motEx}

Since our main constructions will be somewhat technical, let us start with a simple example in order to provide the readers with an intuitive picture of the main concepts. To this end, consider the task of defining a sound notion of rewriting for \emph{undirected simple graphs}, the type of graphs where at most one undirected edge may exist between any two given vertices of the graph, as opposed to undirected multigraphs. While there exist various attempts in the literature to encode simple graphs directly via a suitable category-theoretical construction (cf.\ e.g.\ 
\cite{adamek1990abstract,Corradini_2006, braatz2008graph, reversibleSqPO,Braatz:2010aa}), the disadvantage of such an approach consists of the fact that these categories are not $\cM$-adhesive, and thus do not permit to profit from the constructions for $\cM$-adhesive categories available in the general DPO-rewriting framework. In fact, constructions of simple graph categories such as $\textbf{RDFGraph}$ of~\cite{braatz2008graph} (for a notion of directed typed simple graphs called RDF graphs) or $\textbf{SGraph}$ of~\cite{adamek1990abstract, Corradini_2006} (of undirected simple graphs) are well known to lack uniqueness of pushout complements (POCs), quintessential in defining DPO-semantics (Section~\ref{sec:DPO}) as well as (compositional) SqPO-semantics (Section~\ref{sec:SqPO}) in the first place. Albeit several elaborate workarounds to this problem have been considered in the literature (such as e.g.\ the \emph{minimal-POC-PO} semantics for RDF-graph rewriting in~\cite{braatz2008graph}), these workarounds introduce alterations to the rewriting semantics that do not necessarily reflect the intuitions one might have for simple graph rewriting from a mathematics or general applied sciences background.

To circumvent these issues, we will instead start our construction from a ``host category'' $\mathbf{uGraph}$ of \emph{undirected multigraphs}\footnote{Interestingly, the category $\mathbf{uGraph}$ constitutes one of the simplest non-trivial examples of a category that is $\cM$-adhesive, but \emph{not} adhesive, unlike the prototypical example of an adhesive category $\mathbf{Graph}$~\cite{ls2004adhesive}.}~\cite{padberg2017towards,bp2019-ext,BK2020} that is itself $\cM$-adhesive and satisfies all requirements to admit compositional rewriting theories of both DPO- and SqPO-type. Imposing the aforementioned structural constraint prohibiting parallel edges via the framework of conditions, we obtain an algorithmically versatile implementation of simple graphs. Following the rewriting theory paradigms, one may envision manipulating graphs via application of \emph{graph rewriting rules}. Intuitively, specifying a rewriting rule amounts to providing the following data:
\begin{itemize}
  \item An \emph{\underline{I}nput pattern} $I$,
  \item an \emph{\underline{O}utput pattern} $O$, and
  \item a \emph{\underline{K}ontext pattern} $K$ together with embedding morphisms $i:K\hookrightarrow I$ and $o:K\hookrightarrow O$.
\end{itemize}
Then in order to apply a rule $r\eqdef(O\xhookleftarrow{o}K\xhookrightarrow{i}I)$ to a graph $G$, one has to 
\begin{enumerate}
\item specify an embedding $m:I\hookrightarrow G$ of the input pattern $I$ into $G$ (i.e.\ one has to select a particular occurrence of the pattern $I$ in $G$), referred to as a \emph{match},
\item in $m(I)\subset G$, replace $m(I)$ with $m\circ i(K)$, and
\item ``glue'' a copy of $O$ onto $m\circ i(K)$ (according to the embedding morphism $o:K\hookrightarrow O$).
\end{enumerate}
For example, one may consider the rewriting rules
\begin{equation}
  e_{+}\eqdef
  \left(%
  \inputtikz{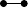} \hookleftarrow
\inputtikz{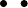}
\hookrightarrow
\inputtikz{TVG}
\right)\,,\quad
  e_{-}\eqdef
  \left(%
  \inputtikz{TVG} 
  \hookleftarrow
  \inputtikz{TVG} 
  \hookrightarrow 
  \inputtikz{TVEG}  \right)\,,
\end{equation}
which implement the manipulations of \emph{adding} a new edge between two vertices ($e_{+}$) and \emph{deleting} an edge between two vertices ($e_{-}$). However, in the setting of rewriting of \emph{simple} graphs, the above tentative definition for the rewriting of graphs is as of yet incomplete. For example, given the graph below,
\[
  \inputtikz{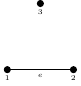}\,,
\]
if we were to apply our edge creation rule at a match comprising the vertices marked $1$ and $2$, we would produce the non-simple graph depicted below (with $e'$ the newly produced edge):
\[
  \inputtikz{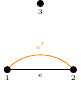}
\]
Therefore, in order to ensure that our transformations via application of rules keep intact the constraint of graphs to be simple, we need to endow the rules with \emph{application conditions}. While this approach is well known in the rewriting literature, it turns out that, via a careful re-implementation of the traditional framework, some interesting mathematical structures may be uncovered: rules with conditions are endowed with a structure of composition operation that allows one to synthesize sequences of causally sound rewriting steps \emph{without} reference to a host object. Crucially, this operation can be shown to be \emph{associative}. Contact with the traditional techniques of rewriting is then made in the form of a suitable adaption of the concurrency theorem, whereby a sequence of rule applications to a given object can always be equivalently described by a ``one-shot'' application of a sequential composite of the rules to that object (i.e.\ in a sense characterizing the effect of the sequential rule applications by the application of a single composite rule). The most subtle part of our results then consists in a certain compositional associativity property present in triple sequential compositions of rules.

To relate to the example at hand, in addition to using the standard application condition technique  to express \emph{constraints} on rules (such as the constraint that the edge creation rule should only be applied to a graph if its input is matched to two vertices that are not already linked), we will also be able to compute \emph{causal information}, such as e.g.\ that applying the edge creation rule $e_{+}$ followed by applying the edge deletion rule $e_{-}$ (in a fashion such as to delete the previously created edge) will lead to a rule with no effect, which can only be applied at two vertices that cannot be already linked (transforming the graph identically).  Finally, we will present examples of the property of associativity in Examples~\ref{ex:assocDPO} and~\ref{ex:assocSqPO} in the setting of DPO- and SqPO-type rewriting theories, respectively.

\subsection{Overview of main results and structure of the paper}

The main results of this paper are summarized as follows (cf.\ Figure~\ref{fig:structure}):
\begin{itemize}
\item We provide a self-contained account of all category-theoretical prerequisites and specialized definitions extracted from the rich literature on categorical rewriting theories (Sections~\ref{sec:CTP} and~\ref{sec:cond}; Appendix~\ref{sec:appendix}), intended not least as an entry-point for the general applied category theory audience.
\item For \emph{Double-Pushout (DPO) rewriting} (Section~\ref{sec:DPO}), we identify a set of sufficient assumptions on the underlying categories such that the resulting rewriting theories have suitable \emph{compositionality properties} (in the sense of a certain form of concurrency and associativity properties of sequential rule compositions). This part of our theory is thus to a large extent a careful ``fine-tuning'' of results from the traditional DPO-rewriting literature, combined with a number of original results.
\item For \emph{Sesqui-Pushout (SqPO) rewriting} (Section~\ref{sec:SqPO}), we present the first-of-its-kind category-theoretical compositional SqPO-type rewriting theory for rules with conditions, profiting from fruitful analogies to results in the DPO-type theory as presented in the first part of the paper.
\end{itemize}

\begin{figure}[htp]
\centering
\inputtikz{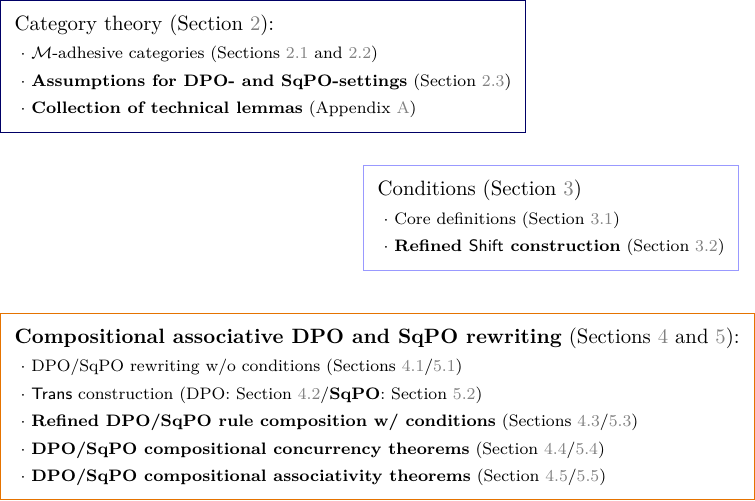}
\caption{\label{fig:structure}Structure and \textbf{original contributions} of the paper.}
\end{figure}

\section{Category-theoretical preliminaries}\label{sec:CTP}

Rewriting in its modern formulations is a concept that heavily relies on specific types of categorical structures. In this section, we collect all the necessary prerequisites that allow one to formulate consistent frameworks of associative rewriting theories with conditions on objects and morphisms, in both the \emph{Double-Pushout (DPO)} and the \emph{Sesqui-Pushout (SqPO)} approaches. While many of the mathematical details of these setups are by now standard in the literature, we will emphasize the specific additional conditions that are required in order to guarantee \emph{associativity} (in the sense of~\cite{bp2018}) and \emph{compositionality}, where both of these properties concern concurrent compositions of rules with conditions.

\subsection{$\cM$-adhesive categories}\label{sec:MadhCats}

We begin by quoting a number of essential definitions and standard results from the literature, where our main references will be~\cite{Braatz:2010aa,habel2009correctness,ehrig2014mathcal} (see 
also~\cite{radke2016theory,schneider2017symbolic} for some more recent works). Let us first recall the notion of $\cM$-adhesive categories, which is the most general mathematical setting currently known that allows one to define DPO- and SqPO-type rewriting theories, and which generalizes adhesive categories~\cite{lack2005adhesive}.
\begin{definition}[\cite{Braatz:2010aa}, Def.~2.1]\label{def:Madh}
Let $\bfC$ be a category, and let $\cM$ be a class of \emph{monomorphisms}. Then the data $(\bfC,\cM)$ defines an \emph{$\cM$-adhesive category} if the following requirements are satisfied:
\begin{romanenumerate}
\item The class $\cM\subset\mor{\bfC}$ contains all isomorphisms and is \emph{closed under}
\begin{enumerate}
\item \emph{composition},
\begin{equation}
\forall g\circ f\in \mor{\bfC}:\; f\in \cM\land g\in \cM\Rightarrow g\circ f\in \cM\,,
\end{equation}
\item \emph{decomposition},
\begin{equation}
\forall g\circ f\in \cM:\; g\in \cM\Rightarrow f\in \cM\,.
\end{equation}
\end{enumerate}
\item $\bfC$ \emph{has pushouts and pullbacks along $\cM$ morphisms}, i.e.\ pushouts of spans and pullbacks of cospans where at least one of the two morphisms of the (co-) span is in $\cM$ exist. 
\item \emph{$\cM$ morphisms are closed under pushouts and pullbacks:} if in the diagram below $(1)$ is a pushout, then $m\in \cM$ implies $n\in \cM$, while if $(1)$ is a pullback, then $n\in \cM$ implies $m\in \cM$:
\begin{equation}
\inputtikz{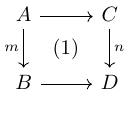}
\end{equation}
\item \emph{Pushouts along $\cM$-morphisms are $\cM$-van Kampen ($\cM$-VK) squares (also referred to as vertical weak VK squares in the literature):} given a commutative cube in $\bfC$ as shown below, where the bottom square is a pushout along an $\cM$-morphism $m\in \cM$, the back faces are both pullbacks, and if $b,c,d\in \cM$, then the top face is a pushout if and only if the two front faces are pullbacks.
\begin{equation}\label{eq:VKsquareDef}
\inputtikz{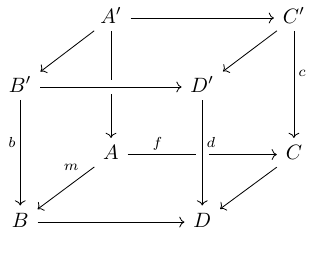}\,.
\end{equation}
\end{romanenumerate}
For certain applications, it is also of interest to consider a variant of the definition called \emph{weak $\cM$-adhesive categories}, which are categories in which all of the above axioms hold except for the $\cM$-VK property; the latter is modified to the \emph{weak $\cM$-VK property}:
\begin{itemize}
\item[(iv)'] \emph{Pushouts along $\cM$-morphisms are \emph{weak $\cM$-van Kampen ($\cM$-VK) squares}:} given a commutative cube in $\bfC$ as shown in~\eqref{eq:VKsquareDef}, where the bottom square is a pushout along an $\cM$-morphism $m\in \cM$, the back faces are both pullbacks, and if 
\[
(b,c,d\in \cM) \quad \text{or }\quad (f\in \cM)\,,
\]
then the top face is a pushout if and only if the two front faces are pullbacks.
\end{itemize} 
\end{definition}
Strictly speaking, several points in the above definition are redundant, since the closure under isomorphisms $(i)$ and the decomposition property $(ib)$ are known to follow directly from closure under compositions and stability of $\cM$-morphisms under pullbacks (see Appendix~\ref{app:isoM} for the former property). It is moreover worthwhile to note that in some of the earlier literature (cf.\ \cite{ehrig2010categorical} for a review), several authors referred to categories satisfying properties $(i-iv)$ as \emph{\textbf{vertical} weak $\cM$-adhesive categories}, and by analogy to categories satisfying $(i-iii)$ and only the second alternative of $(iv')$ (i.e.\ that the ``horizontal'' morphisms $m$ and $f$ should be in $\cM$) as \emph{\textbf{horizontal} weak $\cM$-adhesive categories}. In practice (cf.\ e.g.\ \cite{bp2019-ext} for a review and a list of examples), most categories considered in the rewriting literature in fact qualify as both horizontal and vertical weak $\cM$-adhesive categories, yet since a few exceptions only satisfy the vertical variant, and since moreover in the relevant proofs only this variant is required, the standard convention has become to refer to the vertical variant simply as $\cM$-adhesive categories~\cite{Braatz:2010aa}. $\cM$-adhesive categories enjoy a number of special properties (some of which referred to in the literature as \emph{high-level replacement (HLR) properties}) that will be important to our main constructions. We collect these properties in Appendix~\ref{app:lem}. 

As advocated in particular in the work of Ehrig~\cite{Braatz:2010aa}, one of the optimal compromises for a general setting in which DPO (and, as we shall see, also SqPO) rewriting theories involving constraints on objects and morphisms  can be formulated efficiently is provided by $\cM$-adhesive categories with certain additional special properties. A central role in this setup is played by the following concepts:

\begin{definition}[$\cM$-initial object; \cite{Braatz:2010aa}, Def.~2.5]\label{def:Minit}
An object $\mIO$ of an $\cM$-adhesive category $\bfC$ is defined to be an \emph{$\cM$-initial object} if for each object $A\in \obj{\bfC}$ there exists a unique monomorphism $i_A:\mIO\hookrightarrow A$, which is moreover required to be in $\cM$. An $\cM$-initial object $\mIO$ is said to be \emph{strict} if for each object $X\in \obj{\bfC}$, every morphism $X\rightarrow \mIO$ must be an isomorphism.
\end{definition}

\begin{lemma}[\cite{Braatz:2010aa}, Fact~2.6]\label{lem:binaryCoproducts}
  If an $\cM$-adhesive category $\bfC$ possesses an $\cM$-initial object $\mIO\in \obj{\bfC}$, then the category has \emph{finite coproducts}, and moreover the coproduct injections are in $\cM$. In particular, the coproduct $A+B$ of two objects $A,B\in\obj{\bfC}$ is then given as the pushout of the span $(A\xleftarrow{i_A}\mIO\xrightarrow{i_B}B)$.
\end{lemma}

For later convenience, we present a number of consequences of an $\cM$-adhesive category possessing a (strict) $\cM$-initial object in Appendix~\ref{app:sMinit}.

\begin{definition}[Finite objects, finitary categories, finitary restrictions; \cite{Braatz:2010aa}, Def.~2.8 and Def.~4.1]\label{def:finitary}
Let $\bfC$ be an $\cM$-adhesive category. An object $A\in \obj{\bfC}$ is said to be \emph{finite}, if there exist only finitely many isoclasses of $\cM$-morphisms $B\hookrightarrow A$ into $A$ (i.e.\ if ``$A$ has finitely many $\cM$-subobjects''). $\bfC$ is \emph{finitary} if all its objects are finite. Let $\bfC_{fin}$ denote the full subcatgory of $\bfC$ spanned by finite objects, and let $\cM_{fin}$ denote the class of $\cM$-morphisms between finite objects. Then we refer to $(\bfC_{fin},\cM_{fin})$ as the \emph{finitary restriction} of $\bfC$.
\end{definition}

\begin{theorem}[\cite{Braatz:2010aa}, Thm.~3.14]~
The finitary restriction $(\bfC_{fin},\cM_{fin})$ of an $\cM$-adhesive category $\bfC$ is a finitary $\cM$-adhesive category.
\end{theorem}

\begin{definition}[Epi-$\cM$-factorizations; cf.\ e.g.\ \cite{habel2009correctness}, Def.~3]\label{def:epiM}
 An $\cM$-adhesive category $\bfC$ is said to possess an \emph{epi-$\cM$-factorization} if every morphism $f$ of $\bfC$ can be factorized into an epimorphism $e$ and a monomorphism $m\in \cM$ such that $f=m\circ e$, and such that this factorization is unique up to isomorphism.
\end{definition}

It is worthwhile to note that a large class of $\cM$-adhesive categories of practical importance (see e.g.\ the examples listed below) indeed possess an epi-$\cM$-factorization.

\begin{example}
Referring to~\cite[Ex.~2.3ff]{Braatz:2010aa} and~\cite[Section~4.2]{ehrig:2006aa} for further details, well-known examples of $\cM$-adhesive categories with $\cM$-initial objects include:
\begin{itemize}
  \item $(\mathbf{Set},\cM_S)$, the category of \emph{sets and (total) set functions}, with $\cM_S$ the class of all injective set morphisms, and with the $\cM_S$-initial object the empty set $\mIO_S\in \obj{\mathbf{Set}}$.
  \item $(\mathbf{uGraph},\cM_U)$, the category of \emph{undirected multigraphs}~\cite{bp2019-ext} and graph homomorphisms, with $\cM_U$ the class of all injective homomorphisms of $\mathbf{uGraph}$, and with $\cM_U$-initial object the empty graph $\mIO_U\in \obj{\mathbf{uGraph}}$.
  \item $(\mathbf{Graph},\cM_G)$, the category of \emph{directed multigraphs}, with $\cM_G$ the class of all injective graph homomorphisms, and with $\cM_G$-initial object the empty graph ${\mIO}_G\in \obj{\mathbf{Graph}}$. 
  \item $(\mathbf{Graph}_{TG},\cM_{TG})$, the category of \emph{typed graphs} and morphisms thereof (constructed as the slice category $\mathbf{Graph}_{TG}\eqdef\mathbf{Graph}/ TG$ for some fixed type graph $TG\in \mathbf{Graph}$), with $\cM_{TG}$ the class of all injective typed graph homomorphisms.
  \item The categories of \emph{Petri nets} and of \emph{elementary Petri nets}~\cite[Ex.~2.3ff]{Braatz:2010aa} are $\cM$-adhesive and $\cM$-initial for certain classes of $\cM$.
\end{itemize}
All categories in the above list possess an epi-$\cM$-factorization, as do their finitary restrictions. For example, in $\mathbf{Graph}$, every graph homomorphism can be factored into the a surjective composed with an injective graph homomorphism. Interestingly, the well-known \textbf{non-examples}, which are certain categories of (typed or untyped) attributed graphs, fail to possess an $\cM$-initial object and an epi-$\cM$-factorization. There exist a number of \emph{functorial constructions} that allow one to construct finitary $\cM$-adhesive categories with the desired properties from known such categories. We refer the reader to~\cite[Sec.~5]{Braatz:2010aa} for the details.
\end{example}

One of the most important additional properties required in view of compositionality of rewriting rules is the following one:
\begin{definition}[$\cM$-effective unions]\label{def:MeffUn}
  Let $\bfC$ be an $\cM$-adhesive category, for $\cM$ a class of monomorphisms. Then $\bfC$ is said to possess \emph{$\cM$-effective unions} if for every commutative diagram as below where all morphisms except $d$ are in $\cM$, where $(1)$ is a pushout and where the exterior square is a pullback,
  \begin{equation}\label{eq:MEU}
    \inputtikz{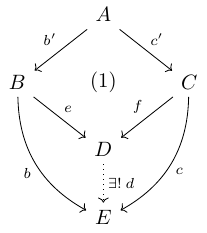}\,,
  \end{equation}
  the morphism $d$ is in $\cM$.
\end{definition}

While \emph{adhesive categories} (which constitute a special case of $\cM$-adhesive categories where $\cM=\mono{\bfC}$) are well known to posses ($\mono{\bfC}$-) effective unions~\cite{lack2005adhesive}, we are not aware of a set of sufficient conditions to ensure the property of $\cM$-effective unions in the general $\cM$-adhesive case. Referring to~\cite{bp2019-ext} for an extended discussion, for some of these more general cases such as for the category $\mathbf{uGraph}$~\cite{BK2020} the following technical result allows one to verify the property:

\begin{theorem}[\cite{bp2019-ext}, Thm.~1.15]\label{thm:RewAux}
  Let $\bfC$ be a \emph{horizontal weak} $\cM$-adhesive category. Then in a diagram of the form~\eqref{eq:MEU}, the morphism $d$ is a monomorphism (however not necessarily in the class $\cM$).
\end{theorem}

Finally, the property of balancedness defined below will be an essential ingredient for our rewriting framework (cf.\ Theorem~\ref{thm:crucialShift}).
\begin{definition}[compare~\cite{lack2005adhesive}, Lem.~4.9] A category is said to be \emph{balanced} if every morphism that is both a mono- and an epimorphism is an isomorphism.
\end{definition}

\subsection{Additional prerequisites for the Sesqui-Pushout (SqPO) framework}\label{sec:SqPOpre}

Referring to~\cite{nbSqPO2019} for a more extensive presentation, we focus here on quoting some necessary background materials, and on discussing the general $\cM$-adhesive setting.

\begin{definition}[Final Pullback Complement (FPC); \cite{Corradini_2006,Loewe_2015}]\label{def:FPC}
Given a commutative diagram of the form
  \begin{equation}\label{eq:FPC}
  \inputtikz{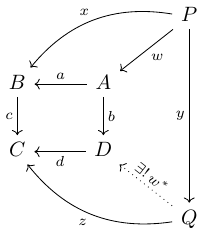}\,,
  \end{equation}
  a pair of morphisms $(d,b)$ is a \emph{final pullback complement (FPC)} of a pair $(c,a)$ if 
  \begin{romanenumerate}
    \item $(a,b)$ is a pullback of $(c,d)$, and 
    \item for each collection of morphisms $(x, y, z, w)$ as in~\eqref{eq:FPC}, where $(x,y)$ is pullback of $(c, z)$ and where $a\circ w=x$, there exists a morphism $w^{*}$ with $d\circ w^{*}=z$ and $w^{*}\circ y=b\circ w$ that is unique (up to isomorphisms).
  \end{romanenumerate}
\end{definition}
 
\begin{lemma}[cf\ \cite{Loewe_2015}, Fact~2, and~\cite{Corradini_2006}, Lemma~2ff]\label{lem:FPCfacts}
  For an arbitrary morphism $f:A\rightarrow B$, $(id_B,f)$ is an FPC of $(f,id_A)$ and vice versa. Moreover, every pushout square is also an FPC square. FPCs are unique up to isomorphism and preserve monomorphisms.
\end{lemma}

If we are working over an \emph{adhesive category} (i.e.\ an $\cM$-adhesive category where $\cM$ coincides with the class of all monomorphisms~\cite{Braatz:2010aa}), the stability of monomorphisms under FPCs as guaranteed by Lemma~\ref{lem:FPCfacts} will be sufficient for our purposes. However, in the more general $\cM$-adhesive setting, we will have to require the following stronger property:

\begin{definition}[Stability of $\cM$-morphisms under FPCs]\label{def:FPCstab}
  Let $\bfC$ be an $\cM$-adhesive category (for $\cM$ a class of monomorphisms). Then $\cM$-morphisms in $\bfC$ are said to be \emph{stable under FPCs} if whenever for a pair of morphisms $(a,b)$ with $a:A\hookrightarrow B$ in $\cM$ and $b:C\rightarrow B$ arbitrary, if $(b',a')$ such that $a\circ b=b'\circ a'$  is the FPC of $(a,b)$, then $b'\in \cM$.
\end{definition}

\subsection{Summary: full set of assumptions for DPO and SqPO rewriting}\label{sec:AssrewritingCats}

Combining all findings of the previous two sections (together with some insights from the constructions presented in the following sections), we present here sets of requirements for associative Double-Pushout (DPO) and Sesqui-Pushout (SqPO) rewriting that admit conditions on both objects and morphisms. As we were primarily motivated by deriving a \emph{sufficient} set of assumptions in order to study rewriting theories in the life-sciences (as in~\cite{BK2020}), we have not attempted to prove that the assumptions provided are strictly \emph{necessary}, albeit the latter point might in itself provide an interesting direction for future work.

\begin{assumption}[Associative DPO rewriting with conditions]\label{as:DPO}
  We assume that $\bfC$ is an \emph{$\cM$-adhesive category} with \emph{epi-$\cM$-factorization}. We also assume that $\bfC$ is \emph{balanced}, possesses a \emph{strict $\cM$-initial object} $\mIO\in \obj{\bfC}$ and \emph{$\cM$-effective unions}.
\end{assumption}

Note that according to~\cite[Lem.~4.9]{lack2005adhesive}, all adhesive categories are balanced, while for the general case the requirement is more non-trivial to demonstrate. An example of the former class which also satisfies all additional assumptions stated above is given by the category $\mathbf{Graph}$ of directed multigraphs~\cite{lack2005adhesive}, while a more general example of an $\cM$-adhesive category satisfying Assumption~\ref{as:DPO} is given by the category  $\mathbf{uGraph}$~\cite{bp2019-ext,BK2020} of undirected multigraphs.

\begin{assumption}[Associative SqPO rewriting with conditions]\label{as:SqPO}
  We assume that $\bfC$ is an \emph{$\cM$-adhesive category} satisfying Assumption~\ref{as:DPO}, and we assume in addition that for all pairs of composable $\cM$-morphisms $A\xhookrightarrow{m}B\xhookrightarrow{n}C$, the final pullback complement exists, and moreover that $\cM$-morphisms are stable under FPCs.
\end{assumption}

According to~\cite{Corradini_2015}, examples of $\cM$-adhesive categories for which the existence of FPCs as required in Assumption~\ref{as:SqPO} is guaranteed are those categories that possess an $\cM$-partial map classifier (cf.\ \cite[Thm.~1 and Sec.~2--5]{Corradini_2015}; compare~\cite{Cockett_2003}). While we refer the interested readers to these references for the full technical details, we mention here that examples of categories possessing an $\cM$-partial map classifier include the $\cM$-adhesive categories $\mathbf{Set}$ and $\mathbf{Graph}$, all presheaf categories, numerous variants of \emph{typed} or \emph{polarized} graphs, and more generally all slice categories $\bfC\downarrow X$ with $\bfC$ a topos and $X\in \obj{\bfC}$. To the best of our knowledge there are no known sufficient conditions to guarantee this stability property, other than a result by J.R.B.~Cockett and S.~Lack~\cite[Prop. 4.16 and Example~4.17]{Cockett_2003}, which however in effect only reaffirms the case of $\bfC$ being an adhesive category. In the general $\cM$-adhesive setting, $\cM$-stability under FPCs will have to be verified at a case-by-case level. Note that the guaranteed existence of final pullback complements in the configurations encountered in SqPO rewriting will drastically simplify the framework, and is in fact necessary to guarantee associativity as discussed in~\cite{nbSqPO2019}.

\section{Conditions on objects and morphisms}\label{sec:cond}

The central concepts of the framework of \emph{conditions} are the following notions of \emph{constraints} (i.e.\ conditions over objects), \emph{application conditions} (i.e.\ conditions over \emph{morphisms}) and the associated notions of satisfiability. We quote the precise definitions from~\cite{ehrig2014mathcal}, and also from~\cite{habel2009correctness}, where some important clarifying details are given (on the notion of satisfiability on objects and morphisms).

\subsection{Core definitions}\label{sec:condCoreDefs}

\begin{definition}[Conditions; 
cf.~\cite{ehrig2014mathcal}, Def.~3.3, and~\cite{habel2009correctness}, Def.~4]\label{def:ac}
(Nested) \emph{conditions} are recursively defined as follows:
\begin{romanenumerate}
\item \emph{Trivial condition:} for every object $P\in \obj{\bfC}$, $\ac{true}$ is a condition over $P$.
\item \emph{``Transported'' conditions:} for every object $P\in \obj{\bfC}$, %
for every\footnote{It is here that our restriction to $\cM$-morphisms in the formulation of conditions reflects our choice of framework, i.e.\ that of $\cM$-satisfiability (for $\cM$-morphisms). We refer the interested readers to~\cite{habel2009correctness} for the proof that this is in fact the most general framework available when working with $\cM$-morphisms in matches and rewriting rules only, i.e.\ generalizing morphisms in conditions of arbitrary morphisms in this setting does not lead to more expressivity.} $\cM$-morphism $(a:P\rightarrow Q)$ and for every condition $\ac{c}_Q$ over $Q$, $\exists(a:P\rightarrow Q,\ac{c}_Q)$ is a condition over $P$.
\item \emph{Negation:} for every condition $\ac{c}_P$ over $P\in \obj{\bfC}$, $\neg\ac{c}_P$ is a condition over $P$.
\item \emph{Conjunction:} given a family of conditions $\{\ac{c}_{P}^{(i)}\}_{i\in I}$ (for some index set $I$) over an object $P\in \obj{\bfC}$, $\land_{i\in I} \ac{c}_{P}^{(i)}$ is a condition over $P$. 
\end{romanenumerate}
The following two \textbf{\emph{shorthand notations}} are customary:
\begin{equation}\label{eq:logicShorthandConds}
\exists a \eqdef \exists (a,\ac{true})\,,
\quad \forall(a,\ac{c})\eqdef \neg \exists(a,\neg \ac{c})\,.
\end{equation}
The precise meaning of the above definitions is specified via the associated notions of \emph{satisfiability}, which are also defined inductively:
\begin{itemize}
\item[\textbf{S1}] Every $\cM$-morphism $p:P\rightarrow P'$ satisfies the trivial condition $\ac{true}$.
\item[\textbf{S2}] Given $\cM$-morphisms $p:P\rightarrow P'$ and $a:P\rightarrow Q$ as well as a condition $\exists(a,\ac{c}_Q)$ over $P$, the morphism $p$ is defined to satisfy the condition $\exists(a,\ac{c}_Q)$ if and only if there exists an $\cM$-morphism $q:Q\rightarrow P'$ such that $q\circ a=p$ and such that $q$ satisfies the condition $\ac{c}_Q$,
\begin{equation}\label{eq:acMorphSat}
\inputtikz{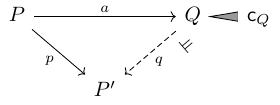}\,.
\end{equation}
\item[\textbf{S3}] Given an $\cM$-morphism $p:P\rightarrow P'$ and a condition $\ac{c}_P$ over $P$, $p$ satisfies $\neg\ac{c}_P$ if it does not satisfy $\ac{c}_P$. If $\{\ac{c}_{P}^{(i)}\}_{i\in I}$ (for some indexing set $I$) is a family of conditions over $P$, $p$ satisfies $\land_{i\in I}\ac{c}_{P}^{(i)}$ if it satisfies each of the application conditions $\ac{c}_{P_i}$.
\end{itemize}
For an $\cM$-morphism $p:P\rightarrow P'$, we write 
\[
  p\vDash \ac{c}_P
\]
to denote that \emph{$p$ satisfies the condition $\ac{c}_P$}. Two application conditions $\ac{c}_P,\ac{c}_P'$ over some object $P\in \obj{\bfC}$ are \textbf{\emph{equivalent}}, denoted $\ac{c}_P\equiv\ac{c}_P'$, if and only if for all $\cM$-morphisms $p:P\rightarrow H$, i.e.\ for arbitrary $H\in \obj{\bfC}$ with $P$ as an $\cM$-subobject, we find that
\begin{equation}\label{eq:acEqeuiv}
p\vDash\ac{c}_P\iff p\vDash \ac{c}_P'\,.
\end{equation}
Finally, a (nested) condition or \textbf{constraint} on an object $P\in \obj{\bfC}$ is defined as a (nested) condition over the $\cM$-initial object $\mIO$, and the associated notion of satisfiability of constraints as satisfaction of the condition over $\mIO$ by the $\cM$-initial morphism $(i_P:\mIO\hookrightarrow P)\in \cM$~\cite{habel2009correctness}:
\begin{itemize}
\item[\textbf{C1}] Every object $P\in \obj{\bfC}$ satisfies $\ac{true}$.
\item[\textbf{C2}] An object $P\in \obj{\bfC}$ satisfies the condition  $\exists(i_Q:\mIO\rightarrow Q,\ac{c}_Q)$ if there exists an $\cM$-morphism $q:Q\rightarrow P$ such that $q\circ i_Q =i_P$ and $q\vDash \ac{c}_Q$.
\item[\textbf{C3}] Given an object $P\in \obj{\bfC}$ and a condition $\ac{c}_{\mIO}$ over the $\cM$-initial object $\mIO$, $P$ satisfies $\neg\ac{c}_{\mIO}$ if it does not satisfy $\ac{c}_{\mIO}$. 
\item[\textbf{C4}]  If $\{\ac{c}_{\mIO}^{(i)}\}_{i\in I}$ (for some indexing set $I$) is a family of conditions over $\mIO$, $P$ satisfies $\land_{i\in I}\ac{c}_{\mIO}^{(i)}$ if it satisfies each of the conditions $\ac{c}_{\mIO}^{(i)}$.
\end{itemize}
\end{definition}

It may be instructive to the readers to explicitly parse the potentially somewhat counter-intuitive definition of conditions on objects in a concrete example:

\begin{example}
Given an object $P\in \obj{\bfC}$, a condition of the form ``$P$ contains an $\cM$-subobject $Q$'' is expressed in the present framework as $P\vDash \exists(i_Q)=\exists(i_Q:\varnothing \rightarrow Q,\ac{true})$, since by virtue of the definition of satisfiability of conditions on objects, $P\vDash \exists(i_Q)$ iff there exists an $\cM$-morphism $q:Q\rightarrow P$ such that $q\circ i_Q=i_P$: 
\begin{equation}
\inputtikz{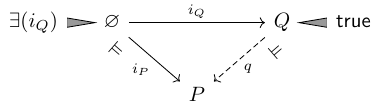}
\end{equation}
\end{example}

Next, consider the following example for illustration of the concept of nested conditions. We will, in practice, be interested exclusively in \emph{finite} nested  conditions, i.e. in sequences (or in general directed acyclic graphs) of conditions that ultimately end in an instance of a condition of the form $\exists(x,\ac{true})$. In this sense, the example below is sufficiently generic.
\begin{example}
The condition below (on undirected multigraphs)
\begin{equation}
  \exists\left(a:\inputtikzB{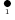}{-0.19}\rightarrow 
        \inputtikzB{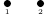}{-0.17},
  \exists\left(b:
  \inputtikzB{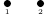}{-0.17}\rightarrow
  \inputtikzB{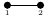}{-0.17},\ac{true}\right)
    \right)
\end{equation}
parses more explicitly into the diagram
\begin{equation}
\inputtikz{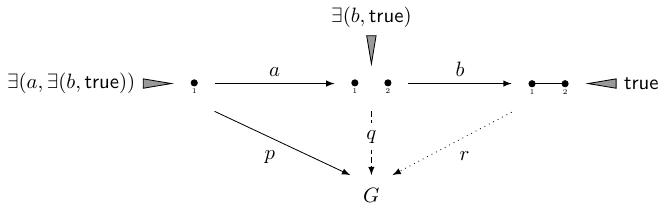}
\end{equation}
expressing the condition that a morphism $p:\inputtikz{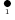}\hookrightarrow G$ satisfies the condition if $G$ contains at least one other vertex $\inputtikz{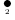}$ (which is the information encoded in the first part of the condition), and such that $\inputtikz{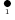}$ and $\inputtikz{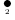}$ are linked by an edge. Moreover, the $\cM$-morphism $q$ in the above diagram automatically exists if the entire condition is satisfied. This is in fact a typical example of \emph{refinement} (or \emph{$\cM$-coverability}~\cite{habel2009correctness}), whereby the condition $\exists(a,\ac{true})$ is refined by the condition $\exists(a,\exists(b,\ac{true}))$.
\end{example}

\subsection{A refined notion of shift construction}\label{sec:shift}

One of the key concepts in the theory of rewriting with conditions is the notion of \emph{shift construction}, which is, in essence, a category-theoretical characterization of the interplay between conditions and extensions of their domains. We introduce here an optimized version of the classical shift construction presented in~\cite{ehrig2014mathcal}) (which is itself a variant of an earlier construction as reviewed in~\cite{habel2009correctness}). Our optimization hinges on the assumed properties of the underlying $\cM$-adhesive categories according to Assumption~\ref{as:DPO}. We believe this optimization will be of key importance in future developments of algorithms and software implementations of our framework. The following theorem is at the basis of our novel construction:

\begin{theorem}\label{thm:crucialShift}
  Given an $\cM$-adhesive category $\bfC$ satisfying Assumption~\ref{as:DPO}, consider a commutative diagram of the form below, 
  \begin{equation}\label{eq:crucialShiftDiag}
  \inputtikz{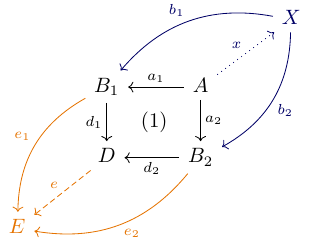}
  \end{equation}
  where the square marked $(1)$ is a pushout, where $a_1,a_2\in \cM$ (and thus by stability of $\cM$-morphisms under pushout also $d_1,d_2\in \cM$), ${\color{h2color}e_1},{\color{h2color}e_2}\in \cM$ and ${\color{h1color}X}=\pB{B_1{\color{h2color}\xrightarrow{e_1}E\xleftarrow{e_2}}B_2}$ (and thus by stability of $\cM$-morphisms under pullbacks also ${\color{h1color}b_1},{\color{h1color}b_2}\in \cM$, plus due to the decomposition property of $\cM$ morphisms, also ${\color{h1color}x}\in \cM$). Then the following holds: the morphism ${\color{h2color}e}$ is an epimorphism if and only if the exterior square is a pushout.
\end{theorem}
\begin{proof}
\begin{figure}[htp]
\begin{equation}\label{eq:thm3proofDiag}
\vcenter{\hbox{\includegraphics[scale=1]{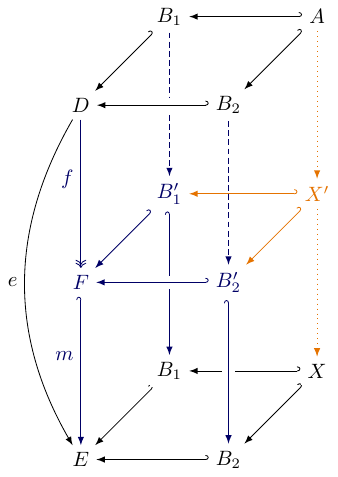}}}
\qquad \qquad 
\vcenter{\hbox{\includegraphics[scale=1]{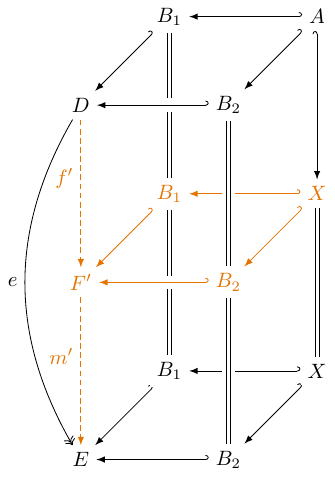}}}
\end{equation}
{\caption*{}}
\end{figure}%
\textbf{``$\Rightarrow$'' direction:} Suppose that the exterior square in~\eqref{eq:crucialShiftDiag} is a pushout. %
Let $e={\color{h1color}m}\circ {\color{h1color}f}$ be the epi-$\cM$-factorization of $e$ (with ${\color{h1color}(f:D\rightarrow F)}\in \epi{\bfC}$ and ${\color{h1color}(m:F\rightarrow E)}\in \cM$). %
Then construct the commutative diagram presented in the left part of~\eqref{eq:thm3proofDiag} as follows: %
form the two pullbacks ${\color{h1color}B_i'}=\pB{{\color{h1color}F\rightarrow} E\leftarrow B_i}$, %
which by the universal property of pullbacks also induces morphisms ${\color{h1color}B_i\rightarrow B_i'}$ (for $i=1,2$); %
according to Lemma~\ref{lem:pbEpiMonoId} of Appendix~\ref{app:epiMcons}, we conclude that ${\color{h1color}B_1'}\cong B_1$ and ${\color{h1color}B_2'}\cong B_2$. %
Next, let ${\color{h2color}X'}=\pB{{\color{h1color}B_1'\rightarrow F\leftarrow B_2'}}$; %
by the universal property of pullbacks, this implies the existence of morphisms $A{\color{h2color}\rightarrow X'}$ %
and ${\color{h2color}X'\rightarrow} X$ (and by decomposition of $\cM$-morphisms, the latter is in $\cM$). %
Since the square just constructed is a pullback, the bottom square a pushout along $\cM$-morphisms (and thus a pullback), %
and since, due to ${\color{h1color}B_i'}\cong B_i$ (for $i=1,2$) and ${\color{h1color}m}\in \cM$, %
Lemma~\ref{lem:Main}\ref{lem:idPB} entails that the bottom left and bottom front vertical squares are pullbacks, %
we conclude via invoking pullback-pullback decomposition (Lemma~\ref{lem:Main}\ref{lem:PBPBdec}) that also the bottom back and bottom right vertical squares are pullbacks. %
Thus by virtue of stability of isomorphisms under pullbacks (Lemma~\ref{lem:Main}\ref{lem:PBiso}), ${\color{h2color}X'}\cong X$. %
Moreover, the $\cM$-van Kampen property entails that the middle horizontal square is a pushout, %
whence by uniqueness of pushouts up to isomorphism, we have that ${\color{h1color}F}\cong E$, %
which proves the claim that $e={\color{h1color}m}\circ {\color{h1color}f}\in \epi{\bfC}$. 

\paragraph{``$\Leftarrow$'' direction:} %
Suppose that $e\in \epi{\bfC}$ and that the exterior square in~\eqref{eq:crucialShiftDiag} is a pullback. %
Construct the commutative diagram depicted in the right part of~\eqref{eq:thm3proofDiag} as follows: %
start by forming the pushout ${\color{h2color}F'}=\pO{B_1\leftarrow X\rightarrow B_2}$, %
which by the universal property of pushouts (recalling that the top square is by assumption also a pushout) %
furnishes morphisms ${\color{h2color}f'}:D{\color{h2color}\rightarrow F'}$ and %
${\color{h2color}m'}:{\color{h2color}F'\rightarrow} E$ such that $e={\color{h2color}m'}\circ {\color{h2color}f'}$. %
By stability of $\cM$-morphisms under pushouts, the morphisms ${\color{h2color}B_1\rightarrow F'}$ and  %
${\color{h2color}B_2\rightarrow F'}$ are in $\cM$. %
The top commutative cube thus precisely satisfies the properties necessary to invoke the theorem in the ``$\Rightarrow$'' direction, in order to conclude that ${\color{h2color}f'}\in\epi{\bfC}$. %
Since by assumption $e\in \epi{\bfC}$, and since $e={\color{h2color}m'}\circ {\color{h2color}f'}$, %
invoking decomposition of epimorphisms yields that ${\color{h2color}m'}\in\epi{\bfC}$. %
On the other hand, since by assumption the bottom square is a pullback, %
invoking the property of $\cM$-effective unions permits us to conclude that ${\color{h2color}m'}\in \cM$. %
As the underlying category is assumed to be balanced, and since $\cM\subseteq\mono{\bfC}$, %
we conclude that ${\color{h2color}m'}\in \iso{\bfC}$, i.e.\ ${\color{h2color}F'}\cong E$, %
which proves that the bottom square is a pushout (by uniqueness of pushouts, and since pushouts along $\cM$-morphisms are also pullbacks).
\end{proof}

An interesting consequence of the above theorem is the following result, which allows one to compare our technical framework more directly with the traditional literature:
\begin{corollary}\label{cor:jE}
In an $\cM$-adhesive category $\bfC$ satisfying Assumption~\ref{as:DPO}, given two $\cM$-morphisms $(c_1:C_1\hookrightarrow D),(c_2:C_2\hookrightarrow D)\in \cM$, let the $\cM$-morphisms 
\[
	{\color{h1color}(c_1':C_1\hookrightarrow \bar{D})}\,,\; {\color{h1color}(c_2':C_2\hookrightarrow \bar{D})}\,,\; {\color{h1color}(d:\bar{D}\hookrightarrow D)}\in \cM
\]
be defined (uniquely up to isomorphisms) as follows: denoting by ${\color{h1color}(C_1\hookleftarrow X\hookrightarrow C_2)}$ the pullback of the cospan $(c_1,c_2)$, define ${\color{h1color}(c_1',c_2')}$ as the pushout of this span, and let ${\color{h1color}d}$ be defined as the induced morphism from ${\color{h1color}\bar{D}}$ into $D$. Then the cospan ${\color{h1color}(c_1',c_2')}$ is \textbf{\emph{jointly epimorphic}}.
\end{corollary}
\begin{proof}
Given that $c_1,c_2,d\in \cM$ (which follows from stability of $\cM$-morphisms under pullback and pushout, and from the property of $\cM$-effective unions, respectively), construct the following commutative diagram:
\begin{equation}
\inputtikz{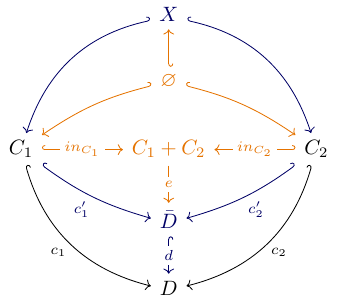}
\end{equation}
Here, given the unique $\cM$-morphisms ${\color{h2color}(\mIO\hookrightarrow C_1)}$,  ${\color{h2color}(\mIO\hookrightarrow C_2)}$ and  ${\color{h2color}(\mIO\hookrightarrow X)}$ from the $\cM$-initial object $\mIO$, and with ${\color{h2color}(C_1\hookrightarrow C_1+C_2\hookleftarrow C_2)}$ the pushout of ${\color{h2color}(C_1\hookleftarrow \mIO\hookrightarrow C_2)}$, the existence of the morphism ${\color{h2color}(e:C_1+C_2\rightarrow \bar{D})}$ follows by the universal property of pushouts. Then it follows via invoking Theorem~\ref{thm:crucialShift} that ${\color{h2color}e}$ is an epimorphism.
\end{proof}

We will take advantage of this corollary when discussing the results of Theorem~\ref{thm:Shift} below as well as the notion of DPO-type rule compositions in its various forms in Section~\ref{sec:DPOR}. Theorem~\ref{thm:Shift} below constitutes another interesting consequence of Theorem~\ref{thm:crucialShift}, deriving an algorithmic refinement of the so-called $\Shift$ construction, which allows one to extend application conditions to larger contexts:
\begin{theorem}[Shift construction; compare~\cite{habel2009correctness}, Thm.~5 and Lem.~3, \cite{ehrig2014mathcal} Lem.~3.11]\label{thm:Shift}
Given an $\cM$-adhesive category $\bfC$ satisfying Assumption~\ref{as:DPO}, there exists a \emph{shift construction}, denoted $\Shift$, %
such that for every condition $\ac{c}_P$ over an object $P\in \obj{P}$ and for every $\cM$-morphism $p:P\rightarrow Q$, %
an $\cM$-morphism ${\color{h3color}n}\circ p: P {\color{h4color}\rightarrow H}$ (with ${\color{h3color}n}\in \cM$) satisfies the condition $\ac{c}_P$ %
iff $({\color{h3color}n}:Q{\color{h3color}\rightarrow}{\color{h4color}H})$ satisfies $\ac{Shift}(p,\ac{c}_P)$, referred to as the \emph{shift of $\ac{c}_P$ along $p$}:
\begin{subequations}\label{eq:shiftCondequiv}
\begin{align}
&{\color{h3color}n}\circ p\vDash \ac{c}_P\;\Leftrightarrow\; {\color{h3color}n}\vDash\Shift(p,\ac{c}_P)\,,\\
\intertext{with}
&
\inputtikz{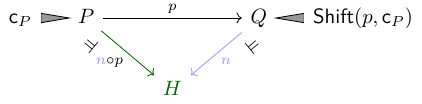}\,.
\end{align}
\end{subequations}
Here, the application condition $\Shift(p,\ac{c}_P)$ is constructed inductively as follows:
\begin{romanenumerate}
\item \emph{Case $\ac{c}_P=\ac{true}$:} 
\begin{equation}
\Shift(p,\ac{true})\eqdef\ac{true}\,.
\end{equation}
\item \emph{Case $\ac{c}_P=\exists({\color{blue}a},{\color{blue}\ac{c}_A})$ (for some ${\color{blue}(a:P\rightarrow A)}\in \cM$ and ${\color{blue}\ac{c}_A}$ an application condition over ${\color{blue}A}\in \obj{\bfC}$):} construct the commutative diagram below, where the square marked $\mathsf{PO}$ is a pushout\footnote{Besides a more concrete characterization of $\Shift$, we also profit from restricting all morphisms involved in this construction (except for the epimorphisms $e$) to be $\cM$-morphisms, which in particular guarantees due to the assumed $\cM$-adhesivity of the underlying category $\bfC$ that the pushout to form the object $S$ always exists.}: 
 \begin{equation}\label{eq:ShiftDiagram}
 \inputtikz{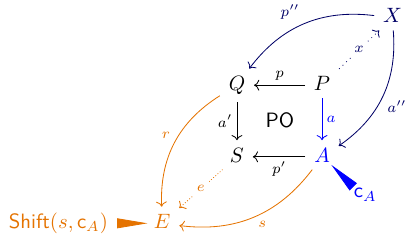}
  \end{equation}
Here, each $\cM$-morphism ${\color{h1color}x:P\rightarrow X}$ such that there exist $\cM$-morphisms ${\color{h1color}p'':X\rightarrow Q}$ and %
${\color{h1color}a'':X\rightarrow{A}}$, and with ${\color{h1color}p''}\circ {\color{h1color}x}=p$ and ${\color{h1color}a''}\circ {\color{h1color}x}={\color{blue}a}$, %
induces an object ${\color{h2color}E}$ and $\cM$-morphisms ${\color{h2color}r:Q\rightarrow E}$ and ${\color{h2color}s:A\rightarrow{E}}$ via taking the pushout of the span $(Q{\color{h1color}\xleftarrow{p''}X\xrightarrow{a''}}A)$. Since objects in $\bfC$ are assumed to be finite, which entails in particular that there are only finitely many $\cM$-subobjects of $Q$ and ${\color{blue}A}$, up to isomorphisms of spans %
(induced by isomorphisms of ${\color{h1color}X}$) there are only finitely many isomorphism classes of spans %
$(Q{\color{h1color}\xleftarrow{p''}X\xrightarrow{a''}}{\color{blue}A})$. Denoting the set of all isomorphism classes of $\cM$-morphism pairs ${\color{h2color}(r,s)}$ thus obtained by $\cE$, we define
\begin{equation}\label{eq:shiftConstruction}
  \ac{Shift}(p,\exists({\color{blue}a},{\color{blue}\ac{c}_A}))\eqdef
  \bigvee\limits_{{\color{h2color}(r,s)}\in \cE} \exists({\color{h2color}r},\ac{Shift}({\color{h2color}s},{\color{blue}\ac{c}_A}))\,.
\end{equation}
\item \emph{Case $\ac{c}_P=\neg\ac{c}_P'$:}
\begin{equation}
\Shift(b,\neg \ac{c}_P')\eqdef\neg \Shift(b,\ac{c}_P')\,.
\end{equation}
\item \emph{Case $\land_{i\in I}\ac{c}_i$:}
\begin{equation}
  \Shift(b,\land_{i\in I}\ac{c}_i)\eqdef\land_{i\in I}\Shift(b,\ac{c}_i)
\end{equation}
\end{romanenumerate}
\end{theorem}
\begin{proof}
We will closely follow the proof strategy of~\cite[Lem.~3.11]{ehrig2014mathcal}, adapted to our variant of the $\Shift$ construction (and proving en passent the equivalence of our construction to the ones of  \cite[Thm.~5 and Lem.~3]{habel2009correctness} and \cite[Lem.~3.11]{ehrig2014mathcal}, see Remark~\ref{rem:Shift} below). The statement is trivially true for the case $\ac{c}_P=\ac{true}$. For the case $\ac{c}_P=\exists({\color{blue}a:P\hookrightarrow A}, {\color{blue}\ac{c}_A})$ of a nested application condition, we will prove the claim by induction over the levels of nesting.

\begin{equation}\label{eq:shiftCondThmProof}
\vcenter{\hbox{\includegraphics{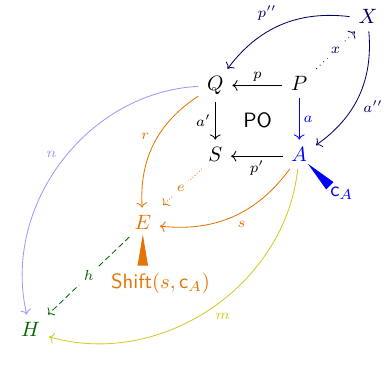}}}
\end{equation}

\paragraph{``$\Rightarrow$'' direction:} Suppose that ${\color{h3color}n}\circ p\vDash \exists({\color{blue}a:P\hookrightarrow A}, {\color{blue}\ac{c}_A})$, and take~\eqref{eq:shiftConstruction} as the \emph{induction hypothesis} in case that ${\color{blue}\ac{c}_A}$ is itself a nested condition. %
By definition of satisfiability, this entails that there exists ${\color{h5color}(m:A\hookrightarrow H)}\in \cM$ such that ${\color{h5color}m}\circ {\color{blue}a} = {\color{h3color}n}\circ p$. %
Construct the span ${\color{h1color}(p'',a'')}$ by taking the pullback of the cospan $({\color{h3color}n},{\color{h5color}m})$, %
which by universal property of pullbacks furnishes a morphism ${\color{h1color}x:P\rightarrow X}$. %
By stability of $\cM$-morphisms under pullbacks, we find that ${\color{h1color}p''},{\color{h1color}a''}\in \cM$, and thus %
by the decomposition property of $\cM$-morphisms also that ${\color{h1color}x}\in\cM$. %
Next, construct the cospan ${\color{h2color}(r,s)}$ via taking the pushout of the span ${\color{h1color}(p'',a'')}$, which %
by stability of $\cM$-morphisms under pushout entails that ${\color{h2color}r},{\color{h2color}s}\in \cM$. %
By the universal property of pushouts, there exist morphisms ${\color{h2color}e:S\rightarrow E}$ and $h:E\rightarrow H$ (see~\eqref{eq:shiftCondThmProof}), with %
${\color{h2color}e}\in \epi{\bfC}$ (via Theorem~\ref{thm:crucialShift}) and $h\in \cM$ (via $\cM$-effective unions). %
By definition of satisfiability, ${\color{h4color}m}\vDash {\color{blue}\ac{c}_A}$ and ${\color{h4color}m}= h\circ{\color{h2color}s}$ together with the induction hypothesis for $\ac{c}_A$ entail that $h\vDash \Shift({\color{h2color}s},\ac{c}_A)$. %
Again invoking the definition of satisfiability, since $h\circ {\color{h2color}r} ={\color{h3color}n}$ and $h\vDash \Shift({\color{h2color}s},{\color{blue}\ac{c}_A})$, we have thus verified that indeed ${\color{h3color}n}\vDash \exists ({\color{h2color}r}, \Shift({\color{h2color}s},{\color{blue}\ac{c}_A})$.

\paragraph{``$\Leftarrow$'' direction:} Let us assume that ${\color{h3color}n}\vDash \Shift(p,\exists({\color{blue}a},{\color{blue}\ac{c}_A}))$, with the shifted condition constructed according to~\eqref{eq:shiftConstruction}. We have to verify that this entails $n\circ p\vDash \exists({\color{blue}a},{\color{blue}\ac{c}_A})$. Combining the assumption with the definition of satisfiability of conditions, there must exist ${\color{h2color}r},{\color{h2color}s},h\in \cM$ such that $h\circ{\color{h1color}r}={\color{h3color}n}$ and $h\vDash \Shift({\color{h2color}s},{\color{blue}\ac{c}_A})$. By the induction assumption, $h\vDash \Shift({\color{h2color}s},{\color{blue}\ac{c}_A})$ entails that the morphism ${\color{h5color}m}\in \cM$ given by ${\color{h5color}m}:= h\circ {\color{h2color}s}$ (see~\eqref{eq:shiftCondThmProof}) satisfies ${\color{h3color}m}\vDash\ac{c}_A$. Since by construction ${\color{h5color}m}\circ a={\color{h3color}n}\circ p$, by definition of satisfiability we have thus demonstrated that indeed ${\color{h3color}n}\circ p\vDash \exists({\color{blue}a},{\color{blue}\ac{c}_A})$.

Finally, the proof of statements $(iii)$ and $(iv)$ is obtained in an analogous fashion.
\end{proof}

\begin{remark}\label{rem:Shift}
Since an algorithmically tractable $\Shift$ construction is quintessential for developing a calculus of compositional rewriting, let us briefly discuss the precise relationship of our refined construction to the pre-existing constructions in the literature, highlighting in particular the nature of the refinement:
\begin{itemize}
\item In the variant of the $\Shift$ construction and corresponding proof strategy presented in~\cite[Lem.~3.11]{ehrig2014mathcal}, the authors base their construction upon assuming \emph{a priori} that a set of cospans
\[
	\cF = \{ (a',p')\in \cE'\mid p'\in \cM\;\land\;  \cSquare{P,Q,{\color{h2color}E},{\color{blue}A}}\text{ commutes}\}
\]
is provided, where no additional properties for the commutative squares $\cSquare{P,Q,{\color{h2color}E},{\color{blue}A}}$ are required. As for the class of cospans $\cE'$, the origin of this class in~\cite{ehrig2014mathcal} is an \emph{$\cE'$-$\cM$-pair factorization} property that is assumed to hold for the category $\bfC$, whereby any cospan $(f_1,f_2)$ of arbitrary morphisms $(f_i:C_i\rightarrow D)\in \mor{\bfC}$ ($i=1,2$) factors uniquely up to isomorphisms as
\[
	f_1=m\circ e_1\,,\; f_2=m\circ e_2\,,\; m\in \cM\,,\; (e_1,e_2)\in \cE'\,.
\]
On the one hand, this construction permits to consider the seemingly more general case of conditions formulated in terms of non-$\cM$-morphisms, which however according to highly technical results presented in~\cite{habel2009correctness} in fact does \emph{not} increase the expressivity of the calculus of conditions. On the other hand, the construction of Ehrig et al.\ is strictly more general than ours in that it permits to consider \emph{non-monic matches} for applications of linear rules. %
Yet for an $\cM$-adhesive category $\bfC$ satisfying Assumption~\ref{as:DPO} and for the case of $\cM$-morphisms as admissible matches, Corollary~\ref{cor:jE} entails that our construction is mathematically equivalent to the one of~\cite[Lem.~3.11]{ehrig2014mathcal}, since in this case the relevant $\cE'$-$\cM$-pair-factorizations of cospans of $\cM$-morphisms are precisely of the form presented in the corollary. %
From an algorithmic standpoint, the class $\cE'$ is typically non-trivial to construct, whereas our construction based upon the results of Theorem~\ref{thm:crucialShift} only requires algorithmically rather concrete notions of ``forming all larger overlaps'' $(Q{\color{h1color}\hookleftarrow X\hookrightarrow}{\color{blue}A})$ starting from a given ``overlap'' $(Q\hookleftarrow \!P\,{\color{blue} \hookrightarrow A})$. In summary, for categories satisfying Assumption~\ref{as:DPO}, our construction thus constitutes a certain algorithmic refinement of the ``traditional'' Ehrig et al.\ construction. 
\item The construction of~\cite[Thm.~5 and Lem.~3]{habel2009correctness} was developed for the setting of $\cM$-morphisms as admissible matches as in our case. This variant of a $\Shift$ construction is based upon a shape of commutative diagram as in~\eqref{eq:shiftCondThmProof} without the subdiagrams involving the object ${\color{h1color}X}$ and morphisms incident to it. Instead, the authors devise an algorithm whereby one iterates over the triples of morphisms $({\color{h2color}r},{\color{h2color}s},{\color{h2color}e})$ such that  ${\color{h2color}r},{\color{h2color}s}\in\cM$, ${\color{h2color}e}\in\epi{\bfC}$ and such that the diagram commutes. Inspecting the proof of Theorem~\ref{thm:Shift}, our refined construction precisely mirrors this algorithmic idea, yet it provides by virtue of Theorem~\ref{thm:crucialShift} also a concrete construction principle for such triples of morphisms (which was not available in~\cite{habel2009correctness}).
\item Finally, it is worthwhile noting that the complexity of the $\Shift$ construction may be significantly reduced\footnote{We would like to thank one of the anonymous reviewers for drawing our attention to this result.} in an alternative setting within which conditions as well as satisfiability are defined not in terms of $\cM$-morphisms, but in terms of arbitrary morphisms~\cite[Lem~.1]{navarro2016logic}. However, this generalized setting is unfortunately not the setting relevant to the type of compositional rewriting theories considered in the present paper (i.e.\ those that satisfy both concurrency and associativity properties), which is why we indeed must necessarily rely upon the $\cM$-morphism variant of the $\Shift$ construction as presented in Theorem~\ref{thm:Shift}.
\end{itemize}
\end{remark}

As a first application of our refined $\Shift$ construction, let us consider a special situation for shifts that plays a role later on in the theory of compositions of rewriting rules:
\begin{lemma}[$\Shift$ along coproduct injections]\label{lem:shiftCop}
Let $\bfC$ be an $\cM$-adhesive category satisfying Assumption~\ref{as:DPO}. Let $P,Q\in \obj{\bfC}$ be objects, and %
let $\exists({\color{blue}a:P\rightarrow A},{\color{blue}\ac{c}_A})$ be a condition over $P$. Then $\Shift(P\rightarrow P+Q,\exists({\color{blue}a:P\rightarrow A},{\color{blue}\ac{c}_A}))$ is computed via the following type of diagram:
  \begin{equation}\label{eq:ShiftDiagramCop}
  \inputtikz{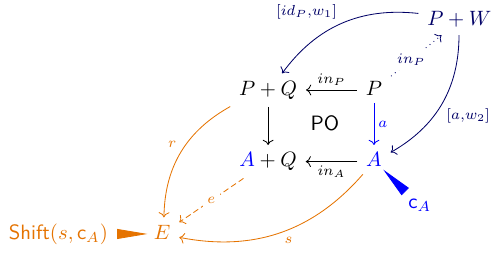}
  \end{equation}
\end{lemma}
\begin{proof}
Consider the following specialization of Lemma~\ref{lem:MintoCop}: starting from the diagram depicted in~\eqref{eq:crucialShiftDiag}, any $\cM$-morphism $f:X\rightarrow P+Q$ decomposes into the form $[f_p,f_q]:X_P+X_Q\rightarrow P+Q$ (where $X_P+X_Q\cong X$ and $f_p,f_q\in \cM$), and analogously the morphism $x:P\rightarrow X$ decomposes into the pair of $\cM$-morphisms $[g,h]:P'+P''\rightarrow X_P+X_Q$ (with $P'+P''\cong P$). But since by commutativity of the diagram in~\eqref{eq:ShiftDiagramCop} $[f_p,f_q]\circ [g,h]=in_P$, it follows that $P'\cong P$ and $P''\cong \mIO$, which upon invoking Theorem~\ref{thm:Shift} proves the claim.
\end{proof}

\begin{example} As an illustrative application of Lemma~\ref{lem:shiftCop}, consider the following explicit computation in the category $\mathbf{uGraph}$ of undirected multigraphs (which satisfies Assumption~\ref{as:DPO}, with $\mIO$ the empty graph):
\begin{equation}
\vcenter{\hbox{\includegraphics[scale=0.7,page=1]{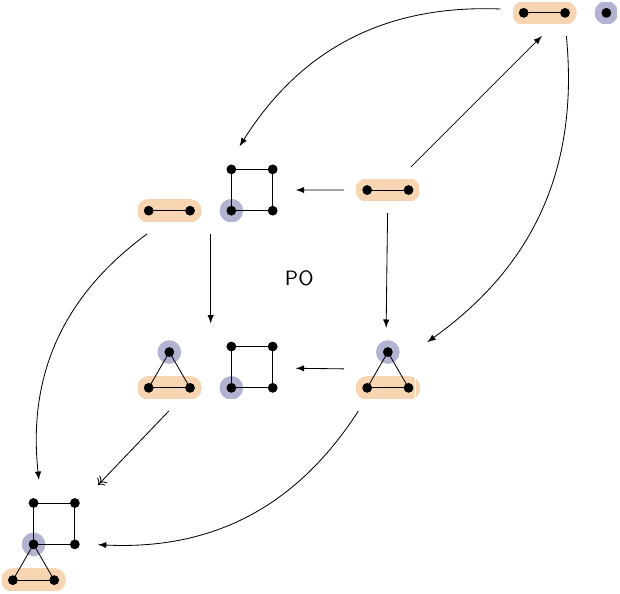}}}
\end{equation}
The image demonstrates how the shift along the embedding of the two-vertex graph with an edge (highlighted in {\color{h2color!60!white}orange}) into a disjoint union with a ``square'' graph yields a condition over this disjoint union of graphs that tests for a disjoint pattern, but also (via the only other possible non-trivial overlap up to isomorphisms, along an additional disjoint vertex marked in {\color{h1color!60!white}blue}) an alternative condition that tests for a non-disjoint pattern:
\begin{equation}\gdef\siScale{0.5}
\begin{aligned}
&\Shift\left(
\vcenter{\hbox{\includegraphics[scale=\siScale,page=5]{images/shiftOplusIllustr.pdf}}}\rightarrow \vcenter{\hbox{\includegraphics[scale=\siScale,page=4]{images/shiftOplusIllustr.pdf}}},
\exists\left(
\vcenter{\hbox{\includegraphics[scale=\siScale,page=5]{images/shiftOplusIllustr.pdf}}}\rightarrow \vcenter{\hbox{\includegraphics[scale=\siScale,page=6]{images/shiftOplusIllustr.pdf}}},\ac{true}
\right)
\right)\\
&\qquad =\quad  \exists\left(
\vcenter{\hbox{\includegraphics[scale=\siScale,page=4]{images/shiftOplusIllustr.pdf}}}\rightarrow \vcenter{\hbox{\includegraphics[scale=\siScale,page=3]{images/shiftOplusIllustr.pdf}}},\ac{true}
\right)\quad 
\bigvee \quad \exists\left(
\vcenter{\hbox{\includegraphics[scale=\siScale,page=4]{images/shiftOplusIllustr.pdf}}}\rightarrow \vcenter{\hbox{\includegraphics[scale=\siScale,page=2]{images/shiftOplusIllustr.pdf}}},\ac{true}
\right)\,.
\end{aligned}
\end{equation}
\end{example}

Finally, we will require two additional technical lemmas of key importance to our framework of compositional rewriting. Both results rely on the notion of \emph{equivalences} of conditions (see Definition~\ref{def:ac} and equation~\eqref{eq:acEqeuiv}).

\begin{lemma}[Units for Shift]\label{lem:shiftUnit}
  For every object $P\in \obj{\bfC}$ and for every condition $\ac{c}_P$ over $P$, we have that
  \begin{equation}
    \Shift(id_P:P\rightarrow P,\ac{c}_P)\equiv \ac{c}_P\,.
  \end{equation}
\end{lemma}
\begin{proof}
  This follows directly from the definition of $\Shift$ according to Theorem~\ref{thm:Shift}, by specializing~\eqref{eq:shiftCondequiv} in the form
\begin{subequations}
\begin{align}
&n\circ id_P=n\vDash \ac{c}_P\;\Leftrightarrow\; n\vDash\Shift(id_P,\ac{c}_P)\,,\\
\intertext{with}
&\inputtikz{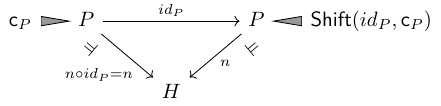}\,.
\end{align}
\end{subequations}
\end{proof}

\begin{lemma}[Compositionality of Shift; compare~\cite{GOLAS2014}, Fact~3.14]\label{lem:shiftComp}
Let $X\in \obj{\bfC}$ be an object, $\ac{c}_X$ an application condition over $X$, and let $f:X\rightarrow Y$ and $g:Y\rightarrow Z$ be two morphisms of $\bfC$. Then the following equivalence of conditions holds:
\begin{equation}
\Shift(g,\Shift(f,\ac{c}_X))\equiv \Shift(g\circ f,\ac{c}_X)
\end{equation}
\end{lemma}
\begin{proof}
  The proof follows by a repeated application of Theorem~\ref{thm:Shift}. The equivalence holds if for all morphisms $c:Z\rightarrow H$, we find that
\begin{equation}\label{eq:ShiftComp}
  c\vDash \Shift(g,\Shift(f,\ac{c}_X))\;\Leftrightarrow\;
  c\vDash \Shift(g\circ f,\ac{c}_X)\,.
\end{equation}
Starting from the diagram below,
\begin{equation}
\inputtikz{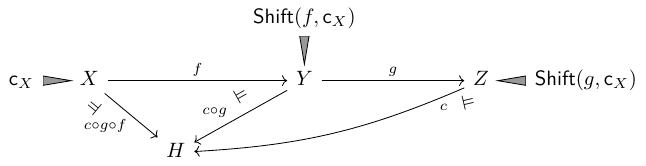}\,,
\end{equation}
we may calculate:
\begin{equation}
\begin{aligned}
  &&c&\vDash \Shift(g,\Shift(f,\ac{c}_X))\\
  &\overset{\eqref{eq:shiftCondequiv}}{\iff}&\quad 
  c\circ g&\vDash \Shift(f,\ac{c}_X)\\
  &\overset{\eqref{eq:shiftCondequiv}}{\iff}&\quad
  c\circ g\circ f &\vDash \ac{c}_X\\
  &\overset{\eqref{eq:shiftCondequiv}}{\iff}&\quad
  c&\vDash \Shift(g\circ f,\ac{c}_X)\,.
\end{aligned}
\end{equation}
\end{proof}

\section{Compositional associative Double-Pushout rewriting with conditions}\label{sec:DPO}

In this section, we will develop a variant of \emph{Double Pushout (DPO) rewriting} for \emph{linear rules}, which may be obtained from the well-known ``traditional'' DPO-rewriting framework of Ehrig et al.~\cite{ehrig1973} with its nearly 50 years of developments~\cite{ehrig1991parallelism,CorradiniMREHL97,DBLP:conf/gg/1997handbook,lack2005adhesive,ehrig:2006aa,ehrig2006adhesive,ehrig2010categorical,Habel:2012aa,Braatz:2010aa,ehrig2014mathcal} via requiring the underlying $\cM$-adhesive categories to satisfy Assumption~\ref{as:DPO}. As we will demonstrate, this particular specialization of the original theory yields a semantics for rewriting with certain additional properties, which we refer to as \emph{compositionality} and \emph{associativity}. The former property is a statement on the existence of a certain form of \emph{sequential rule composition} (cf.\ Definition~\ref{def:compACrefined}), which via a variant of the classical \emph{concurrency theorem} (cf.\ Theorem~\ref{thm:conc}) opens the possibility to develop algorithms for the static analysis of rule application sequences. The \emph{associativity} phenomenon for sequential compositions of linear rules with conditions (cf.\ Theorem~\ref{thm:assocAC}) is an original result of the present paper, adding to the aforementioned tool-set of static analysis the possibility to analyze sequences of more than two rule applications (ultimately leading to the notions of \emph{tracelet analysis}~\cite{behr2019tracelets} and \emph{rule algebras}~\cite{bdg2016,bp2018,bp2019-ext,nbSqPO2019,BK2020}).

\begin{notConv}
Since many of the constructions and results presented in the following constitute variants of well-known results from the rich DPO-rewriting literature, by \textbf{convention} a mention of ``\textbf{compare}'' or ``\textbf{cf.}'' without additional comments indicates that the results coincide with those in the literature (modulo possibly some notational adjustments and minor modifications). On several occasions, we nevertheless provide full proofs of some of the essential statements, since these will provide the ``blueprint'' for the corresponding constructions in the SqPO-setting presented in Section~\ref{sec:SqPO}. Our choices of notations follows mostly those taken in~\cite{bp2018}.
\end{notConv}

\subsection{DPO-rewriting in $\cM$-adhesive categories}\label{sec:DPOR}

Contrary to the traditional graph rewriting literature, we prefer to interpret spans as encoding a partial injective map going from the \emph{right} leg to the left leg, rather than the other way around. This notational convention as well as the focus on \emph{linear rules} (i.e.\ rules based upon spans of $\cM$-morphisms as opposed to \emph{non-linear} rules based upon generic morphisms) was motivated by our work on the rule algebra and related tracelet constructions~\cite{bdg2016,bp2018,bp2019-ext,nbSqPO2019, behr2019tracelets,BK2020}. Furthermore, all constructions presented in the following are naturally defined only up to isomorphisms, whence the choice to consider isomorphism classes of linear rules.

\begin{definition}[Linear rules]
  Let $\bfC$ be an $\cM$-adhesive category. We denote by $\Lin{\bfC}$ the \emph{set of linear rules}, defined as the set of isomorphism classes
  \begin{equation}
    \Lin{\bfC}\eqdef\left.\left\{
        r=\rSpan{O}{o}{K}{i}{I} 
        \right\vert O,K,I\in \obj{\bfC},\; o,i\in \cM
      \right\}\diagup_{\cong}\,.
  \end{equation}
 Here, we define $r=(O\leftarrow K\rightarrow I)$ and $r'=(O'\leftarrow K'\rightarrow I')$ as isomorphic when there exist isomorphisms $\omega:O\rightarrow O'$, $\kappa:K\rightarrow K'$ and $\iota:I\rightarrow I'$ that make the evident diagram commute. Thus a (representative of a) linear rule $r\in \Lin{\bfC}$ is a span of $\cM$-morphisms $o,i\in \cM$ with \emph{\underline{O}utput object} $O$, \emph{\underline{K}ontext object} $K$ and \emph{\underline{I}nput object} $I$.
\end{definition}

 The precise interpretation of the concept of linear rules is provided in the form of the following main definition of DPO rewriting:
\begin{definition}[DPO rewriting]\label{def:DPOR}
  Let $r\eqdef\rSpan{O}{o}{K}{i}{I}\in \Lin{\bfC}$ be a linear rule, let $X\in \obj{\bfC}$ be an object, and let $m:I\rightarrow X \in \cM$ be an $\cM$-morphism\footnote{For our construction of a DPO-type rule algebra framework, we will only be interested in admissible matches that are in the class $\cM$, even though DPO rewriting in its most general form  as e.g.\ discussed in~\cite{ehrig2014mathcal} would permit also non-monic matches.}. Then $m$ is defined to be an \emph{admissible match} for the application of $r$ to $X$, if and only if the diagram below is constructible:
  \begin{equation}\label{eq:DPOappl}
  \inputtikz{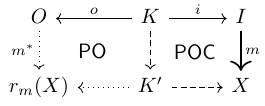}
  \end{equation}
  Here, the square marked ${\mathsf{POC}}$ must be constructible as a \emph{pushout complement}, while if this square exists the square marked ${\mathsf{PO}}$ is always constructible as a \emph{pushout} (cf.\ Assumption~\ref{as:DPO}), whence the moniker \emph{Double-Pushout (DPO) rewriting} is derived. In this case, we refer to $r_m(X)\in \obj{\bfC}$ as the \emph{rewrite of $X$ via the rule $r$ along the (admissible) match $m$}. We introduce the notation $\rMatch{r}{X}$ for the \emph{set of admissible matches} for the application of the rule $r$ to the object $X$:
  \begin{equation}\label{eq:MplainDPOrules}
    \rMatch{r}{X}\eqdef\{(m:I\rightarrow X)\in \cM\mid \text{POC $(1)$ in~\eqref{eq:DPOappl} exists}\}\,.
  \end{equation}
  For compatibility with the standard DPO rewriting literature, we will sometimes use the notation
  \[
    r_m(X)\xLeftarrow[{r,m}]{} X
  \]
  in order to explicitly reference the information contained in~\eqref{eq:DPOappl}. Moreover, the morphism $m^{*}$ is referred to as the \emph{comatch of $m$} (under the application of linear rule $r$ to the object $X$).
\end{definition}

In order to provide a quick intuitive illustration of the DPO rewriting concept, consider the edge creation rule described in the introduction:
\begin{equation}\label{eq:ePlus}
  e_{+}\eqdef \rSpan{%
  \inputtikzB{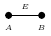}{-0.25}
  }{}{\inputtikzB{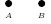}{-0.25}
    }{}{\inputtikzB{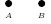}{-0.25}
}
\end{equation}
Here, as customary in the graph rewriting literature, the structure of the linear rule (in this case a span of undirected multigraphs) is indicated via the  labeling indices, i.e.\ in the present case reflecting that the vertices marked $A$ and $B$ are preserved in applications of this rule, and that the edge marked $E$ is created. An example for an admissible match and the respective rule application to the example graph
\[
  X\eqdef \inputtikzB{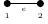}{-0.24}
\]
along the example match $m$ (which sends the vertices $A$ and $B$ of $I$ to the vertices $1$ and $2$ of $X$, respectively) is depicted below:
\begin{equation*}
\vcenter{\hbox{\includegraphics{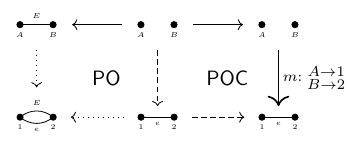}}}
\end{equation*}

\subsection{From conditions to application conditions for rewriting rules}\label{sec:CtoAC}

Conditions formulated for objects and for morphisms interact in a straightforward manner with the concept of rewriting rules, which requires two key constructions: the \emph{shift construction}, as introduced in Theorem~\ref{thm:Shift}, and the so-called \emph{transport construction}. We will follow the standard literature on $\cM$-adhesive categories (cf.\ e.g.~\cite{ehrig2014mathcal}) in defining the latter construction\footnote{In this definition, there is, \textit{a priori}, a choice to be made about the ``direction'' of the transport; our chosen convention agrees with the one in the literature and will prove convenient in our later applications to notions of compositionality of rules with application conditions.}.

\begin{definition}[Transport of conditions over rules; cf.\ \cite{ehrig2014mathcal}, Construction~3.15]\label{def:Trans}
Let 
\[
r\eqdef \rSpan{O}{o}{K}{i}{I}\in \Lin{\bfC}
\]
be a linear rule, and let $\ac{c}_O$ be an application condition over $O$. Then we define a \emph{transport construction} $\Trans$ such that $\Trans(r,\ac{c}_O)$ is an application condition over $I$, and which is constructed inductively as follows:
\begin{romanenumerate}
\item Case $\ac{c}_O=\ac{true}$:
\begin{equation}
  \Trans(r,\ac{true})\eqdef\ac{true}\,.
\end{equation}
\item Case $\ac{c}_O=\exists(a,\ac{c}_O')$ with $a:O\rightarrow O'$: if the diagram below
\begin{equation}\label{eq:transExists}
\inputtikz{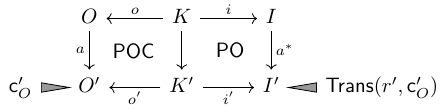}
\end{equation}
is constructible, i.e.\ if the pushout complement marked $\mathsf{POC}$ exists, we define
\begin{equation}
  \Trans(r,\exists(a,\ac{c}_O'))\eqdef\exists(a^{*},\Trans(r',\ac{c}_O'))\,,
\end{equation}
with $r'\eqdef\rSpan{O'}{o'}{K'}{i'}{I'}$. Otherwise, we define
\begin{equation}
  \Trans(r,\exists(a,\ac{c}_O'))\eqdef\ac{false}\,.
\end{equation}
\item Case $\ac{c}_O=\neg \ac{c}_O'$: 
\begin{equation}
  \Trans(\neg \ac{c}_O')\eqdef\neg \Trans(\ac{c}_O')\,.
\end{equation}  
\item Case $\ac{c}_O=\land_{j\in J}\, \ac{c}_{O}^{(j)}$: 
\begin{equation}
  \Trans(r,\land_{j\in J}\, \ac{c}_{O}^{(j)})\eqdef\land_{j\in J} \Trans(r,\ac{c}_{O}^{(j)})\,.
\end{equation}  
\end{romanenumerate}
\end{definition}
It is straightforward to verify that the transport construction is invariant under the various possible isomorphisms involved in the relevant constructions of pushouts and pushout complements, for precisely the same reasons as those ensuring the invariance of the shift construction under isomorphisms as detailed in the proof of Theorem~\ref{thm:Shift}.

\begin{lemma}[Property of transport along DPO-type rules; cf.\ \cite{ehrig2014mathcal}, Lemma~3.14]\label{lem:trans}
In an $\cM$-adhesive category $\bfC$ satisfying Assumption~\ref{as:DPO}, let $r=(O\leftarrow K\rightarrow I)\in \Lin{\bfC}$ be a linear rule, and let $\ac{c}_O$ be an application condition over $O$. Then for any DPO-admissible match $(m:I\rightarrow X)\in\Match{r}{X}$ of the rule $r$ into an object $X\in \obj{\bfC}$, and if $m^{*}$ denotes the comatch of $m$, one finds that
\begin{subequations}\label{eq:trans}
\begin{align}
&m^{*}\vDash \ac{c}_O\iff m\vDash \Trans(r,\ac{c}_O)\,,\\
\intertext{with}
&\inputtikz{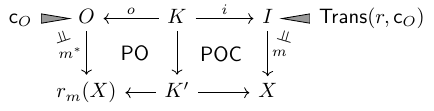}\,.
\end{align}
\end{subequations}
\end{lemma}

The transport construction allows us to choose, without loss of generality,  a ``standard position'' for the application conditions in a linear rule, where we fix the following conventions:

\begin{definition}[Standard form for DPO-type rules with conditions and admissible matches]\label{def:SF}
  Let $\LinAc{\bfC}$ denote the \emph{set of linear rules with application conditions in standard form}, whose elements $R\in\LinAc{\bfC}$ are defined to be of the form
  \begin{equation}
    R= (r,\ac{c}_I)\,,\quad r=\rSpan{O}{o}{K}{i}{I}\in\Lin{\bfC}\,.
  \end{equation}
  Consequently, we introduce the notion of \emph{admissible matches} for applications of rules with application conditions to objects as follows; let $X\in \obj{\bfC}$ be an object, $R\in \LinAc{\bfC}$  as above a rule with application conditions, and $m:I\rightarrow X$ an element of $\cM$. Then we refer to $m$ as an \emph{admissible match} if and only if $m$ satisfies the application condition,
  \[
    m\vDash \ac{c}_I\,,
  \]
  and if the diagram below is constructible:
\begin{equation}\label{eq:Rapp}
\inputtikz{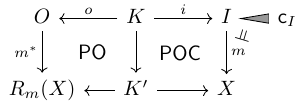}\,.
\end{equation}
Equivalently, admissibility of $m$ thus amounts to admissibility with respect to the linear rule without application conditions (i.e.\ in the sense of~\eqref{eq:MplainDPOrules}) combined with satisfaction of the application condition, resulting in the following compact formula for the \emph{set of admissible matches} $\Match{R}{X}$:
\begin{equation}\label{eq:defAdmRac}
  \Match{R}{X}\eqdef\{m\in M_r(X)\mid m\vDash \ac{c}_I\}\,.
\end{equation} 
\end{definition}

For later convenience, we will employ the shorthand notation $\dot{\equiv}$ to signify \emph{``equivalence under the constraint of admissibility''}~\footnote{Following the standard notational convention in the rewriting literature, we choose to not make the linear rule with respect to which admissibility is required explicit in our notation $\dot{\equiv}$, since we will only utilize this form of equivalence in situations where the nature of this rule will be clear from the given context.}:
\begin{definition}
Let $r=\rSpan{O}{o}{K}{i}{I}\in \Lin{\bfC}$ be a linear rule and $\ac{c}_I,\tilde{\ac{c}}_I$  conditions over $I$. Then we define
\begin{equation}\label{def:dotEquiv}
  \left(\ac{c}_I \dot{\equiv} \tilde{\ac{c}}_I\right)\iffeq \left(\forall X\in \obj{\bfC}:\forall m\in \Match{r}{X}:\quad m\vDash \ac{c}_I\iff m\vDash \tilde{\ac{c}}_I\right)\,.
\end{equation}
\end{definition}

As a further refinement, the $\Trans$ construction enables a certain form of compression of application conditions for linear rules.

\begin{definition}[Compressed standard form for conditions]\label{def:csfc}
Let $R\eqdef (r=(O\leftarrow K\rightarrow I),\ac{c}_I)\in \LinAc{\bfC}$ be a linear rule with application conditions. Then we define the \emph{compressed standard form} for $\ac{c}_I$ as 
\begin{equation}
  \dot{\ac{c}}_I\eqdef \Trans(r,\Trans(\bar{r},\ac{c}_I))\,,
\end{equation}
where $\bar{r}\eqdef(I\leftarrow K\rightarrow O)$. 
\end{definition}
The intuition behind the above definition of \emph{``equivalence up to non-transportable subconditions''} is that while it is perfectly possible to define arbitrary conditions of the form $\exists(a:I\rightarrow A,\ac{c}_A)$ over the input $I$ of a linear rule, only those conditions will contribute in applications of the linear rule via matches that are transportable via $\Trans$, since by definition of $\Trans$ the operation $\Trans(r,\Trans(\bar{r},\exists(a:I\rightarrow A,\ac{c}_A))$ in effect tests whether or not the relevant pushout complement exists such that an admissible match of the rule could satisfy $\exists(a:I\rightarrow A,\ac{c}_A)$. This also implies that
\begin{equation}
  \ac{c}_I\, \dot{\equiv}\, \dot{\ac{c}}_I\,,
\end{equation}
thus motivating the notation $\dot{\ac{c}}_I$. A further illustration of this phenomenon is provided in Example~\ref{ex:assocDPO}.\\

We conclude this subsection by stating a number of technical lemmas that are necessary in order to derive our novel \emph{associative compositional DPO rewriting} framework as presented in the following subsection, which concern certain important properties of the $\Trans$ construction and of the compatibility of the $\Shift$ and $\Trans$ constructions:

\begin{lemma}[Units for Trans]\label{lem:transUnit}
    Let $X\in \obj{\bfC}$ be an arbitrary object and $\ac{c}_X$ a condition over $X$. Then with $r_{id_X}\eqdef\rSpan{X}{id_X}{X}{id_X}{X}\in \Lin{\bfC}$ the ``identity rule on $X$'', we find that
    \begin{equation}
        \Trans(r_{id_X},\ac{c}_X)\equiv \ac{c}_X\,.
    \end{equation}
\end{lemma}
\begin{proof}
The proof follows directly from the property of the $\Trans$ construction stated in Lemma~\ref{lem:trans}, whereby one finds for arbitrary admissible matches $(m:X\rightarrow Y)\in\Match{r_{id_X}}{X}$ that
\begin{subequations}\label{eq:transSat}
\begin{align}
&m^{*}=m\vDash \ac{c}_X\; \Leftrightarrow\; m\vDash \Trans(r_{id_X},\ac{c}_X)\,,\\
\intertext{with}
&\inputtikz{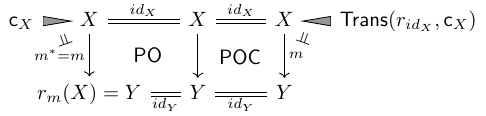}\,.
\end{align}
\end{subequations}
Here, the pushout complement in the squares marked $\mathsf{POC}$ always exists by virtue of Lemma~\ref{lem:Main}\ref{lem:fPOPB}, and $m^{*}=m$ as well as $r_m(X)=Y$ follow due to stability of isomorphisms under pushouts.
\end{proof}

\begin{lemma}[Compositionality of Trans]\label{lem:transComp}
  Given two composable spans of $\cM$-morphisms
\[
  r\equiv \rSpan{C}{b}{B}{a}{A}\; \quad \text{and }\quad
  s\equiv\rSpan{E}{d}{D}{c}{C}\,,
\]
define their composition via pullback as
  \begin{equation}
  \inputtikz{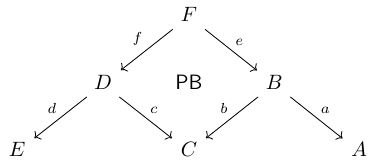}\qquad s\circ r\eqdef (E\xleftarrow{d\circ f}F\xrightarrow{a\circ e}A)\,,
  \end{equation}
  which is again a span of $\cM$-morphisms (by stability of $\cM$-morphisms under pullbacks and compositions), and thus $r,s,s\circ r\in \Lin{\bfC}$. Let $\ac{c}_E$ be a condition over $E$. Then we find that
  \begin{equation}
    \Trans(r,\Trans(s,\ac{c}_E))\,\dot{\equiv}\, \Trans(s\circ r,\ac{c}_E)\,.
  \end{equation}
\end{lemma}
\begin{proof} The proof relies upon the property of the transport construction stated in Lemma~\ref{lem:trans} as well as on the $\cM$-adhesivity of the underlying category $\bfC$. We proceed by constructing the following commutative diagram in two different ways for the two directions of the proof:
  \begin{equation}\label{eq:TransAux}
  \inputtikz{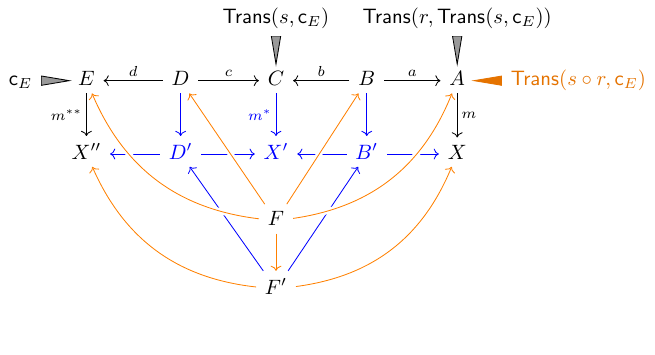}
\end{equation}

  \paragraph{``$\Rightarrow$'' direction:} Suppose that $m\in \Match{r}{X}$ and that the comatch $m^{*}$ of $m$ satisfies $m^{*}\in \Match{s}{X'}$ (with $X'=r_m(X)$). Then according to Lemma~\ref{lem:trans}, this implies that
  \[
       m^{**}\vDash \ac{c}_E \quad \Leftrightarrow \quad 
       m^{*}\vDash \Trans(s,\ac{c}_E) \Leftrightarrow \quad 
       m\vDash \Trans(r,\Trans(s,\ac{c}_E))\,.
  \]
  We have to demonstrate that $m\in \Match{s\circ r}{X}$ as well as
  \begin{equation}
      m\vDash \Trans(r,\Trans(s,\ac{c}_E))\quad \Rightarrow \quad m\vDash \Trans(s\circ r,\ac{c}_E)\,.
  \end{equation}
Admissibility of $m$ and $m^{*}$ entails that the squares formed in the back row of~\eqref{eq:TransAux} (the ones drawn in black and {\color{blue}blue}) are constructible as pushouts and pushout complements, respectively. Construct the objects $F$ and $F'$ as pullbacks,
  \[
      F=\pB{D\xrightarrow{c}C\xleftarrow{b}B}\quad \text{and }\quad
      F'=\pB{D'\rightarrow X'\leftarrow B'}\,,
  \]
  which by the universal property of pullbacks induces a unique arrow $F\rightarrow F'$. By stability of $\cM$-morphisms under pullbacks and by the $\cM$-morphism decomposition property, respectively, all morphisms thus constructed are found to be in $\cM$. Invoking pullback-pullback decomposition (Lemma~\ref{lem:Main}\ref{lem:PBPBdec}) and the $\cM$-van Kampen property twice (cf.\ Definition~\ref{def:Madh}), we conclude that the induced squares $\cSquare{F,F',D',D}$ and $\cSquare{F,F',B',B}$ are in fact pushouts. Thus by pushout composition, the front left and right ``curvy'' faces (in orange) are pushouts. This entails that $m\in \Match{s\circ r}{X}$, and thus by definition of $\Trans$ that indeed
  \[
      m\vDash \Trans(s\circ r,\ac{c}_E)\,.
  \]

\paragraph{``$\Leftarrow$'' direction:} Suppose that $m\in \Match{s\circ r}{X}$, which by Lemma~\ref{lem:trans} implies that if $m^{**}$ is the comatch of $m$ under the application of the rule $s\circ r$, then
\[
    m^{**}\vDash \ac{c}_E\quad \Leftrightarrow \quad m \vDash \Trans(s\circ r,\ac{c}_E)\,.
\]
Admissibility of $m$ for the rule $s\circ r$ applied to the object $X$ entails that the ``curvy'' front right and front left squares of the diagram~\eqref{eq:TransAux} drawn in black and {\color{orange}orange} are constructible as pushout complement and pushout, respectively. We may then construct the remaining parts of the diagram via forming the pushouts
  \[
    D'=\pO{D\leftarrow F\rightarrow F'}\,,\;
    B'=\pO{B\leftarrow F\rightarrow F'}\,,\; 
    X'=\pO{C\leftarrow B\rightarrow B'}\,,
  \]
which uniquely induces the remaining arrows drawn in {\color{blue}blue} (and where we could have equivalently defined $X'$ as $X'=\pO{C\leftarrow D\rightarrow D'}$). By virtue of three applications of pushout-pushout decomposition (Lemma~\ref{lem:Main}\ref{lem:POPOdec}), we conclude that all squares in the back of diagram~\eqref{eq:TransAux} thus constructed are pushouts. Furthermore, stability of $\cM$-morphisms under pushouts implies that all newly constructed morphisms are in $\cM$. Since thus the back part of the diagram encodes two DPO rewrite steps with $m\in \Match{r}{X}$, $m^{*}$ the comatch of $m$ under application of the rule $r$, $m^{*}\in \Match{s}{r_m(X)}$, and since $m^{**}$ is also the comatch of $m^{*}$ under application of $s$, we find by Lemma~\ref{lem:trans} that 
  \[
      m^{**}\vDash \ac{c}_E \quad \Leftrightarrow \quad
      m^{*}\vDash\Trans(s,\ac{c}_E) \quad \Leftrightarrow \quad
      m\vDash \Trans(r,\Trans(s,\ac{c}_E))\,,
  \]
  which concludes the proof.
\end{proof}

The compositionality of the $\Trans$ construction in particular permits an efficient encoding of the reduced standard form of application conditions:

\begin{corollary} 
Let $R\eqdef ((O\xleftarrow{o} K\xrightarrow{i} I),\ac{c}_I)\in \LinAc{\bfC}$ be a linear rule with application conditions. Then the compressed standard form $\dot{\ac{c}}_I$ for $\ac{c}_I$ according to Definition~\ref{def:csfc} evaluates to 
\begin{equation}
  \dot{\ac{c}}_I\,\dot{\equiv}\, \Trans(I\xleftarrow{i} K\xrightarrow{i} I,\ac{c}_I)\,.
\end{equation}
\end{corollary}

\begin{lemma}[Compatibility of Shift and Trans; compare~\cite{GOLAS2014}, Fact~3.14]\label{lem:ST}
Given the data as in the commutative diagram below,
\begin{equation}\label{eq:transCST}
\inputtikz{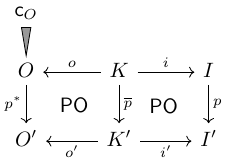}\,,
\end{equation}
letting $r=\rSpan{O}{o}{K}{i}{I}$ and $r'=\rSpan{O'}{o'}{K'}{i'}{I'}$, we have that for all objects $X$ and for all admissible matches $n\in \Match{r'}{X}$ of $r'$ into $X$,
\begin{equation}
n\vDash \Shift(p,\Trans(r,\ac{c}_{O}))
  \Leftrightarrow n\vDash \Trans(r',\Shift(p^{*},\ac{c}_{O}))\,,
\end{equation}
which we can write more compactly as
\begin{equation}
\Shift(p,\Trans(r,\ac{c}_{O}))\,\dot{\equiv}\,\Trans(r',\Shift(p^{*},\ac{c}_{O}))\,.
\end{equation}
\end{lemma}
\begin{proof} Let us fix an object $X$ and some admissible match $n\in \Match{r'}{X}$.

\paragraph{``$\Rightarrow$'' direction:} Suppose that $(n:I'\rightarrow X)\in \Match{r'}{X}$ satisfies $n\vDash \Shift(p,\Trans(r,\ac{c}_{O}))$ (which by definition of satisfaction of conditions entails in particular that $\Trans(r',\Shift(p^{*},\ac{c}_{O}))\not\equiv \ac{false}$). %
Since $n$ is by assumption an admissible match of $r'$, we can rewrite $X$ by applying $r'$ along $n$, %
resulting in the diagram below (where the top part is inserted from the assumption of the lemma):
\begin{equation}\label{eq:TransProofA}
\inputtikz{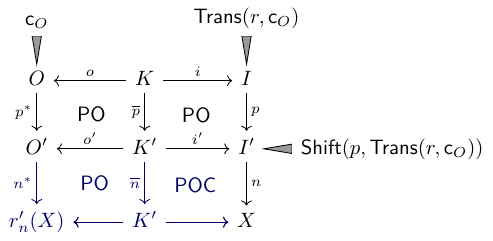}
  \end{equation}
By composition of pushout squares, we conclude that the $\cM$-morphism $m=n\circ p$ is an admissible match for $r$, %
which entails that $r'_n(X)\cong r_m(X)$. %
By definition of $\Shift$, $n\vDash \Shift(p,\Trans(r,\ac{c}_{O}))$ implies that $m=n\circ p$ satisfies $m\vDash \Trans(r',\ac{c}_O)$, %
and moreover that $\Trans(r',\ac{c}_O)\not\equiv \ac{false}$. %
Since we found that $m\in \Match{r}{X}$, %
by definition of $\Trans$ we have that the comatch $m^{*}:O\rightarrow r_m(X)$ of $m$ (which is by construction of DPO rule applications in $\cM$) has the property $m^{*}=n^{*}\circ p^{*}\vDash \ac{c}_O$. %
The latter implies that $n^{*}\vDash \Shift(p^{*},\ac{c}_O)$. %
Since by assumption $n\in \Match{r'}{X}$, %
and since $n^{*}$ is the comatch of $n$, %
we finally conclude that indeed $n\vDash \Trans(r',\Shift(p^{*},\ac{c}_O))$.

\paragraph{``$\Leftarrow$'' direction:} The proof is entirely analogous to the previous case (starting from the observation that $n\in \Match{r'}{X}$ together with the data provided in~\eqref{eq:transCST} entails that $m\in \Match{r}{X}$).
\end{proof}

\subsection{A refined notion of sequential compositions of DPO-type rules with conditions}\label{sec:rCompRefined}

In this subsection, we will present a refinement of the notion of sequential rule compositions known from the traditional rewriting literature~\cite{ehrig2014mathcal} obtained via requiring the underlying $\cM$-adhesive category $\bfC$ to satisfy Assumption~\ref{as:DPO}, which as we will demonstrate guarantees a number of additional technical properties of this type of composition. Referring to Remark~\ref{rem:dpoCC} for the precise technical disambiguation, the motivation for our construction consisted in the idea that the operation of \emph{sequential composition} of two linear rules should lift to an operation permitting to compose arbitrary numbers of linear rules, in the sense that the composition operation should satisfy a certain abstract \emph{associativity} property (cf.\ Section~\ref{sec:assocDPO}). From a technical standpoint, ensuring this notion of \emph{compositionality} permits to develop novel static analysis techniques for rule-based systems, such as the \emph{DPO-type rule algebra framework}~\cite{bdg2016,bp2018,bp2019-ext} for the study of stochastic rewriting systems, which has recently been extended to the setting of linear rules with conditions based upon the present article in~\cite{BK2020}. The second main application of compositional rewriting has been the development of the theory of \emph{tracelets}~\cite{behr2019tracelets}.

\begin{definition}\label{def:compACrefined}
  Let $\bfC$ be a category satisfying Assumption~\ref{as:DPO}. Let ${\color{blue}R_j\equiv (r_j,\ac{c}_{I_j})}\in\LinAc{\bfC}$ be two linear rules with application conditions ($j=1,2$), and let 
  \[
  {\color{h1color}\mu_{21}\equiv\rSpan{I_2}{m_2}{M_{21}}{m_1}{O_1}}
  \]
  be a span of monomorphisms (i.e.\ $m_{1},m_{2}\in \cM$). If the diagram below is constructible,
  \begin{equation}\label{eq:defDPOcomp}
  \inputtikz{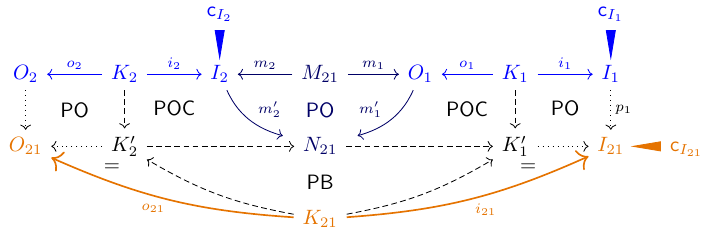}
  \end{equation}
  where
  \begin{equation}\label{eq:acIcomp}
   {\color{h2color} \ac{c}_{I_{21}}}
   \eqdef\Shift(p_1,{\color{blue}\ac{c}_{I_1}})\;\bigwedge\;
   \Trans\left({\color{h1color}N_{21}}\leftarrow K_1'\rightarrow {\color{h2color}I_{21}},
   \Shift({\color{h1color}m_2'},{\color{blue}\ac{c}_{I_2}})\right)\,,
  \end{equation}
  and if $\ac{c}_{I_{21}}\not\equiv \ac{false}$, then we call $\mu_{21}$ an \emph{admissible match} for the rules with conditions $R_2$ and $R_1$, denoted 
  \[
    \mu_{21}\in M_{R_2}(R_1)\,.
  \]
  In this case, we introduce the notation $\comp{R_2}{\mu_{21}}{R_1}$ to denote the composite,
  \begin{equation}\label{eq:defCompRC}
    \comp{R_2}{\mu_{21}}{R_1}\eqdef{\color{h2color}\left(O_{21}\xleftarrow{o_{21}}
    K_{21}\xrightarrow{i_{21}}I_{21},\ac{c}_{I_{21}}\right)}\,.
  \end{equation}
\end{definition}

\begin{remark}\label{rem:dpoCC}
The variant of rule composition provided in Definition~\ref{def:compACrefined} (i.e.\ starting from an ``overlap'' of two rules encoded as an $\cM$-monic span, followed by taking a pushout of this span) follows the philosophy put forward in~\cite{Boehm:1987aa,lack2005adhesive} (sometimes referred to as \emph{D-concurrent} composition). However, for the setting of DPO-rewriting over $\cM$-adhesive categories for rules with conditions, an alternative construction, referred to as \emph{E-concurrent} composition, had been the de facto standard in the rewriting literature, developed by Ehrig et al. (cf.\ e.g.\ \cite{ehrig2014mathcal}, Definition~4.13). More precisely, based upon the assumption that the underlying $\cM$-adhesive category $\bfC$ possesses an \emph{$\cE'$-$\cM$-pair-factorization} (cf.\ Remark~\ref{rem:Shift}), the starting point of constructing an \emph{E-concurrent} rule composition according to~\cite{ehrig2014mathcal} consists in picking a \emph{cospan} ${\color{h1color}(m_2',m_1')}\in \cE'$, as opposed to picking a \emph{span} of $\cM$-morphisms ${\color{h1color}(m_2,m_1)}$ and taking the pushout to obtain the cospan ${\color{h1color}(m_2',m_1')}$ as in~\eqref{eq:defDPOcomp} for the \emph{D-concurrent} approach. While the \emph{E-concurrent} approach possesses the technical advantage that it does not require the category $\bfC$ to possess $\cM$-effective unions, nor that rules must be linear nor matches necessarily in $\cM$ (and thus applying to a broader class of rewriting systems), the reliance upon the $\cE'$-$\cM$-factorizations is also responsible for certain algorithmic disadvantages.  In close analogy to the discussion provided in Remark~\ref{rem:Shift}, exhaustively enumerating all possible cospans in $\cE'$ suitable for the composition of two rewriting rules is algorithmically quite non-trivial in the general case, while enumerating the (isomorphism classes of $\cM$-) monic spans (i.e.\ the ``overlaps'') of two linear rules is comparatively straightforward (c.f.\ e.g.\ \cite{behr2020commutators} for a recent implementation via the \textsc{Microsoft Z3} SMT-solver). On the other hand, utilizing yet again the results of Corollary~\ref{cor:jE}, our refinement of \emph{D-concurrent} compositions as well as the respective variant of the \emph{concurrency theorem} (see Theorem~\ref{thm:conc}) can be demonstrated to constitute a special case of the \emph{E-concurrent} constructions of Ehrig et al.~\cite{ehrig2014mathcal}: for an $\cM$-adhesive category satisfying Assumption~\ref{as:DPO}, a cospan ${\color{h1color}(m_2',m_1')}$ of $\cM$-morphisms constructed via taking the pushout of a span of $\cM$-morphisms ${\color{h1color}(m_2,m_1)}$ is in particular jointly epimorphic, and thus constitutes a special case of a cospan in $\cE'$ (cf.\ \cite[Fact~3.7]{Braatz:2010aa}, applied under the assumption of $\cM$-effective unions).

\end{remark}

The definition of the composition operation $\comp{.}{.}{.}$ entails a number of highly non-trivial effects in practical computations, which originate from the interplay of admissibility of matches for rules without application conditions and the requirements on the induced composite application conditions. One of the most striking such results, well known also from the traditional graph rewriting literature~\cite{ehrig2014mathcal}, is the following:

\begin{lemma}[Trivial matches]\label{lem:trivMatch}
  By definition of the notion of admissible matches, for any two linear rules with application conditions $R_j\equiv(r_j,\ac{c}_{I_j})\in\LinAc{\bfC}$ ($j=1,2$), the \emph{trivial match}
  \[
    \mu_{\mIO}\eqdef\rSpanAlt{I_2}{}{\mIO}{}{O_1}
  \]
  is an admissible match $\mu_{\mIO}\in M_{R_2}(R_1)$ if and only if the composite condition $\ac{c}_{I_{21}}$ does not evaluate to $\ac{false}$. 
  \end{lemma}
\begin{proof}
  The proof follows directly from the definition of the composition operation $\comp{.}{.}{.}$, namely by construction of the following diagram:
\begin{equation}\label{eq:trivCompAux}
\inputtikz{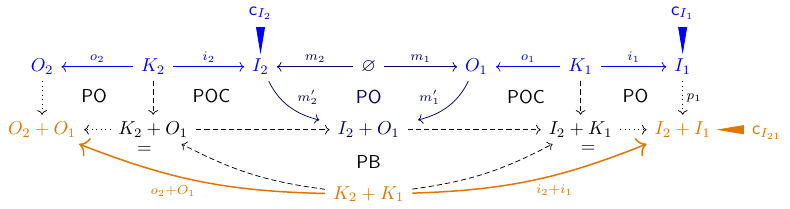}
  \end{equation}
By virtue of Lemma~\ref{lem:binaryCoproducts} and Lemma~\ref{lem:POCsepcial}, the pushout complements marked $\mathsf{POC}$ in the diagram above always exist. To determine whether the trivial match $\mu_{\mIO}$ is an admissible match, it then remains to evaluate the composite condition $\ac{c}_{I_{21}}$, which according to~\eqref{eq:acIcomp} of Definition~\ref{def:compACrefined} reads
 \begin{equation}\label{eq:acIcompTM}
 \begin{aligned}
   {\color{h2color} \ac{c}_{I_{21}}}
   &\eqdef\Shift(p_1:I_1\rightarrow I_2+I_1,\ac{c}_{I_1})\\
   &\qquad \bigwedge\;
   \Trans\left({\color{h1color}I_2+O_1}\leftarrow I_2+K_1\rightarrow {\color{h2color}I_2+I_1},
   \Shift({\color{h1color}m_2':I_2\rightarrow I_2+O_1},\ac{c}_{I_2})\right)\,.
  \end{aligned}
  \end{equation}
Thus the claim follows, since the above condition may evaluate to $\ac{false}$ in general, such as in the case where $\ac{c}_{I_2}=\neg \exists(I_2\rightarrow I_2+O_1,\ac{true})$.
\end{proof}

Nevertheless, it is possible to exhibit one special rule for which $\mu_{\mIO}$ is always an admissible match:

\begin{lemma}[Neutral element for $\comp{.}{.}{.}$]\label{lem:neut}
  By definition of $\comp{.}{.}{.}$, the \emph{trivial rule}
  \[
    R_{\mIO}\eqdef(\mIO\leftarrow \mIO \rightarrow \mIO,\ac{true})
  \]
  is the (left- and right-) \emph{neutral element} for $\comp{.}{.}{.}$.
  \end{lemma}
\begin{proof}
  The proof follows from a specialization of~\eqref{eq:trivCompAux} in the proof of Lemma~\ref{lem:trivMatch}, by specializing either of the two linear rules involved to the trivial rule. Note first that on the level of rules without application conditions, the only admissible match between the trivial rule and another linear rule is the trivial match $\mu_{\mIO}$. Let us then compute the condition $\ac{c}_{I_{21}}$ of the composite for the case $r_1=r_{\mIO}$ and for generic $r_2=(O_2\leftarrow K_2\rightarrow I_2)$, which reads according to \eqref{eq:acIcomp} of Definition~\ref{def:compACrefined} reads
 \begin{equation}\label{eq:acIcompTMa}
 \begin{aligned}
   {\color{h2color} \ac{c}_{I_{21}}}
   &=\Shift(p_1:\mIO\rightarrow I_2,\ac{true})\\
   &\qquad \bigwedge\;
   \Trans\left({\color{h1color}I_2}\leftarrow I_2\rightarrow {\color{h2color}I_2},
   \Shift({\color{h1color}q_2:I_2\rightarrow I_2},\ac{c}_{I_2})\right)\\
   &\equiv \ac{c}_{I_2}\,.
  \end{aligned}
  \end{equation}
Thus for every linear rule $r_2$ with application condition $\ac{c}_{I_2}\not\equiv \ac{false}$, we have $\mu_{\mIO}\in \Match{R_2}{R_{\emptyset}}$. For the remaining case, consider that $r_1=(O_1\leftarrow K_1\rightarrow I_1)$ is an arbitrary linear rule with application condition $\ac{c}_{I_1}\not\equiv \ac{false}$. Again, admissibility of $\mu_{\mIO}$ as a match of the rules without conditions follows by a specialization of~\eqref{eq:trivCompAux}, so it remains to compute the composite condition $\ac{c}_{I_{21}}$: 
  \begin{equation}\label{eq:acIcompTMb}
 \begin{aligned}
   {\color{h2color} \ac{c}_{I_{21}}}
   &=\Shift(p_1:I_1\rightarrow I_1,\ac{c}_{I_1})\\
   &\qquad \bigwedge\;
   \Trans\left({\color{h1color}O_1}\leftarrow K_1\rightarrow {\color{h2color}I_1},
   \Shift({\color{h1color}q_2:\mIO\rightarrow O_1},\ac{true})\right)\\
   &\equiv \ac{c}_{I_1}\,.
  \end{aligned}
  \end{equation}
\end{proof}

\subsection{Concurrency theorem for DPO-type rules with conditions}\label{sec:DPOconcur}

We will need the following \emph{concurrency theorem}, which is a variant of a result of~\cite{Habel:2012aa} adapted to our refined notion of rule compositions, and which for the case of rules without conditions was introduced in~\cite{bp2019-ext}:

\begin{restatable}[Concurrency theorem, compare Thm.~4 of~\cite{Habel:2012aa-ext} and Thm.~2.7 of~\cite{bp2019-ext}]{theorem}{thmConc}\label{thm:conc}
Let $\bfC$ be an $\cM$-adhesive category satisfying Assumption~\ref{as:DPO}, $X_0\in \obj{\bfC}$ an object, and ${\color{blue}R_j\equiv (r_j,\ac{c}_{I_j})}$ be two linear rules with application conditions ($j=1,2$). Then there exists the following \emph{bijection}: 
\begin{romanenumerate}
  \item \emph{``Synthesis'':} For every sequence of rule applications
    \begin{equation}\label{eq:concTHM1}
      X_2\xLeftarrow[R_2,m_2]{} X_1 \xLeftarrow[R_1,m_2]{} X_0
    \end{equation}
    along admissible matches $m_1\in M_{R_1}(X_0)$ and $m_2\in M_{R_2}(X_1)$ with $X_1=r_{1_{m_1}}(X_0)$, there exist admissible matches $\mu_{21}\in M_{R_2}(R_1)$ of the linear rule $R_2$ into $R_1$ and $m_{21}\in M_{R_{21}}(X_0)$, with $R_{21}\equiv(r_{21},\ac{c}_{I_{21}})$, $r_{21}=\comp{R_2}{\mu_{21}}{R_1}$, and an application condition $\ac{c}_{I_{21}}$ computed as
    \begin{equation}\label{eq:concThmcompAC}
    \ac{c}_{I_{21}}=\Shift(I_1\rightarrow I_{21},\ac{c}_{I_1})\;\bigwedge\;
    \Trans\left(\rSpanAlt{N_{21}}{}{K_1'}{}{I_{21}},
      \Shift(I_2\rightarrow N_{21},\ac{c}_{I_1})\right)\,,
    \end{equation}
    where the morphisms and objects in this formula depend (uniquely up to isomorphism) on the input data, such that $X_2\cong R_{21_{m_{21}}}(X_0)$.
  \item \emph{``Analysis'':}  For every admissible match $\mu\in M_{R_2}(R_1)$ and for every rule application
    \begin{equation}
      X_2 \xLeftarrow[R_{21},m_{21}]{}X_0
    \end{equation}
    with $m_{21}\in M_{R_{21}}(X)$ and $R_{21}=\comp{R_2}{\mu_{21}}{R_1}$, there exists a pair of admissible matches such as in~\eqref{eq:concTHM1} which transform $X_0$ via $X_1$ into the same (up to isomorphism) object $X_2$.
\end{romanenumerate}
\end{restatable}
\begin{proof}
Referring the interested readers to~\cite{bp2019-ext} for the precise details, note first that at the level of linear rules without application conditions, the concurrency theorem of~\cite{bp2019-ext} entails the parts of the above statements pertaining to the existence of the admissible matches of ``plain'' rules. The concrete technical construction of the proof provided in~\cite{bp2019-ext} is summarized by the diagram below, where all vertical squares are pushouts (and where we have marked the relevant conditions for later convenience). The aforementioned proof consisted in verifying that the parts of the diagram marked in {\color{h1color}blue} can be uniquely constructed from the parts of the diagram colored in {\color{h2color}orange} and vice versa:
\begin{equation}\label{eq:concurrencyDPOproof}
\vcenter{\hbox{\includegraphics[width=0.9\textwidth]{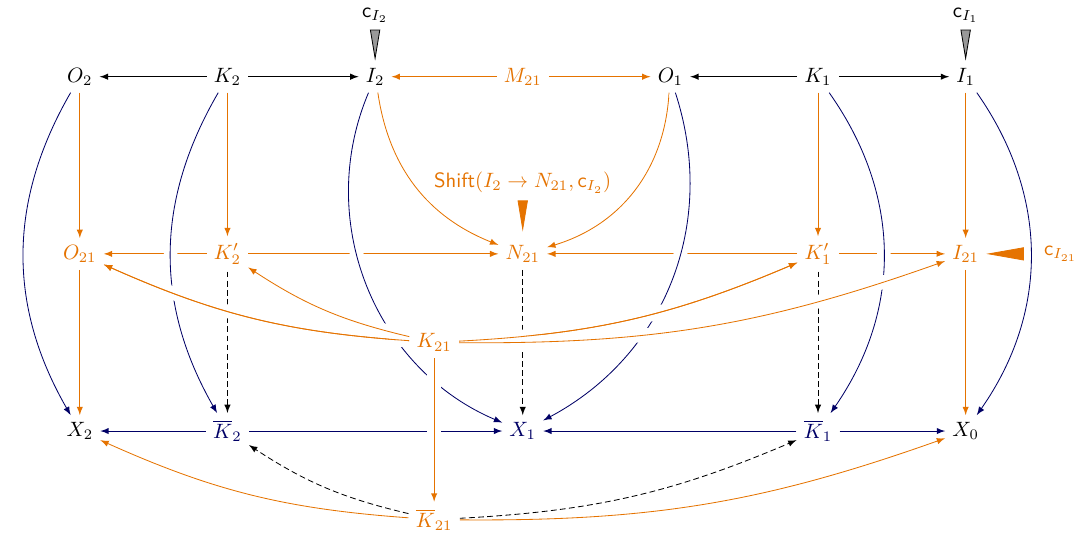}}}
\end{equation}
It thus remains to verify the part of the claim pertaining to the relevant conditions of the rules. 

\paragraph{``Analysis'' part of the proof:} Suppose that we are given  admissible matches ${\color{h1color}(m_1:I_1\rightarrow X_0)}\in M_{R_1}(X)$ and ${\color{h1color}(m_2:I_2\rightarrow X_1)}\in M_{R_2}(X_1)$ with $X_1=r_{1_{{\color{h1color}m_1}}}(X_0)$. Admissibility entails in particular that ${\color{h1color}m_1}\vDash \ac{c}_{I_1}$ and ${\color{h1color}m_2}\vDash \ac{c}_{I_2}$. By construction of the diagram in~\eqref{eq:concurrencyDPOproof}, we have that ${\color{h1color}m_1}$ and ${\color{h1color}m_2}$ factor as
\[
\begin{aligned}
    {\color{h1color}m_1}&={\color{h1color}(X_0\leftarrow I_1)}=(X_0{\color{h2color}\leftarrow I_{21}})\circ({\color{h2color}I_{21}\leftarrow }I_1)\\
    {\color{h1color}m_2}&={\color{h1color}(X_1\leftarrow I_2)}=(X_1{\color{h2color} \leftarrow N_{21}})\circ({\color{h2color} N_{21}\leftarrow} I_2)\,,
  \end{aligned}
\]
which entails by definition of the $\Shift$ construction that ${\color{h2color}m_{21}}:=({\color{h2color}I_{21}\rightarrow}X_0)$ and %
${\color{h2color}\bar{m}_{21}}=({\color{h2color}N_{21}\rightarrow} X_1)$ satisfy %
${\color{h2color}m_{21}}\vDash \Shift(I_1{\color{h2color}\rightarrow I_{21}},\ac{c}_{I_1})$ and %
${\color{h2color}\bar{m}_{21}}\vDash \Shift(I_2{\color{h2color}\rightarrow N_{21}},\ac{c}_{I_2})$, respectively. Noting that the rightmost two bottom squares in the back of~\eqref{eq:concurrencyDPOproof} are pushouts, we find in addition that
\[
    {\color{h2color}m_{21}}\vDash \Trans\left({\color{h2color}N_{21}\leftarrow K_1' \rightarrow I_{21}},
      \Shift(I_2{\color{h2color}\rightarrow N_{21}},\ac{c}_{I_1})\right)\,.
\]
Since according to Definition~\ref{def:compACrefined}, ${\color{h2color}R_{21}}:=\comp{R_2}{\mu_{21}}{R_1}={\color{h2color}(r_{21},\ac{c}_{I_{21}})}$, %
with ${\color{h2color}\ac{c}_{I_{21}}}$ as defined in~\eqref{eq:concThmcompAC}, we confirm that ${\color{h2color}m_{21}}\vDash{\color{h2color}\ac{c}_{I_{21}}}$, which concludes the proof of the \emph{``analysis''} part of the theorem.

\paragraph{``Synthesis'' part of the proof:} Supposing that we are given an admissible match ${\color{h2color}m_{21}}$ of %
the composite ${\color{h2color}R_{21}}$ of the rules with application conditions $R_2$ with $R_1$ along the admissible match ${\color{h2color}\mu_{21}}\in \Match{R_1}{R_2}$, the construction of the diagram in~\eqref{eq:concurrencyDPOproof} provides two admissible matches %
${\color{h1color}(m_1:I_1\rightarrow X_0)}\in \Match{r_1}{X_0}$ and ${\color{h1color}(m_2:I_2\rightarrow X_1)}\in \Match{r_2}{{\color{h1color}X_1}}$ with ${\color{h1color}X_1}=r_{1_{{\color{h1color}m_1}}}(X_0)$. It thus remains to verify the claim that these matches satisfy the conditions $\ac{c}_{I_1}$ and $\ac{c}_{I_2}$, respectively, which is demonstrated by running the corresponding arguments of the \emph{``analysis''} part of the proof ``in reverse''.
\end{proof}

\subsection{Associativity of DPO-type composition of rules with conditions}\label{sec:assocDPO}

We will now state one of the main results of this paper, in the form of an associativity property afforded by the DPO-type composition operation on rules with conditions. The case of DPO-type compositions of rules without conditions was studied in~\cite{bdg2016,bp2018,bp2019-ext}, and the following result is an extension to the setting of rules with conditions afforded by our refined framework for conditions as introduced in Section~\ref{sec:cond} and the current section. In contrast to the DPO-type concurrency theorem (which, in a slightly different formulation, had been previously known in the literature), the associativity result presented below is, to the best of our knowledge, the first of its kind. 

\begin{theorem}[DPO-type Associativity Theorem]\label{thm:assocAC}
Let $R_j\equiv(r_j,\ac{c}_{I_j})$ ($j=1,2,3$) be three linear rules with application conditions. Then there exists a \emph{bijection} between the pairs of admissible matches $M_A$ and $M_B$ defined as
\begin{equation}\label{eq:thmAssocRepart}
\begin{aligned}
  M_A&\eqdef\{(\mu_{21},\mu_{3(21)})\vert \mu_{21}\in \Match{R_2}{R_1}\,,\; 
  \mu_{3(21)}\in \Match{R_3}{R_{21}}\}\\
  M_B&\eqdef\{(\mu_{32},\mu_{(32)1})\vert \mu_{32}\in \Match{R_3}{R_2}\,,\; 
  \mu_{(32)1}\in \Match{R_{32}}{R_{1}}\}
\end{aligned}
\end{equation}
with $R_{i,j}\eqdef(\comp{r_i}{\mu_{ij}}{r_j},\ac{c}_{I_{ij}})$ (and $\ac{c}_{I_{ij}}$ defined as in~\eqref{eq:acIcomp}) such that 
\begin{subequations}
\begin{align}
  \forall (\mu_{21},\mu_{3(21)})\in M_A: \exists! (\mu_{32},\mu_{(32)1})&\in M_B:\nonumber\\
  \comp{r_3}{\mu_{3(21)}}{\left(
    \comp{r_2}{\mu_{21}}{r_1}
  \right)}&\cong 
  \comp{\left(
    \comp{r_3}{\mu_{32}}{r_2}
  \right)}{\mu_{(32)1}}{r_1}\label{eq:thmAAC1}\\
  \land \quad\ac{c}_{I_{3(21)}}&\,\dot{\equiv}\, \ac{c}_{I_{(32)1}}\label{eq:thmAAC2}
\end{align}
\end{subequations}
and vice versa. In this particular sense, the operation $\comp{.}{.}{.}$ is \textbf{\emph{associative}}.
\end{theorem}
\begin{proof}
Considering first the case of compositions of ``plain'' rules $r_1,r_2,r_3\in \Lin{\bfC}$, i.e.\ of rules without application conditions, we quote from~\cite{bp2019-ext} (Theorem~2.9) an isomorphism of pairs of admissible matches of the form
\begin{equation}
\begin{aligned}
  M_A'&\eqdef\{(\mu_{21},\mu_{3(21)})\vert \mu_{21}\in \Match{r_2}{r_1}\,,\; 
  \mu_{3(21)}\in \Match{r_3}{r_{21}}\}\\
  \cong \quad 
  M_B'&\eqdef\{(\mu_{32},\mu_{(32)1})\vert \mu_{32}\in \Match{r_3}{r_2}\,,\; 
  \mu_{(32)1}\in \Match{r_{32}}{r_1}\}\,.
\end{aligned}
\end{equation}
The isomorphism entails that for each corresponding pair, equation~\eqref{eq:thmAAC1} is verified. Consider then two such isomorphic pairs 
\[
(\mu_{21},\mu_{3(21)})\in M_A'\quad \text{and} \quad (\mu_{32},\mu_{(32)1})\in M_B'
\]
of admissible matches of ``plain'' rules. The isomorphism entails in particular that the following commutative diagram is uniquely constructible (where we also draw the positions of the various application conditions at play for later convenience):
\begin{equation}\label{eq:AsProofDiag}
  \vcenter{\hbox{\includegraphics[width=0.9\textwidth]{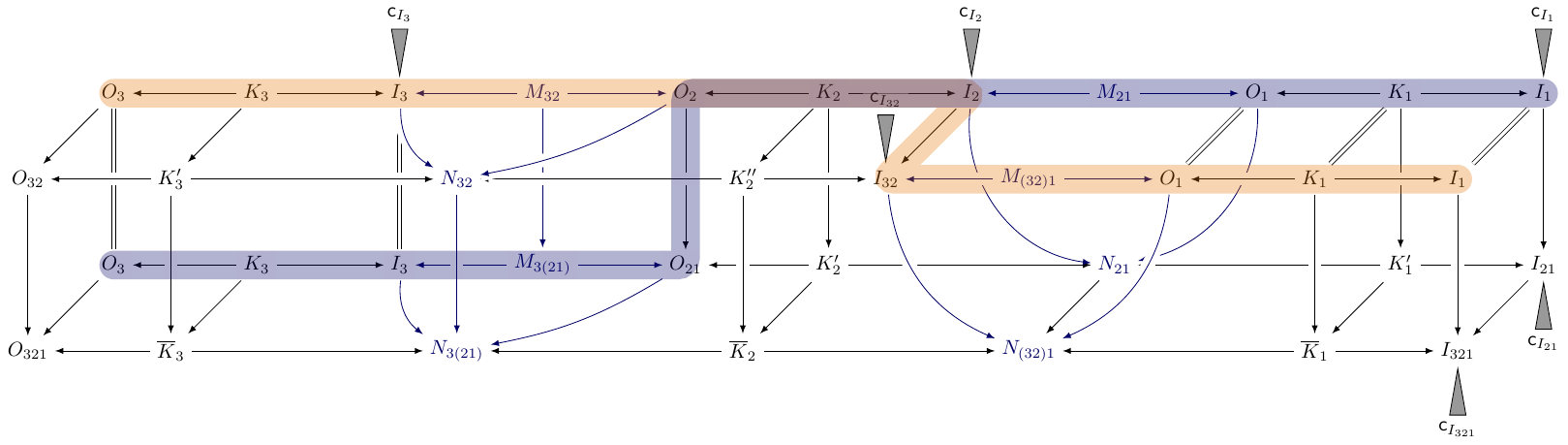}}}
\end{equation}
In order to verify the validity of~\eqref{eq:thmAAC2}, it is sufficient to utilize our various technical lemmas pertaining to the $\Shift$ and $\Trans$ constructions, and to follow the ``paths'' in the diagram depicted in~\eqref{eq:AsProofDiag} along which the three conditions $\ac{c}_{I_j}$ for $j=1,2,3$ have to be shifted and transported in order to form the conditions $\ac{c}_{I_{3(21)}}$ and $\ac{c}_{I_{(32)1}}$, respectively.
\begin{romanenumerate}
\item contribution of $\ac{c}_{I_1}$:
\begin{equation}\label{eq:AssocProofEqI1}
 \Shift\left(I_{321}\leftarrow I_{21},\Shift\left(I_{21}\leftarrow I_1,\ac{c}_{I_1}\right)\right)
\overset{Lem.~\ref{lem:shiftComp}}{\equiv}
\Shift\left(I_{321}\leftarrow I_1,\ac{c}_{I_1}\right)
 \end{equation}
\item contribution of $\ac{c}_{I_2}$:
\begin{equation}\label{eq:AssocProofEqI2}
\begin{aligned}
  &\Shift\left(I_{321}\leftarrow I_{21},
    \Trans\left({\color{h1color}N_{21}\leftarrow} K_1'\rightarrow I_{21},
    \Shift\left({\color{h1color}N_{21}\leftarrow} I_2,\ac{c}_{I_2}\right)\right)\right)\\
  &\overset{Lem.~\ref{lem:ST}}{\dot{\equiv}}
  \Trans\left({\color{h1color}N_{(32)1}\leftarrow} \overline{K}_1\rightarrow I_{321},
    \Shift\left({\color{h1color}N_{(32)1}\leftarrow N_{21}},
    \Shift\left({\color{h1color}N_{21}\leftarrow} I_2,\ac{c}_{I_2}\right)\right)\right)\\
  &\overset{Lem.~\ref{lem:shiftComp}}{\equiv}
  \Trans\left({\color{h1color}N_{(32)1}\leftarrow} \overline{K}_1\rightarrow I_{321},
    \Shift\left({\color{h1color}N_{(32)1}\leftarrow} I_{32},
    \Shift\left(I_{32}\leftarrow I_2,\ac{c}_{I_2}\right)\right)\right)\,.
\end{aligned}
\end{equation}
Here, in the last step, we have made use of the commutativity of the diagram~\eqref{eq:AsProofDiag}, whereby
\begin{equation}
\left({\color{h1color}N_{(32)1}\leftarrow N_{21}}\right)\circ
    \left({\color{h1color}N_{21}\leftarrow} I_2\right)
=
\left({\color{h1color}N_{(32)1}\leftarrow} I_{32}\right)\circ
    \left(I_{32}\leftarrow I_2\right)\,.
\end{equation}
\item contribution of $\ac{c}_{I_3}$:
\begin{equation}\label{eq:AssocProofEqI3}
\begin{aligned}
  &\Trans\bigg({\color{h1color}N_{(32)1}}\leftarrow\overline{K}_1\rightarrow I_{321},
    \Trans\big({\color{h1color}N_{3(21)}}\leftarrow\overline{K}_2\rightarrow{\color{h1color}N_{(32)1}},\\
    &\qquad \qquad \qquad \qquad
    \Shift\big({\color{h1color} N_{3(21)}}\leftarrow I_3,\ac{c}_{I_3}\big)\big)\bigg)\\
    &\overset{Lem.~\ref{lem:shiftComp}}{\equiv}
    \Trans\bigg({\color{h1color}N_{(32)1}}\leftarrow\overline{K}_1\rightarrow I_{321},
    \Trans\big({\color{h1color}N_{3(21)}}\leftarrow\overline{K}_2\rightarrow
    {\color{h1color}N_{(32)1}},\\
    &\qquad \qquad \qquad \qquad
    \Shift\big({\color{h1color}N_{3(21)}\leftarrow N_{32}},
    \Shift\big({\color{h1color}N_{32}\leftarrow }I_3,
    \ac{c}_{I_3}\big)\big)\big)\bigg)\\
     &\overset{Lem.~\ref{lem:ST}}{\dot{\equiv}}
    \Trans\bigg({\color{h1color}N_{(32)1}}\leftarrow\overline{K}_1\rightarrow I_{321},
    \Shift\big({\color{h1color}N_{(32)1}}\leftarrow I_{32},\\
    &\qquad \qquad \qquad \qquad
    \Trans\big({\color{h1color}N_{32}}\leftarrow K_2''\rightarrow I_{32},
    \Shift\big({\color{h1color}N_{32}\leftarrow } I_3,\ac{c}_{I_3}\big)\big)\big)\bigg)\\
\end{aligned}
\end{equation}
\end{romanenumerate}
It is then easy to verify (using the definition of concurrent compositions of rules with application conditions with rules according to Definition~\ref{def:compACrefined}  and eq.~\eqref{eq:acIcomp}) that since
\begin{equation}
  \ac{c}_{I_{3(21)}}= lhs\eqref{eq:AssocProofEqI1} \bigwedge lhs\eqref{eq:AssocProofEqI2} 
   \bigwedge lhs\eqref{eq:AssocProofEqI3}\,,\qquad
   \ac{c}_{I_{(32)1}}= rhs\eqref{eq:AssocProofEqI1} \bigwedge rhs\eqref{eq:AssocProofEqI2} 
   \bigwedge rhs\eqref{eq:AssocProofEqI3} \,,
\end{equation}
where by ``rhs'' we mean the very last equality in each set of equations, we find indeed that
\begin{equation}
 \ac{c}_{I_{3(21)}}\dot{\equiv}\ac{c}_{I_{(32)1}}\,.
\end{equation}
\end{proof}

\begin{example}\label{ex:assocDPO} In order to illustrate the notion of associativity more intuitively, we provide a concrete example of a triple rule composition as depicted in Figure~\ref{fig:exAssoc}, with the application conditions of the three rules given as
\[
	\ac{c}_{I_1}:=\neg\exists\left(
	\inputtikz{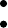}
	 \hookrightarrow
	\inputtikz{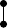}, \ac{true}\right)\,, \; \ac{c}_{I_2}:= \ac{true}\,,\; \ac{c}_{I_3}:=\ac{c}_{I_1}\,.
\]
Note that the application conditions as specified guarantee that  applying the rules to \emph{simple undirected graphs} preserves the constraint of no two vertices being linked by more than one edge. Rather than depicting the relevant data for this computation in the form of a ``commutative tube'' as in~\eqref{eq:AsProofDiag}, we opt here to ``unfold'' the two pair-wise sequential rule composition steps for each branch of the equivalence
\[
	r_{3(21)}\cong r_{(32)1} \quad \land \quad \ac{c}_{I_{3(21)}}\dot{\equiv}\,\ac{c}_{I_{(32)1}}\,,
\]
with the two sides of the equivalences depicted on the top and bottom parts of Figure~\ref{fig:exAssoc}, respectively. The example allows us to highlight a number of interesting phenomena:
\begin{itemize}
\item Upon closer inspection, one may verify that the application condition $\ac{c}_{I_1}$ of the rule $R_1$ is not specified in ``compressed standard form'' in the sense of Definition~\ref{def:csfc}, since in fact $\ac{c}_{I_1}\dot{\equiv}\, \ac{true}$. Intuitively, under DPO-semantics the application condition $\ac{c}_{I_1}$ does not constrain the admissible matches, since in cases where the two vertices in $I_1$ would be matched against two vertices in a host graph linked by an edge, the rule $R_1$ would not be applicable (since edges cannot be deleted implicitly in DPO-rewriting).
\item Utilizing the $\Shift$ and $\Trans$ constructions, followed by applying the ``compression'' according to Definition~\ref{def:csfc}, the application condition of the composite rule can be computed as follows:
\[
	{\color{h1color}\ac{c}_{I_{21}}} \dot{\equiv}\, \ac{true}\,,\; 
	{\color{h1color}\ac{c}_{I_{32}}} \dot{\equiv}\, \neg \exists\left(
	\inputtikz{TVEGB}
	\hookrightarrow 
	\inputtikz{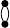}, \ac{true}\right)\,,\;
	{\color{h2color}\ac{c}_{I_{3(21)}}}\dot{\equiv}\, {\color{h2color}\ac{c}_{I_{(32)1}}} \dot{\equiv}\,\ac{true}\,.
\]
This is a typical illustration of the way associativity in sequential compositions gives rise to \emph{causal relationships}: while the rule $R_3$ can only be applied to vertices not linked by an edge, the fact that in the triple composition depicted the rule $R_1$ is in a sense ``providing'' the two vertices on its output that are later acted upon by $R_3$, and since by the application of $R_2$ as shown the edge between the two vertices on the output of $R_1$ is removed, it is indeed the case that we find a graph pattern that $R_3$ can be applied to. 
\item Finally, associativity also entails that one may equivalently lead the preceding causality argument from the viewpoint depicted in the bottom part of Figure~\ref{fig:exAssoc}: composing rule $R_3$ with rule $R_2$ as demonstrated, the composite rule has an application condition  ${\color{h1color}\ac{c}_{I_{32}}}$ that demands that the two vertices on the input of the composite rule are not linked to each other by two edges. But then if the input pattern of the composite rule (two vertices linked by an edge) is ``provided'' by the output of $R_1$ as depicted (i.e.\ via computing the composition ${\color{h1color}\comp{r_3}{\mu_{32}}{r_2}}$ and the composite rule's application condition ${\color{h1color}\ac{c}_{I_{32}}}$), and since $R_1$ can only be applied at two vertices not linked to each other, one finds again that the overall triple composite possesses an application condition that is equivalent to a trivial condition $\ac{true}$.
\end{itemize}
\end{example}

We refer the interested readers to~\cite{behr2019tracelets} for further examples of such types of sequential rule compositions and associativity properties, which play a quintessential role in the theory of \emph{tracelets}, where the latter are a form of invariant encoding of length $n$ sequences of rule compositions for general $n\geq1$.  Associativity is moreover at the core of the construction of \emph{rule algebra theories}~\cite{bdg2016,bp2018,bp2019-ext,nbSqPO2019}, and with the case of rule algebras for linear rules with conditions recently covered in~\cite{BK2020}, yet a full discussion of this novel mathematical concept is out of the scope of the present paper.

\begin{figure}[htp]
\centering
\includegraphics[width=0.9\textwidth]{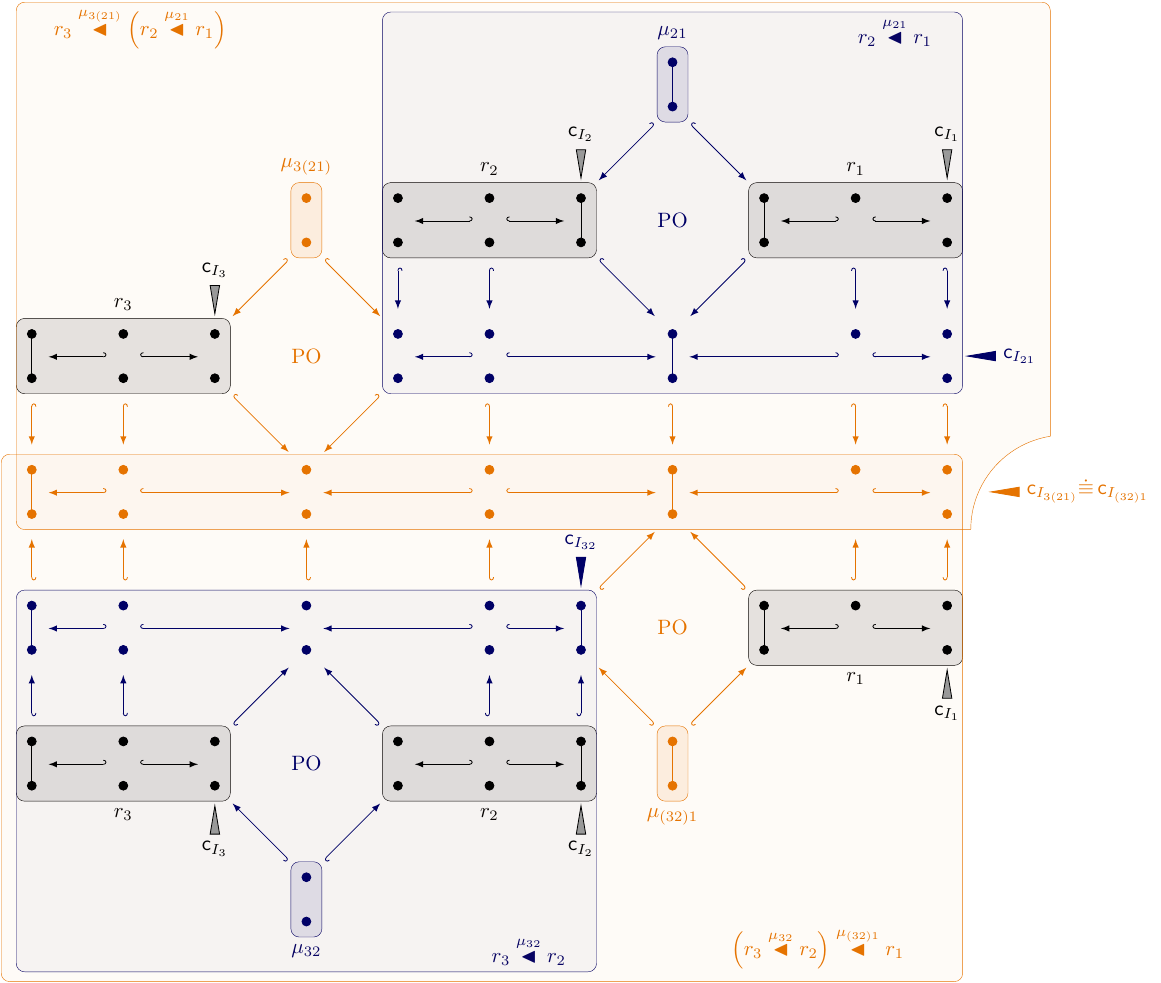}
\caption{\label{fig:exAssoc}Illustration of the notion of associativity in a sequential composition of three rules with conditions. The squares explicitly marked with $\text{PO}$ are pushouts of the admissible matches of rules; all other squares are obtained according to the semantics of DPO rule compositions (and are thus also pushout squares). The top portion of the figure illustrates the order of composition ${\color{h1color}\comp{r_2}{\mu_{21}}{r_1}}$ followed by a composition with $r_3$ along the admissible match ${\color{h2color}\mu_{3(21)}}$. The bottom portion depicts the composition ${\color{h1color}\comp{r_3}{\mu_{32}}{r_2}}$ (with ${\color{h1color}\mu_{32}}$ computed via taking pullback as described in detail in Theorem~\ref{thm:assocAC}), precomposed with $r_1$ along the induced match ${\color{h2color}\mu_{(32)1}}$. Conversely, one could start from providing explicitly the pair of admissible matches $({\color{h1color}\mu_{32}},{\color{h2color}\mu_{(32)1}})$ and compute from this data the pair of admissible matches $({\color{h1color}\mu_{21}},{\color{h2color}\mu_{3(21)}})$. Here, the composite rules' application conditions are indicated at the input interfaces of the composite rules. The second part of the statement of associativity of rule compositions entails that in both orders of pair-wise sequential compositions along the matches as specified, the resulting ``triple composites'' are isomorphic on their plain rule parts and have equivalent application conditions (up to satisfiability of admissible matches, i.e.\ in the sense of $\dot{\equiv}$).}
\end{figure}

\section{Compositional associative Sesqui-Pushout rewriting with conditions}\label{sec:SqPO}

While the previously discussed notion of compositional DPO-type rewriting may be seen as a refinement of pre-existing notions of DPO-rewriting from the literature (apart from the associativity theorem), the corresponding construction for a framework of SqPO-rewriting is almost entirely new. The first framework for a compositional SqPO-rewriting framework for rules without conditions was introduced in~\cite{nbSqPO2019}. The essential technical step in order to extend our framework from DPO- to SqPO-type semantics consists in analyzing the interplay of final pullback complements with application conditions and transformations thereof. We will first provide a brief introduction to this type of rewriting, and then develop our new framework.

\subsection{Definition of SqPO rule applications and rule compositions}\label{sec:SqPOr}

\begin{definition}[compare \cite{Corradini_2006}, Def.~4]
\label{def:SqPOr}
Given an object $X\in \obj{\bfC}$ and a linear rule $r\in \Lin{\bfC}$, we denote the \emph{set of SqPO-admissible matches} $\sqMatch{r}{X}$ as 
\begin{equation}
  \sqMatch{r}{X}\eqdef\{(m:I\rightarrow X)\in\cM\}\,.
\end{equation}
Let $m\in \sqMatch{r}{X}$. Then the diagram below is constructed by taking the final pullback complement marked $\mathsf{FPC}$ followed by taking the pushout marked $\mathsf{PO}$:
\begin{equation}\label{eq:DPOr}
\inputtikz{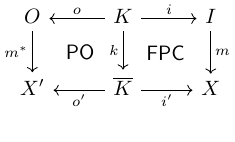}
\end{equation}
We write $r_m(X)\eqdef X'$ for the object ``produced'' by the above diagram. The process is called \emph{(SqPO)-derivation} of $X$ along rule $r$ and admissible match $m$, and denoted $r_m(X)\xLeftarrow[r,m]{{\tiny SqPO}} X$.
\end{definition}

Notably, SqPO-type rewriting thus differs from DPO-type rewriting in the important aspect that final pullback complements as well as pushouts are guaranteed to exist in our base category (cf.\ Assumption~\ref{as:SqPO}), whereas pushout complements may fail to exist in general. A typical example already mentioned in the introduction concerns the application of a vertex-deletion rule to a graph that consists of two vertices linked by an edge:
\begin{equation*}
  \vcenter{\hbox{\includegraphics{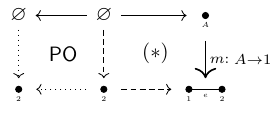}}}
\end{equation*}
In DPO-type rewriting, the deletion rule is not applicable along the match presented, since the square $(*)$ is not constructible as a pushout complement. In SqPO-type rewriting however, since the square $(*)$ is constructible as a final pullback complement as presented, the deletion rule is applicable, resulting in a graph with just a single vertex. This example demonstrates the distinguishing feature of SqPO-type rewriting over DPO-rewriting, in that the former admits ``deletion in unknown context'' (here the implicit deletion of the edge via application of the vertex deletion rule).

It should be noted that one of the additional distinctive features of SqPO-rewriting (see~\cite{Corradini_2006}) consists in the possibility to faithfully model \emph{cloning} of structures via considering \emph{non-linear rules} (i.e.\ rules based upon non-monic spans). While an interesting topic in its own right, we found that our proofs for key properties such as the concurrency and associativity theorems in the forms presented here would not easily carry over to this more general setting, which is why we did not consider the case of non-linear rules in any further detail.

\subsection{The transport construction in the SqPO-type setting}\label{sec:SqPOtransp}

The following theorem demonstrates that the construction $\Trans$ as introduced in Definition~\ref{def:Trans} is precisely the construction needed in order to implement ``transporting'' conditions over linear rules also in the SqPO-rewriting setting. This quintessential result appears to be new.

\begin{theorem}[Transport construction in SqPO rewriting]\label{thm:SqPOtrans}
Let $\bfC$ be an $\cM$-adhesive category satisfying Assumption~\ref{as:SqPO}. Then the \emph{transport construction} $\Trans$ satisfies the following property: for every linear rule $r\equiv(O\xleftarrow{o} K\xrightarrow{i} I)\in \Lin{\bfC}$, for every SqPO-admissible match $(m:I\rightarrow X)\in \sqMatch{r}{X}$ of $r$ into an object $X\in \obj{C}$ and for every application condition $\ac{c}_O$ over $O$, in the commutative diagram below
\begin{equation}
\inputtikz{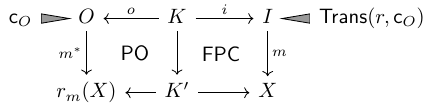}\,,
\end{equation}
it holds that
\begin{equation}
m^{*}\vDash \ac{c}_O\quad \Leftrightarrow \quad m\vDash \Trans(r,\ac{c}_O)\,.
\end{equation}
\end{theorem}
\begin{proof}
 Due to the recursive nature of the definition of $\Trans$ (Definition~\ref{def:Trans}), it suffices to verify the claim on conditions of the form %
 $\ac{c}_O=\exists({\color{h2color}b^{*}:O\hookrightarrow B},{\color{h2color}\ac{c}_B})$ (for ${\color{h2color}b^{*}}\in \cM$).

\paragraph{``$\Rightarrow$'' direction:} Suppose that we are given an SqPO-type rewriting step and a condition of the form %
$\ac{c}_O=\exists({\color{h2color}b^{*}:O\hookrightarrow B},{\color{h2color}\ac{c}_B})$, which by definition of satisfiability of conditions according to Definition~\ref{def:ac} furnishes a morphism ${\color{h2color}(n^{*}:B\hookrightarrow r_m(X))}\in\cM$ such that %
${\color{h2color}n^{*}}\circ {\color{h2color}b^{*}}={\color{h1color}m^{*}}$:
\begin{equation}
\vcenter{\hbox{\includegraphics[scale=1.0,page=1]{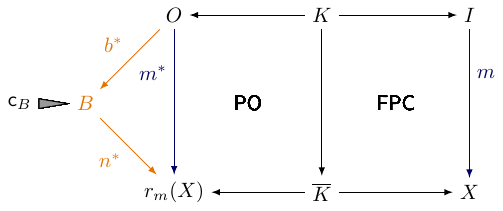}}}
\end{equation}
We then construct the commutative diagram below,
\begin{equation}\label{eq:TransRproofAuxA}
\vcenter{\hbox{\includegraphics[scale=1.0,page=2]{images/transSqPOproof.pdf}}}\,,
\end{equation}
with the following individual steps taken:
\begin{romanenumerate}
\item The square ${\color{h2color}(1)}$ is formed by taking pullback, %
which by the universal property of pullbacks furnishes a morphism $K{\color{h2color}\rightarrow K'}$. %
By stability of $\cM$-morphisms under pullback, the morphisms of the induced span are in $\cM$. %
Since thus in particular ${\color{h2color}K'\rightarrow}\overline{K}$ is in $\cM$, and %
since by stability of $\cM$-morphisms under pullback (and thus under FPC) also $(K\rightarrow \overline{K})\in\cM$, %
we conclude by $\cM$-morphism decomposition that $K{\color{h2color}\rightarrow K'}$ is in $\cM$. %
Moreover, by virtue of pushout-pullback decomposition (Lemma~\ref{lem:Main}\ref{lem:POPBdec}), %
the squares ${\color{h2color}(1)}$ and ${\color{h2color}(2)}$ are both pushouts.
\item The square ${\color{h2color}(3)}$ is formed by taking pushout. %
The universal property of pushouts provides a morphism ${\color{h2color}n:C\rightarrow} X$ such that ${\color{h2color}n}\circ {\color{h2color}b}={\color{h1color}m}$. %
Then invoking vertical FPC-pushout decomposition (Lemma~\ref{lem:Main}\ref{lem:vertFPCpoDec}), %
the resulting square $(4)$ is an FPC and ${\color{h2color}n}\in \cM$.
\end{romanenumerate}

If ${\color{h2color}\ac{c}_B}\equiv {\color{h2color}\ac{true}}$, we have thus demonstrated by virtue of the definition of satisfiability of conditions %
that the $\cM$-morphism ${\color{h2color}n}$ satisfies the condition $\exists({\color{h2color}b:I\hookrightarrow C},{\color{h2color}\Trans(r',\ac{true})})$, since by definition of $\Trans$ we have that ${\color{h2color}\Trans(r',\ac{true})}\equiv{\color{h2color}\ac{true}}$. %
If ${\color{h2color}\ac{c}_B}$ is itself nested, we proceed by induction in the evident fashion, i.e.\ repeat the preceding argument for ${\color{h2color}\ac{c}_B}$ and ${\color{h2color}r'}:={\color{h2color}(B\leftarrow K'\rightarrow C)}$, thus the claim of the ``$\Rightarrow$'' part of the theorem statement follows.

\paragraph{``$\Leftarrow$'' direction} Suppose that we are given the following part of data presented in diagram~\eqref{eq:TransRproofAuxA}:
\begin{romanenumerate}
\item The squares ${\color{h2color}(1)+(2)}$ (a pushout), ${\color{h2color}(2)}$ (a pushout), %
${\color{h2color}(3)+(4)}$ (an FPC) and ${\color{h2color}(3)}$ (a pushout).
\item An arrow ${\color{h2color}n}\in \cM$ that satisfies the condition $\Trans({\color{h2color}r'},{\color{h2color}\ac{c}_B})$, and with ${\color{h2color}n}\circ{\color{h2color}b}={\color{h1color}m}$.
\end{romanenumerate} 
What is somewhat hidden in this set of data is the fact that the square $\cSquare{K,{\color{h2color}K'},X,I}$ is a pullback %
(where the $\cM$-morphism ${\color{h2color}K'}\rightarrow X$ is provided as the composition of the $\cM$-morphisms %
${\color{h2color}K'\rightarrow C}$ and ${\color{h2color}n:C\rightarrow} X$, which are by assumption part of the above data). %
This statement can be verified by constructing the commutative cube below left:  
\begin{equation}\label{eq:SqPOtransProofB}
\vcenter{\hbox{\includegraphics[scale=1]{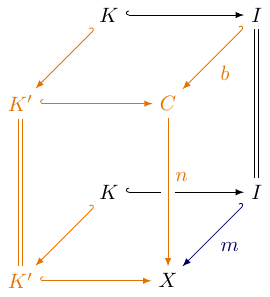}}}
\qquad\qquad 
\inputtikz{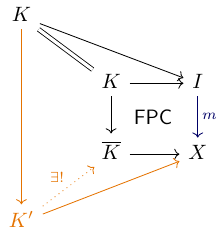}
\end{equation}
By virtue of Lemma~\ref{lem:Main}\ref{lem:idPB}, the right square is a pullback. Since the top square is a pushout and thus also a pullback, by pullback composition the square $\cSquare{K,{\color{h2color}K'},X,I}$ (i.e.\ the composite of the top and right squares) is indeed a pullback.

We then invoke the universal property of FPCs (compare Definition~\ref{def:FPC}) in the form %
presented in the right part of~\eqref{eq:SqPOtransProofB}, which entails that %
there exists a morphism ${\color{h2color}K'\rightarrow} \overline{K}$. %
Then by $\cM$-morphism decomposition, as ${\color{h2color}K'\rightarrow} X$ and %
$\overline{K}\rightarrow X$ are in $\cM$, so is the morphism ${\color{h2color}K'\rightarrow }\overline{K}$. %
Consequently, vertical FPC-pushout decomposition (Lemma~\ref{lem:Main}\ref{lem:vertFPCpoDec}) implies that %
the square ${\color{h2color}(4)}$ in~\eqref{eq:TransRproofAuxA} is an FPC, while %
due to pushout-pushout decomposition (Lemma~\ref{lem:Main}\ref{lem:POPOdec}) %
the square ${\color{h2color}(1)}$ is verified to be a pushout. %
The latter result entails in particular, by virtue of stability of $\cM$-morphisms under pushout, that  the induced morphism ${\color{h2color}n^{*}}$ is in $\cM$.

In summary, if ${\color{h2color}\ac{c}_B}\equiv \ac{true}$, we have verified that ${\color{h2color}n^{*}}\vDash {\color{h2color}\ac{c}_B}$, and thus that %
${\color{h1color}m^{*}}\vDash \exists({\color{h2color}b^{*}:O\rightarrow B},{\color{h2color}\ac{c}_B})$. %
If ${\color{h2color}\ac{c}_B}$ is itself nested, we can prove the statement inductively in the evident fashion.
\end{proof}

The detailed structure of the above proof allows us to clarify that, while the $\Trans$ construction from the DPO rewriting setting carries over to the SqPO setting seemingly verbatim, the precise reasons for why it indeed furnishes an operation of transport of conditions in the desired sense in the two settings are dependent on the semantics (i.e.\ in particular why satisfaction of conditions is ``transported'' against the direction of SqPO linear rules as detailed above). Moreover, since there does not exist a construction that would permit to transport conditions ``with'' the direction of the linear rules (i.e.\ in the direction of rule applications) in SqPO rewriting\footnote{Consider for example the case of a rule $R=(r,\ac{c}_I)$ with ``plain'' rule $r=(\bullet_1\leftarrow \bullet_1\rightarrow \bullet_1\;\bullet_2)$ and application condition $\ac{c}_I =\nexists(\bullet_1\;\bullet_2 \hookrightarrow \bullet_1\!\!-\!\!\bullet_2)$. Applying the ``plain'' rule to a graph with two vertices linked by an arbitrary number of edges would result in a single-vertex graph under SqPO-semantics regardless of the number of edges (which are all implicitly deleted), yet only the case without edges permits admissible matches of $R$, which provides a counter-example to the hypothesis that one could formulate a post-condition ``transport-equivalent'' to $\ac{c}_I$.}, the only degree of freedom in describing linear rules with application conditions in SqPO semantics is to transport any conditions a rule might carry on its output to its input, motivating the following definition:

\begin{definition}[Standard form for SqPO-type linear rules with application conditions and for admissble matches]\label{def:SFsqpo}. Let $\bfC$ be an $\cM$-adhesive category satisfying Assumption~\ref{as:SqPO}. Let $\LinAc{\bfC}$ denote the \emph{set of linear rules with application conditions in standard form} as introduced in Definition~\ref{def:SF}, whence elements of $R\in\LinAc{\bfC}$ are of the form
  \begin{equation}
    R= (r,\ac{c}_I)\,,\quad r=\rSpan{O}{o}{K}{i}{I}\in\Lin{\bfC}\,.
  \end{equation}
We then introduce the notion of \emph{SqPO-admissible matches} for applications of rules with application conditions to objects under SqPO-type semantics as follows: let $X\in \obj{\bfC}$ be an object, $R\in \LinAc{\bfC}$ as above a rule with application conditions, and $(m:I\hookrightarrow X)$ an element of $\cM$. Since according to Assumption~\ref{as:SqPO} FPCs of arbitrary pairs of composable $\cM$-morphisms exist, the diagram below is always constructible,
\begin{equation}\label{eq:RappSqPO}
\inputtikz{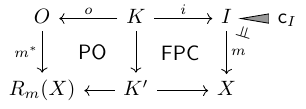}\,,
\end{equation}
the \emph{SqPO-admissibility} of $m$ hinges solely on whether the match satisfies the condition $\ac{c}_{I}$ of the rule $R$. We thus define the \emph{set of SqPO-admissible matches} for the application of the rule $R$ to the object $X$ as
\begin{equation}\label{eq:defAdmRacSqPO}
  \sqMatch{R}{X}\eqdef\{(m:I\rightarrow X)\in \cM\mid m\vDash \ac{c}_I\}\,.
\end{equation} 
\end{definition}

A beneficial side-effect of the former observation is the following result, most statements of which carry over from the DPO-rewriting setting for the $\Trans$ construction:
\begin{theorem}[Properties of $\Trans$ in SqPO-rewriting]\label{thm:TrnasSqPOprops}
Given an $\cM$-adhesive category $\bfC$ satisfying Assumption~\ref{as:SqPO}, the transport construction $\Trans$ has the following properties:
\begin{romanenumerate}
\item \emph{Units for $\Trans$:} Let $X\in \obj{\bfC}$ be an arbitrary object and $\ac{c}_X$ a condition over $X$. Then with $r_{id_X}\equiv\rSpan{X}{id_X}{X}{id_X}{X}\in \Lin{\bfC}$ the ``identity rule on $X$'', we find that
    \begin{equation}
        \Trans(r_{id_X},\ac{c}_X)\equiv \ac{c}_X\,.
    \end{equation} 
\item \emph{Compositionality of $\Trans$:} Given two composable spans of $\cM$-morphisms
\[
  r\equiv \rSpan{C}{b}{B}{a}{A}\; \quad \text{and }\quad
  s\equiv\rSpan{E}{d}{D}{c}{C}\,,
\]
we find that
  \begin{equation}
    \Trans(r,\Trans(s,\ac{c}_E))\equiv \Trans(s\circ r,\ac{c}_E)\,.
  \end{equation}
\item \emph{Compatibility of $\Shift$ and $\Trans$:}
\begin{equation}\label{eq:transR}
\inputtikz{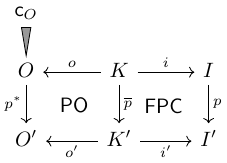}\,,
\end{equation}
letting $r=\rSpan{O}{o}{K}{i}{I}$ and $r'=\rSpan{O'}{o'}{K'}{i'}{I'}$, we have that for all objects $X$ and for all admissible matches $n\in \sqMatch{r'}{X}$ of $r'$ into $X$,
\begin{equation}
n\vDash \Shift(p,\Trans(r,\ac{c}_{O}))
  \Leftrightarrow n\vDash \Trans(r',\Shift(p^{*},\ac{c}_{O}))\,,
\end{equation}
which we write more compactly as\footnote{Note that unlike in the DPO rewriting case, since by Assumption~\ref{as:SqPO} FPCs of arbitrary composable pairs of $\cM$-morphisms exist, it is indeed the case that \emph{any} $\cM$-morphism $(n:I'\rightarrow X)\in \cM$ is an admissible match for the linear rule without application conditions $r'$; consequently, the equivalence takes precisely the form of standard equivalence of application conditions without a constraint of admissibility.}
\begin{equation}
\Shift(p,\Trans(r,\ac{c}_{O}))\,\equiv\,\Trans(r',\Shift(p^{*},\ac{c}_{O}))\,.
\end{equation}
\end{romanenumerate}
\end{theorem}
\begin{proof}
  The proofs of the first two statements take precisely the same shape as the corresponding proofs of Lemma~\ref{lem:transUnit} (units for $\Trans$) and Lemma~\ref{lem:transComp} (compositionality of $\Trans$) for the DPO-rewriting setting, as they are independent of the underlying type of rewriting. It remains to prove the third statement, which follows by suitably adapting the  proof strategy of Lemma~\ref{lem:ST} (compatibility of $\Shift$ and $\Trans$ in DPO-rewriting). More precisely, construct the commutative diagram below, where the top part is inserted from the assumption of the lemma, and where the bottom part constitutes an SqPO-type rewrite step of applying the rule $r'$ to $X$ along the match $n$:
\begin{equation}\label{eq:TransProofASqPO}
\inputtikz{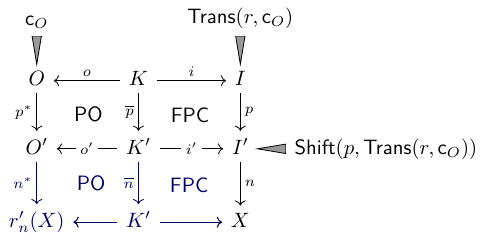}
\end{equation}
From hereon, the proof structure is fully analogous to the DPO rewriting case: since 
\[
  n\vDash(\Shift(p,\Trans(r,\ac{c}_O))\,,
\]
we find that $m=n\circ p\vDash\Trans(r,\ac{c}_O)$. %
Since the top and bottom left squares compose into a pushout and %
the top and bottom right squares into an FPC,  $m\vDash\Trans(r,\ac{c}_O)$ implies that %
${\color{h1color}m^{*}}={\color{h1color}n^{*}}\circ p^{*}\vDash \ac{c}_O$. %
Since ${\color{h1color}m^{*}}={\color{h1color}n^{*}}\circ p^{*}$, %
${\color{h1color}m^{*}}\vDash \ac{c}_O$ implies that ${\color{h1color}n^{*}}\vDash \Shift(p^{*},\ac{c}_O)$, and %
since the bottom left and right squares are of the form of an SqPO-type rewriting step, we find that indeed %
$n\vDash\Trans(r',\Shift(p^{*},\ac{c}_O)$. The proof of the converse direction follows by reversing the order of the preceding steps of the proof.
\end{proof}

\subsection{SqPO-type concurrent composition of rules with conditions}\label{sec:SqPOconcurComp}

The second central definition for our rewriting framework is the following notion of concurrent composition, which is an extension of a construction first introduced in~\cite{nbSqPO2019} to the setting of rules with conditions.

\begin{definition}[SqPO-type concurrent composition]\label{def:SqPO}
Let $\bfC$ be an $\cM$-adhesive category satisfying Assumption~\ref{as:SqPO}. %
Let ${\color{blue}R_j\equiv (r_j,\ac{c}_{I_j})}\in \LinAc{\bfC}$ be two linear rules with application conditions ($j=1,2$), and let 
\[
{\color{h2color}\mu_{21}\equiv\rSpan{I_2}{m_{2}}{M_{21}}{m_{1}}{O_1}}
\]
be a span of $\cM$-morphisms (i.e.\ $m_{1},m_{2}\in \cM$). If the diagram below is constructible (if the pushout complement marked $\mathsf{POC}$ exists),
\begin{equation}
\inputtikz{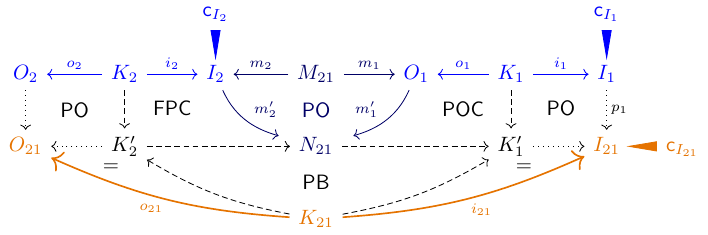}
\end{equation}
where (in close analogy to the DPO-type rewriting setting)
\begin{equation}\label{eq:acIcompSqPO}
    {\color{h2color} \ac{c}_{I_{21}}}
   \eqdef\Shift(p_1,{\color{blue}\ac{c}_{I_1}})\;\bigwedge\;
   \Trans\left({\color{h1color}N_{21}}\leftarrow K_1'\rightarrow {\color{h2color}I_{21}},
   \Shift({\color{h1color}m_2'},{\color{blue}\ac{c}_{I_2}})\right)\,,
  \end{equation}
  and if $\ac{c}_{I_{21}}\neq\ac{false}$, then we call $\mu_{21}$ an \emph{SqPO-admissible match} for the rules with conditions $R_2$ into $R_1$, denoted 
  \[
    \mu_{21}\in \sqMatch{R_2}{R_1}\,.
  \]
  In this case, we introduce the notation $\sqComp{R_2}{\mu_{21}}{R_1}$ to denote the composite,
  \begin{equation}\label{eq:defCompRCsqpo}
    {\color{h2color}\sqComp{R_2}{\mu_{21}}{R_1}}\eqdef{\color{h2color}\left(\rSpan{O_{21}}{o_{21}}{K_{21}}{i_{21}}{I_{21}},\ac{c}_{I_{21}}\right)}\,.
  \end{equation}
\end{definition}

As already noted in~\cite{nbSqPO2019} for the case of SqPO-type rewriting for rules without conditions, it might appear surprising at first sight that there is an asymmetry in the above definition (in that the left part of the diagram consists of an FPC and a pushout, while the right part is formed by a pushout complement and a pushout, respectively), but the precise reason for this definition will become apparent when considering the concurrency theorem for SqPO-type rules with conditions in the following subsection.

\subsection{SqPO-type concurrency theorem for rules with conditions}\label{sec:SqPOconcur}

To the best of our knowledge, the following key result is the first of its kind in the SqPO-rewriting setting:

\begin{theorem}[SqPO-type Concurrency Theorem, extended from \cite{nbSqPO2019}, Thm.~2.9]
Let $\bfC$ be an $\cM$-adhesive category satisfying Assumption~\ref{as:SqPO}, let $R_j\equiv(r_j,\ac{c}_{I_j})\in \LinAc{\bfC}$ ($j=1,2$) be two linear rules with application conditions, and let $X_0\in ob(\bfC)$ be an object.
\begin{itemize}
\item \textbf{Synthesis:} Given a two-step sequence of SqPO derivations 
\[
X_2\xLeftarrow[R_2,M_2]{{\tiny SqPO}} X_1\xLeftarrow[R_1,M_1]{{\tiny SqPO}}X_0\,,
\]
with $X_1\eqdef r_{1_{M_1}}(X_0)$ and $X_2\eqdef r_{2_{M_2}}(X_1)$, there exists a SqPO-composite rule $R_{21}=\sqComp{R_2}{\mu_{21}}{R_1}$
for a unique $\mu_{21}\in \sqRMatch{R_2}{R_1}$ as well as an SqPO-admissible match $n\in \sqMatch{R}{X}$ (that is unique up to isomorphisms) such that 
 \[
  R_{21_n}(X_0)\xLeftarrow[R_{21},n]{{\tiny SqPO}} X_0\qquad \text{and}\qquad R_{21_n}(X_0)\cong X_2\,.  
 \]
\item \textbf{Analysis:} Given an SqPO-admissible match $\mu_{21}\in \sqRMatch{R_2}{R_1}$ of $R_2$ into $R_1$ and an SqPO-admissible match $n\in \sqMatch{R_{21}}{X}$ of the SqPO-composite $R_{21}=\sqComp{R_2}{\mu_{21}}{R_1}$ into $X_0$, there exists a pair of SqPO-admissible matches $m_1\in \sqMatch{R_1}{X_0}$ and $m_2\in \sqMatch{R_2}{X_1}$ (with $X_1\eqdef R_{1_{m_1}}(X_0)$, and unique up to isomorphisms) such that
\[
    X_2\xLeftarrow[R_2,m_2]{{\tiny SqPO}} X_1 \xLeftarrow[R_1,m_1]{{\tiny SqPO}} X_0\qquad \text{and}\qquad
    X_2\cong R_{21_n}(X_0)\,.
\]
\end{itemize}
\end{theorem}
\begin{proof}
  For the part of the proof pertaining to the concurrency of SqPO-type rules without application conditions, we will follow the strategy presented in~\cite{nbSqPO2019} (where slightly stronger conditions than the ones required according to Assumption~\ref{as:SqPO} were made, i.e.\ in~\cite{nbSqPO2019} $\bfC$ was assumed to be adhesive, whence the re-derivation here in the $\cM$-adhesive setting).

\paragraph{``Synthesis'' part of the proof:} %
Suppose we are given rules with application conditions $R_1,R_2\in \LinAc{\bfC}$ and %
SqPO-admissible matches ${\color{h1color}m_1}\in \sqMatch{R_1}{X_0}$ and %
${\color{h1color}m_2}\in \sqMatch{R_2}{X_1}$, with $X_1=R_{1_{{\color{h1color}m_1}}}(X_0)$. %
This data is encoded in the {\color{h1color}blue} part of the diagram in~\eqref{eq:concurrencyDPOproof}. %
Let us begin by constructing the {\color{h2color}orange} parts of~\eqref{eq:concurrencyDPOproof} as follows: %
take the pullback ${\color{h2color}M_{21}}=\pB{I_2{\color{h1color}\rightarrow X_1\leftarrow}O_1}$, and then %
the pushout ${\color{h2color}N_{21}}=\pO{I_2{\color{h2color}\leftarrow M_{21}\rightarrow}O_1}$; %
by the universal property of pushouts, there exists a morphism ${\color{h2color}N_{21}}\rightarrow {\color{h1color}X_1}$, %
which is in $\cM$ since $\bfC$ is assumed to possess $\cM$-effective unions according to Assumption~\ref{as:SqPO}. %
Next, form the pullbacks ${\color{h2color}K_i'}=\pB{{\color{h1color}\overline{K}_i\rightarrow X_1}\leftarrow {\color{h2color}N_{21}}}$ (for $i=1,2$), %
which furnishes morphisms $K_i{\color{h2color}\rightarrow K_i'}$ (for $i=1,2$) that %
are in $\cM$ due to the decomposition property of $\cM$-morphism. %
By virtue of vertical FPC-pullback decomposition (Lemma~\ref{lem:Main}\ref{lem:vertFPCpbDec}), %
the second from the left bottom and top squares  in the back of~\eqref{eq:concurrencyDPOproof} are FPCs, %
while via pushout-pullback decomposition (Lemma~\ref{lem:Main}\ref{lem:POPBdec}) %
the second from the right top and bottom squares in the back of~\eqref{eq:concurrencyDPOproof} are pushouts. %
Let ${\color{h2color}O_{21}}=\pO{O_2\leftarrow K_2{\color{h2color}\rightarrow K_2'}}$, which %
by the universal property of pushouts furnishes a morphism ${\color{h2color}O_{21}\rightarrow}X_2$; %
by pushout-pushout decomposition (Lemma~\ref{lem:Main}\ref{lem:POPOdec}), %
the leftmost bottom square in the back of~\eqref{eq:concurrencyDPOproof} is a pushout, %
which also entails by stability of $\cM$-morphisms under pushouts that $({\color{h2color}O_{21}\rightarrow}X_2)\in \cM$. %
Analogously, let ${\color{h2color}I_{21}}=\pO{{\color{h2color}K_1'\leftarrow}K_1\rightarrow I_1}$, %
which furnishes a morphism ${\color{h2color}I_{21}\rightarrow}X_0$, and %
via vertical FPC-pushout decomposition (Lemma~\ref{lem:Main}\ref{lem:vertFPCpoDec}) that %
the rightmost bottom back square in~\eqref{eq:concurrencyDPOproof} is an FPC and that $({\color{h2color}I_{21}\rightarrow}X_0)\in \cM$.

At this point, note that the bottom row of squares in the back of~\eqref{eq:concurrencyDPOproof} has the shape of two consecutive SqPO rewriting steps. Since ${\color{h1color}(m_1:I_1\rightarrow X_0)\vDash \ac{c}_{I_1}}$ and ${\color{h1color}(m_2:I_2\rightarrow X_1)\vDash \ac{c}_{I_2}}$ by assumption, we  conclude that ${\color{h2color}m_{21}=(I_{21}\rightarrow X_0)}$ satisfies ${\color{h2color}m_{21}\vDash\ac{c}_{21}}$, with
\begin{equation}\label{eq:SqPOconcurAC21}
      {\color{h2color}\ac{c}_{21}}= \Shift(I_1{\color{h2color}\rightarrow I_{21}},\ac{c}_{I_1})\; \bigwedge\;
      \Trans({\color{h2color}N_{21}\leftarrow K_1'\rightarrow I_{21}},
      \Shift(I_2{\color{h2color}\rightarrow N_{21}},\ac{c}_{I_2})\,.
\end{equation}
    
To complete this part of the proof, take the pullbacks
\[
      {\color{h2color}K_{21}}
        =\pB{{\color{h2color}K_2'\leftarrow N_{21}\rightarrow K_1'}}
        \quad \text{and}\quad
      {\color{h2color}\overline{K}_{21}}
      =\pB{{\color{h1color}\overline{K}_2\leftarrow X_1\rightarrow \overline{K}_1}}\,,
\]
which also induces a morphism ${\color{h2color}K_{21}\rightarrow \overline{K}_{21}}$. %
By pullback-pullback decomposition (Lemma~\ref{lem:Main}\ref{lem:PBPBdec}), %
the induced square $\cSquare{{\color{h2color}K_{21}},{\color{h2color}\overline{K}_{21}},{\color{h1color}\overline{K}_1},{\color{h2color}K_{1}'}}$ %
(i.e.\ the inner right ``curvy'' square) is a pullback, %
which entails by stability of $\cM$-morphisms under pullbacks that ${\color{h2color}K_{21}\rightarrow \overline{K}_{21}}$ is in $\cM$. %
Then by the $\cM$-van Kampen property, %
the square $\cSquare{{\color{h2color}K_{21}},{\color{h2color}\overline{K}_{21}},{\color{h1color}\overline{K}_2},{\color{h2color}K_{2}'}}$ %
(i.e.\ the inner left ``curvy'' square) is a pushout. %
The composition of the latter pushout (and thus an FPC) square and %
of the second bottom square from the left in the back of~\eqref{eq:concurrencyDPOproof} (an FPC) yields an FPC square, %
while the third from the left bottom back square is a pushout (and thus an FPC), %
whence, by horizontal FPC decomposition (Lemma~\ref{lem:Main}\ref{lem:horFPCdec}), %
we derive that the inner right ``curvy'' square is an FPC. %
Consequently, forming the front left ``curvy'' square as a composition of pushout squares and %
the front right ``curvy'' square as a composition of FPCs, %
we have in summary exhibited an SqPO rewriting step of $X_0$ along the rule ${\color{h2color}(O_{21}\leftarrow K_{21}\rightarrow I_{21})}$ and admissible match ${\color{h2color}m_{21}}$.

\paragraph{``Analysis'' part of the proof:} %
Suppose that we are given an SqPO-composite $R_{21}=\sqComp{R_2}{\mu_{21}}{R_1}$ of linear rules %
with application conditions $R_2$ with $R_1$ %
along the SqPO-admissible match ${\color{h2color}\mu_{21}}=(I_2{\color{h2color}\leftarrow M_{21}\rightarrow }O_1)$, and that moreover %
${\color{h2color}m_{21}}=({\color{h2color}I_{21}\rightarrow}X_0)$ is an SqPO-admissible match of $R_{21}$ into $X_0$. %
Forming a SqPO-rewrite step by applying $R_{21}$ along ${\color{h2color}m_{21}}$ to $X_0$ yields %
the {\color{h2color}orange} parts of the diagram in~\eqref{eq:concurrencyDPOproof}. %
In order to prove the claim, we have to construct two SqPO-admissible matches %
${\color{h1color}m_1}\in \sqMatch{R_1}{X_0}$ and ${\color{h1color}m_2}\in \sqMatch{R_1}{{\color{h1color}X_1}}$, %
with ${\color{h1color}X_1}=R_{1_{{\color{h1color}m_1}}}(X_0)$, %
and that $X_2\cong R_{2_{{\color{h1color}m_2}}}({\color{h1color}X_1})$ under these assumptions.

We begin by forming the FPC ${\color{h1color}\overline{K}_1}=\FPC{{\color{h2color}K_1'\rightarrow I_{21}\rightarrow}X_0}$, %
which according to Assumption~\ref{as:SqPO} is guaranteed to exist and to yield two $\cM$-morphisms %
${\color{h2color}K_1'}\rightarrow {\color{h1color}\overline{K}_1}$ and ${\color{h1color}\overline{K}_1}\rightarrow X_0$. %
By the universal property of FPCs, there exists a morphism ${\color{h2color}\overline{K}_{21}}\rightarrow {\color{h1color}\overline{K}_1}$, %
which by the $\cM$-morphism decomposition property is in $\cM$. Horizontal FPC decomposition (Lemma~\ref{lem:Main}\ref{lem:horFPCdec}) %
entails that the square $\cSquare{{\color{h2color}K_{21}},{\color{h2color}\overline{K}_{21}},{\color{h1color}\overline{K}_1},{\color{h2color}K_{1}'}}$ %
(inner right ``curvy'' square) is an FPC. %
Next, we take the pushout ${\color{h1color}X_1}=\pO{{\color{h2color}M_{21}\leftarrow K_{1}'\rightarrow}{\color{h1color} \overline{K}_{1}}}$, %
followed by forming the FPC 
\[
    {\color{h1color}\overline{K}_2}=\FPC{{\color{h2color}K_2'\rightarrow N_{21}}\rightarrow {\color{h1color}X_1}}\,.
\]
Since the inner right ``curvy'' square %
$\cSquare{{\color{h2color}K_{21}},{\color{h2color}\overline{K}_{21}},{\color{h1color}\overline{K}_1},{\color{h2color}K_{1}'}}$ is an FPC and %
the second from the right bottom back square in~\eqref{eq:concurrencyDPOproof} a pushout, %
the composition of these two squares is an FPC and thus a pullback. %
Then by the universal property of FPCs, there exists a morphism ${\color{h2color}\overline{K}_{21}}\rightarrow{\color{h1color}\overline{K}_2}$, %
which is by the $\cM$-morphism decomposition property in $\cM$, and %
by the horizontal FPC decomposition property, the left inner ``curvy'' square is an FPC. %
Noting that the square $\cSquare{{\color{h2color}K_{21}},{\color{h2color}K_{2}'},{\color{h2color}N_{21}},{\color{h2color}K_{1}'}}$ %
is by assumption a pullback, %
by pullback-pullback decomposition so is %
$\cSquare{{\color{h2color}\overline{K}_{21}},{\color{h1color}\overline{K}_{2}},{\color{h2color}X_1},{\color{h1color}\overline{K}_{1}}}$; %
thus by the $\cM$-van Kampen property,  the left inner ``curvy'' square is an FPC, noting that the square is in fact a pushout. %
The latter entails by the universal property of pushouts the existence of a morphism ${\color{h1color}\overline{K}_2}\rightarrow X_2$, %
and then by pushout-pushout decomposition that the leftmost bottom back square in~\eqref{eq:concurrencyDPOproof} %
is a pushout (and thus by stability of $\cM$-morphisms under pushouts that $({\color{h1color}\overline{K}_2}\rightarrow X_2)\in \cM$). 

To complete the proof, note that the rightmost and second from right bottom back squares in~\eqref{eq:concurrencyDPOproof} are an FPC and a pushout, respectively, thus %
${\color{h2color}m_{21}\vDash \ac{c}_{I_{21}}}$ (with ${\color{h2color}\ac{c}_{I_{21}}}$ defined as in~\eqref{eq:SqPOconcurAC21}) implies that the $\cM$-morphisms
\[
      {\color{h1color}m_1}
      ={\color{h2color}(X_0\leftarrow I_{21})
          \circ (I_{21}\leftarrow I_1)}\quad \text{and}\quad
      {\color{h1color}m_2}=({\color{h1color}X_1}\leftarrow {\color{h2color}N_{21}})
        \circ ({\color{h2color}N_{21}\leftarrow}I_2)
\]
satisfy $m_1\vDash\ac{c}_{I_1}$ and $m_2\vDash\ac{c}_{I_2}$ according to the properties of the $\Shift$ and $\Trans$ constructions. It then remains to compose pushout squares and FPC squares in the top and bottom back of the diagram~\eqref{eq:concurrencyDPOproof} in order to form the two SqPO rewriting steps claimed to exist by the statement of the theorem, which concludes the proof.
\end{proof}

\subsection{Associativity of SqPO rewriting with conditions}\label{sec:assocSqPO}

Based upon the developments presented thus far and on the central result of~\cite{nbSqPO2019} (the associativity theorem for SqPO rewriting without rules), we may now state another original contribution of this paper. Due to the structural similarities between the strategies of the associativity proof in the DPO- and SqPO-type cases for rules without application conditions, the following statement is almost verbatim equivalent to the corresponding DPO-type statement. More precisely, in both cases, the part of the proof needed on top of the case without conditions consists of a ``diagram chase'' involving the $\Shift$ and $\Trans$ constructions and their properties. 

\begin{theorem}[SqPO-type Associativity Theorem]\label{thm:assocACsqpo}
Let $\bfC$ be an $\cM$-adhesive category satisfying Assumption~\ref{as:SqPO}. 
Let $R_j\equiv(r_j,\ac{c}_{I_j})\in \LinAc{\bfC}$ ($j=1,2,3$) be three linear rules with application conditions. Then there exists a \emph{bijection} between the sets of pairs of admissible matches $M_A$ and $M_B$ defined as
\begin{equation}\label{eq:thmAssocRepartSqPO}
\begin{aligned}
  M_A&\eqdef\{(\mu_{21},\mu_{3(21)})\vert \mu_{21}\in \sqMatch{R_2}{R_1}\,,\; 
  \mu_{3(21)}\in \Match{R_3}{R_{21}}\}\\
  M_B&\eqdef\{(\mu_{32},\mu_{(32)1})\vert \mu_{32}\in \sqMatch{R_3}{R_2}\,,\; 
  \mu_{(32)1}\in \Match{R_{32}}{R_{1}}\}
\end{aligned}
\end{equation}
with $R_{i,j}\eqdef(\sqComp{r_i}{\mu_{ij}}{r_j},\ac{c}_{I_{ij}})$ (and $\ac{c}_{I_{ij}}$ defined as in~\eqref{eq:acIcomp}) such that 
\begin{subequations}
\begin{align}
  \forall (\mu_{21},\mu_{3(21)})\in M_A: \exists! (\mu_{32},\mu_{(32)1})&\in M_B:\nonumber\\
  \sqComp{r_3}{\mu_{3(21)}}{\left(
    \sqComp{r_2}{\mu_{21}}{r_1}
  \right)}&\cong 
  \sqComp{\left(
    \sqComp{r_3}{\mu_{32}}{r_2}
  \right)}{\mu_{(32)1}}{r_1}\label{eq:thmAAC1SqPO}\\
  \land \quad\ac{c}_{I_{3(21)}}&\equiv \ac{c}_{I_{(32)1}}\label{eq:thmAAC2SqPO}
\end{align}
\end{subequations}
and vice versa. In this particular sense, the operation $\sqComp{.}{.}{.}$ is \textbf{\emph{associative}}.
\end{theorem}
\begin{proof}
Referring the interested readers to~\cite{nbSqPO2019} for the precise details of the proof for the part of the above statement pertaining to rules without application conditions, and to~\cite{bp2019-ext} for the generalization of the requisite technical lemmas from the adhesive to the $\cM$-adhesive setting, suffice it here to quote the following result: the diagram (where we have also indicated the relevant conditions on rules for later convenience) is constructible starting from either of the sets of pairs of SqPO-admissible matches in~\eqref{eq:thmAssocRepartSqPO}, and thereby proves the bijection for the sets of matches of rules without conditions:
\begin{equation}\label{eq:AsProofDiagSqPO}
  \vcenter{\hbox{\includegraphics[width=0.9\textwidth]{images/AsProofDiag.pdf}}}
\end{equation}
In contrast to the analogous diagram~\eqref{eq:AsProofDiag} of the proof of the DPO-type associativity theorem, the nature of the various squares in~\eqref{eq:AsProofDiagSqPO} differs at several positions. For the purpose of proving the part of the SqPO-type associativity theorem pertaining to the conditions on rules, we thus quote from~\cite{nbSqPO2019} that
\begin{itemize}
  \item the third and fourth squares in the front (counting from the left) are a pushout and an FPC, respectively
  \item the second to right and rightmost bottom squares are pushouts.
\end{itemize}
Consequently, it suffices to replace the various applications of the DPO-type compatibility of $\Shift$ and $\Trans$ (Lemma~\ref{lem:ST}) in the proof of Theorem~\ref{thm:assocAC} with its SqPO-type variant (Theorem~\ref{thm:TrnasSqPOprops}) in order to obtain a proof of the statements of the present Theorem~\ref{thm:assocACsqpo} pertaining to the conditions on rules.
\end{proof}

\begin{example}[Ex.~\ref{ex:assocDPO} continued]\label{ex:assocSqPO}
It is instructive to compare the associativity property in the SqPO-type variant directly with its counterpart in the DPO-type setting by considering once more the triple rule composition depicted in Figure~\ref{fig:exAssoc}. Note first that since pushout complements along $\cM$-morphisms are also FPCs, the triple composition at the level of ``plain'' rules is a valid composition both in DPO- and in SqPO-semantics. However, the two semantics differ at the level of the application conditions. More precisely, since the SqPO-type $\Trans$ construction is in a sense ``asymmetric'' (in that conditions can be transported from the output to the input of a rule, but not vice versa), the ``compression'' as described in Definition~\ref{def:csfc} for the DPO-setting is not available in the SqPO-setting. Indeed, the application condition
\[
\ac{c}_{I_1}:=\neg\exists\left(
	\inputtikz{TVGB}
	 \hookrightarrow
	\inputtikz{TVEGB}, \ac{true}\right)
\]
yields a \emph{non-trivial} test on candidate matches of rule $R_1$ in the SqPO-rewriting, since an application of $R_1$ at two vertices linked by an arbitrary number of vertices is possible for the ``plain'' rule $r_1$ in SqPO-rewriting (leading to the implicit deletion of all edges incident to the deleted vertex). Applying the SqPO-type $\Trans$ construction as well as $\Shift$, one may compute that the application conditions $\ac{c}_{I_{3(21)}}$ and $\ac{c}_{I_{(32)1}}$ are again equivalent to each other, and are found to evaluate to $\ac{c}_{I_1}$. 
\end{example}

Having scratched but the very surface of the intricate matter of \emph{causality} in SqPO-type rewriting, and referring to~\cite{behr2019tracelets,BK2020} for a more in-depth discussion, suffice it here to mention that it is precisely this type of causality that is of key importance to the analysis of biochemical reaction systems in the \textsc{Kappa} framework~\cite{danos2012graphs,Murphy_2010,Boutillier:2018aa}.

\section{Conclusion and Outlook}

This paper provides a self-contained account of a class of rewriting theories that possess the special property of \emph{compositionality}. Based upon the rich theory of ``traditional'' Double-Pushout (DPO)~\cite{ehrig1973, ehrig1991parallelism, CorradiniMREHL97,DBLP:conf/gg/1997handbook,ehrig:2006aa,ehrig2008graph,ehrig2010categorical,Habel:2012aa,ehrig2014mathcal} and %
Sesqui-Pushout (SqPO)~\cite{Corradini_2006,Loewe_2015} rewriting in the setting of $\cM$-adhesive categories~\cite{lack2005adhesive,ehrig2006adhesive,Braatz:2010aa}, %
we lift our earlier results on \emph{compositional concurrency} and \emph{compositional associativity} as developed in~\cite{bdg2016,bp2018,nbSqPO2019,bp2019-ext} to the realm of rewriting systems with conditions on objects and morphisms. %

A key original contribution of the present paper is a derivation of the precise technical conditions under which compositionality in the aforementioned sense is attainable for a given rewriting theory. Concretely (cf.\ Section~\ref{sec:CTP}), it has proved essential to suitably adapt the requirements on the host categories for both DPO- and SqPO-type rewriting with application conditions on objects and %
morphisms~\cite{Pennemann:aa,habel2009correctness,ehrig2014mathcal,hermann2014analysis,GOLAS2014}, concluding that these categories should be $\cM$-adhesive and possess certain additional properties such as the existence of epi-$\cM$-factorizations and $\cM$-effective unions as described in Assumptions~\ref{as:DPO} (DPO-case) and~\ref{as:SqPO} (SqPO-case). 

Our second original contribution consists of a refinement of the theory of conditions in $\cM$-adhesive categories as presented in Section~\ref{sec:cond}. For categories satisfying either of the assumptions mentioned earlier, it is possible to refine a central construction of the theory of conditions, the so-called \emph{shift construction}, into a form that leads to new results on the interplay of rewriting rules and conditions. In particular, the \emph{transport construction} is shown to possess a compatibility property in interaction with the \emph{shift construction}, which is required to ensure compositionality of rewriting with conditions. 

The main results of this paper are the \emph{compositional concurrency} and \emph{associativity} of rewriting with conditions in both the DPO (Section~\ref{sec:DPO}) and SqPO (Section~\ref{sec:SqPO}) cases, demonstrated to hold under the aforementioned assumptions. Our proof strategy and techniques rely heavily on earlier developments in the setting of rewriting without conditions~\cite{bp2018,nbSqPO2019} in conjunction with the aforementioned results on the properties of the refined shift and transport constructions. While admittedly a rather technical work, we believe that our results can serve as a starting point for a new generation of developments in the field of rewriting, in particular in view of static analysis tasks. Indeed, in most applications of practical interest, idealized data structures such as multigraphs must be restricted to more rigid structures (such as e.g.\ site graphs in the \textsc{Kappa} framework~\cite{Danos:2004aa,Boutillier:2018aa}) in order to obtain tractable algorithms of sufficient predictive power. First results in the fields of stochastic mechanics of continuous-time Markov chains based upon stochastic rewriting systems~\cite{bdg2016,bp2018,nbSqPO2019,bdg2019} hint at a great potential of the framework of compositional rewriting with conditions as presented here.  We recently demonstrated that based upon the results of the present paper,  it is possible to develop a faithful encoding of both bio- and organo-chemical reaction systems via typed undirected simple graphs with suitable sets of additional structure constraints~\cite{BK2020}. Together with the respective SqPO- and DPO-type stochastic mechanics frameworks introduced in~\cite{BK2020}, %
a direct and in-detail comparison of the sophisticated rule-based modeling frameworks \textsc{Kappa}~\cite{Boutillier:2018aa} and \textsc{BioNetGen}~\cite{Harris:2016aa} (SqPO/biochemistry) as well as \textsc{M{\O}D}~\cite{Andersen_2016} (DPO/organic chemistry) with the general theory of categorical rewriting over $\cM$-adhesive categories has become possible. As a first hint at the potential of such an approach in view of the development of algorithmic implementations of rewriting-theoretical concepts, we have recently presented a first prototype of an implementation of SqPO-type compositional rewriting systems for rules with conditions based upon the present theory as well as on the \textsc{Microsoft Z3} SMT-solver in~\cite{behr2020commutators}. Work in progress further includes the theory of \emph{tracelets} as introduced in~\cite{behr2019tracelets}, whereby the classical concept of \emph{derivation traces} (i.e.\ of sequences of applications of rewriting rules) is analyzed via exploiting the associativity theorems in order to characterize a derivation trace of length $n\geq1$ in terms of a \emph{minimal} derivation trace of the same length (a so-called \emph{tracelet} of length $n$). As briefly touched upon in the discussion of Examples~\ref{ex:assocDPO} and~\ref{ex:assocSqPO} (which consider the case of a ``triple'' composite of rules, i.e.\ the case $n=3$), compositional associativity of the operations of DPO- and SqPO-type rule compositions (Theorems~\ref{thm:assocAC} and~\ref{thm:assocACsqpo}) opens the possibility to \emph{statically generate} such minimal derivation traces \emph{without} the need to first obtain generic derivation traces (either from simulations or from unfoldings), posing thus considerable potential in terms of analyzing \emph{pathways} and their dynamics in chemical reaction systems.

\appendix\label{sec:appendix}

\section{Collection of technical lemmata for $\cM$-adhesive categories}\label{app:lem}

In many practical computations in the framework of $\cM$-adhesive categories, one may take advantage of a number of technical results, some of which elementary, some of which rather specialized (such as in particular the lemmata pertaining to final pullback complements (FPCs)). For the readers' convenience, we provide here the full list of results used in the framework of this paper. The list is an adaptation of the list provided in~\cite{nbSqPO2019} from the setting of adhesive to $\cM$-adhesive categories. Note also that while the category-theoretical constructions of objects and morphisms via pullbacks, pushouts, pushout complements and FPCs are by definition unique only up to universal isomorphisms, we follow standard practice in simplifying our notations by employing a convention whereby e.g.\ the pushout of an isomorphism as in~\eqref{eq:lemMain11}(A) is denoted by ``equality arrows'' (rather than keeping a notation with generic labels and $\cong$ decorations on arrows). 

\begin{lemma}\label{lem:Main}
Let $\bfC$ be a category.
\begin{romanenumerate}
\item \emph{``Single-square'' lemmas~(see e.g.\ \cite{bp2019-ext}, Lem.~1.7):}
In any category, given commutative diagrams of the form
 \begin{equation}\label{eq:lemMain11}
 \inputtikz{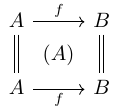}\qquad \inputtikz{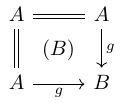}\qquad \inputtikz{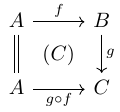}\qquad \inputtikz{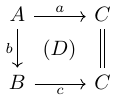}\,,
  \end{equation}
\begin{enumerate}
\item $(A)$ is a pushout for arbitrary morphisms $f$\label{lem:fPOPB}, 
\item $(B)$ is a pullback if and only if the morphism $g$ is a monomorphism\label{lem:monoPB},
\item $(C)$ is a pullback for arbitrary morphisms $f$ if $g$ is a monomorphism\label{lem:idPB},
\item If $a,c\in \mono{\bfC}$ and $(D)$ is a pullback, then $A\cong B$\label{lem:PBiso}.
\end{enumerate}
\item \emph{special $\cM$-adhesivity corollaries} (cf.\ e.g.\ \cite{ehrig2014mathcal}, Lemma~2.6): in any adhesive category,
\begin{enumerate}
  \item pushouts along $\cM$-morphisms are also pullbacks,
  \item (\emph{uniqueness of pushout complements}) given an $\cM$-morphism $A\hookrightarrow C$ and a generic morphism $C\rightarrow D$, the respective pushout complement $A\rightarrow B \xhookrightarrow{b} D$ (if it exists) is unique up to isomorphism, and with $b\in \cM$ (due to stability of $\cM$-morphisms under pushouts).
\end{enumerate}
\item \emph{``Double-square lemmas''}: given commutative diagrams of the shapes
\begin{equation}
\inputtikz{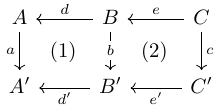}\qquad
\inputtikz{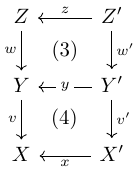}
\end{equation}
then in any category $\bfC$ (cf.\ e.g.\ \cite{lack2005adhesive}):
\begin{enumerate}
\item \emph{Pullback-pullback (de-)composition}: If $(1)$ is a pullback, then $(1)+(2)$ is a pullback if and only if $(2)$ is a pullback.\label{lem:PBPBdec}
\item \emph{Pushout-pushout (de-)composition}: If $(2)$ is a pushout, then $(1)+(2)$ is a pushout if and only if $(1)$ is a pushout.\label{lem:POPOdec}
\end{enumerate}
If the category is $\cM$-adhesive: 
\begin{enumerate}\setcounter{enumii}{2}
\item \emph{pushout-pullback decomposition} (\cite{ehrig2014mathcal}, Lemma~2.6): If $(1)+(2)$ is a pushout, $(1)$ is a pullback, and if $d'\in \cM$ and ($c\in \cM$ or $e\in \cM$), then $(1)$ and $(2)$ are both pushouts (and thus also pullbacks).\label{lem:POPBdec}
\item \emph{pullback-pushout decomposition} (\cite{GOLAS2014}, Lem.~B.2): if $(1)+(2)$ is a pullback, $(2)$ a pushout, $(1)$ commutes and $a\in \cM$, then $(1)$ is a pullback.\label{lem:PBPOdec}
\item \emph{Horizontal FPC (de-)composition} (cf.\ \cite{Corradini_2006}, Lem.~2 and Lem.~3, compare~\cite{Loewe_2015}, Prop.~36):\footnote{It is worthwhile emphasizing that in these FPC-related lemmas, the ``orientation'' of the diagrams plays an important role. Moreover, the precise identity of the pair of morphisms that plays the role of the final pullback complement in a given square can be inferred from the ``orientation'' specified in the condition part of each statement.} If $(1)$ is an FPC (i.e.\ if $(d',b)$ is FPC of $(a,d)$), then $(1)+(2)$ is an FPC if and only if $(2)$ is an FPC.\label{lem:horFPCdec}
\item \emph{Vertical FPC (de-)composition} (ibid):\label{lem:vertFPCdec} if $(3)$ is an FPC (i.e.\ if $(v.w')$ is FPC of $(w,z)$), then 
\begin{enumerate}
\item if $(4)$ is an FPC (i.e.\ if $(x,v')$ is FPC of $(v,y)$), then $(3)+(4)$ is an FPC (i.e.\ $(x,v'\circ w')$ is FPC of $(v\circ w,z)$);
\item if $(3)+(4)$ is an FPC (i.e.\ if $(x,v'\circ w')$ is FPC of $(v\circ w,z)$), and if $(4)$ is a pullback, then $(4)$ is an FPC (i.e.\ $(x,v')$ is FPC of $(v,y)$).
\end{enumerate}
\item \emph{Vertical FPC-pullback decomposition} (compare~\cite{Loewe_2015}, Lem.~38): If $v\in \cM$, if $(4)$ is a pullback and if $(3)+(4)$ is an FPC (i.e.\ if $(x,v'\circ w')$ is FPC of $(v\circ w,z)$), then $(3)$ and $(4)$ are FPCs.\label{lem:vertFPCpbDec}
\end{enumerate}
If the category is $\cM$-adhesive and in addition possesses an epi-$\cM$-factorization and $\cM$-effective unions: 
\begin{enumerate}\setcounter{enumii}{7}
\item \emph{Vertical FPC-pushout decomposition}: If all morphisms of the squares $(3)$ and $(4)$ except $v$ are in $\cM$, if $v\circ w\in \cM$, if $(3)$ is a pushout and if $(3)+(4)$ is an FPC (i.e.\ if $(x,v'\circ w')$ is FPC of $(v\circ w,z)$), then $(4)$ is an FPC and $v\in \cM$.\label{lem:vertFPCpoDec}
\end{enumerate}
\end{romanenumerate}
\begin{proof}
Proofs of all but the very last statement are found in the references quoted in each statement. It thus remains to prove our novel vertical FPC-pushout decomposition result (in the setting of $\cM$-adhesive categories). To this end, we first invoke pullback-pushout decomposition (Lemma~\ref{lem:Main}\eqref{lem:PBPOdec}): since $(3)+(4)$ is an FPC and thus in particular a pullback, since $(3)$ is a pushout, $(4)$ commutes, and finally since $x\in \cM$, we find that $(4)$ is a pullback. %

In order to demonstrate that $v\in \cM$, construct the commutative cube below left:
\begin{equation}
\vcenter{\hbox{\includegraphics[scale=0.9]{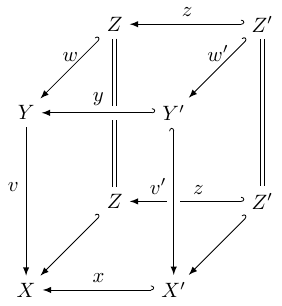}}}\qquad \qquad
\inputtikz{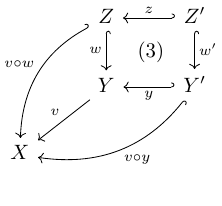}
\end{equation}
Since the bottom square is the FPC (and thus pullback) $(3)+(4)$, and since the right square is a pullback via Lemma~\ref{lem:Main}\eqref{lem:idPB} (because $v'\in \cM\subset \mono{\bfC}$), by pullback composition the square $\cSquare{Z',Z,X,Y'}$ (the right plus the bottom square) is a pullback. Thus assembling the commutative diagram as shown above right, since by assumption $(3)$ is a pushout and all arrows except $v$ are in $\cM$, invoking the property of $\cM$-effective unions (Definition~\ref{def:MeffUn}) allows us to prove that also $v\in \cM$. %
Finally, by applying vertical FPC-pullback decomposition, we conclude that $(4)$ is an FPC. 
\end{proof}
\end{lemma}

\section{Closure of $\cM$ under isomorphisms}\label{app:isoM}

As mentioned in the main text, the definition of $\cM$-adhesive categories as given in Definition~\ref{def:Madh} is well known to be overcomplete. With the $\cM$-decomposition property derivable directly via stability of $\cM$-morphisms under pullbacks, we present below a derivation of the fact that $\cM$ contains all isomorphisms.
\begin{romanenumerate}
	\item Given an arbitrary $\cM$-morphism $(m:A\hookrightarrow B)\in \cM$, since $m= m\circ id_A$,  by the decomposition property of $\cM$-morphisms (Definition~\ref{def:Madh}$(i)(a)$), we find that $id_A\in \cM$, for every $A\in \obj{\bfC}$.
	\item Let $(\phi:A\rightarrow B)\in \iso{\bfC}$ be an isomorphism, and take the pullback of the cospan $(\phi, id_A)$, resulting in a span $(id_A,\phi^{-1})$. Then by stability of $\cM$-morphisms under pullback (Definition~\ref{def:Madh}$(iii)$), we conclude that $\phi^{-1}\in \cM$ (i.e. by $\cM$-decomposition also $\phi\in \cM$).
\end{romanenumerate}

\section{Some useful consequences of epi-$\cM$-factorizations}
\label{app:epiMcons}

\begin{lemma}\label{lem:pbEpiMonoId}
  Let $\bfC$ be an $\cM$-adhesive category satisfying Assumption~\ref{as:DPO}. Then given a commutative diagram as below left where $a,b\in \cM$ and $c$ is an arbitrary morphism,
  \begin{equation}\label{eq:commCubeAuxB}
     \inputtikz{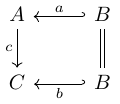}\qquad  \inputtikz{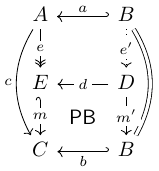}\qquad \vcenter{\hbox{\includegraphics[scale=0.9]{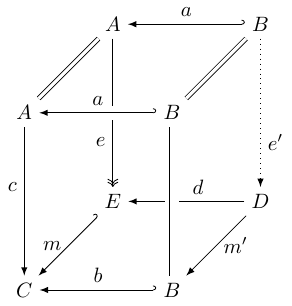}}}
  \end{equation}
then if the epi-$\cM$-factorization of the morphism $c$ reads $c=m\circ e$, with $(m:E\hookrightarrow C)\in \cM$ and  $(e:A\rightarrow E)\in \epi{\bfC}$, and if we form the pullback $D=\pB{E\hookrightarrow C\hookleftarrow B}$ (thus inducing by the universal property a morphism $e':B\rightarrow D$), then $B\cong D$.
\end{lemma}
\begin{proof}
 Note first that by virtue of stability of $\cM$-morphisms under pullbacks and by the decomposition property of $\cM$-morphisms, respectively, one may conclude that $d,m'\in\cM$ and $e'\in \cM$ (since $(id_B:B\rightarrow B)\in \cM$). Next, form the commutative cube diagram as depicted in the right part of~\eqref{eq:commCubeAuxB}. The left square is a pullback (Lemma~\ref{lem:Main}\ref{lem:idPB}), and so is the top square (Lemma~\ref{lem:Main}\ref{lem:fPOPB}). Since by construction also the bottom square is a pullback, by virtue of pullback-pullback decomposition (Lemma~\ref{lem:Main}\ref{lem:PBPBdec}), the right square is a pullback. Thus by stability of isomorphisms under pullbacks (Lemma~\ref{lem:Main}\ref{lem:PBiso}), we obtain that $D\cong B$.
\end{proof}

\section{Some useful consequences of (strict) $\cM$-initiality}\label{app:sMinit}

In the case that an $\cM$-adhesive category possesses a (strict) $\cM$-initial object, the following useful properties can be derived.

\begin{lemma}\label{lem:POCsepcial}
  Let $\bfC$ be an $\cM$-adhesive category with $\cM$-initial object $\mIO$. Then the commutative diagram of $\cM$-morphisms below is both a pushout and a pullback:
\begin{equation}
\inputtikz{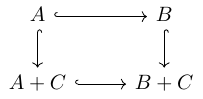}\,.
\end{equation}
\begin{proof}
 Consider the following commutative diagram:
\begin{equation}
\inputtikz{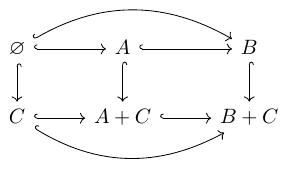}
\end{equation}
Since the outer square and the left square are pushouts, according to the pushout-pushout decomposition property stated in Lemma~\ref{lem:Main}\ref{lem:POPOdec}, the right inner square is also a pushout (and thus a pullback).
\end{proof}
\end{lemma}

\begin{lemma}\label{lem:PBsepcial}
  Let $\bfC$ be an $\cM$-adhesive category with a strict $\cM$-initial object $\mIO$, and consider the commutative diagrams of $\cM$-morphisms below,
\begin{equation}
\inputtikz{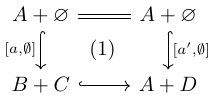}\qquad \inputtikz{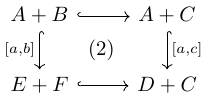}\,.
\end{equation}
If $(1)$ is a pullback, then $A\cong B$, and if $(2)$ is a pullback, then $B\cong F$.
\end{lemma}
\begin{proof}
Consider the following auxiliary commutative diagrams:
\begin{equation}
\vcenter{\hbox{\includegraphics[scale=1]{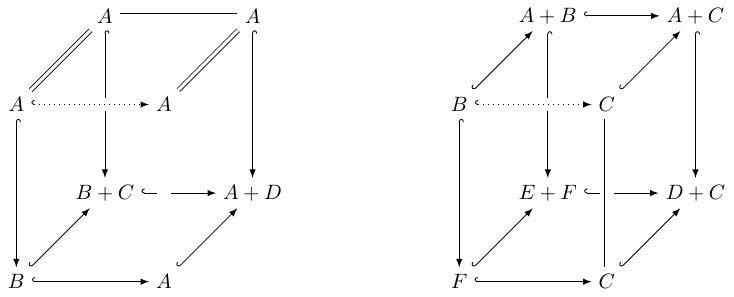}}}\,.
\end{equation}
The left diagram is formed by taking a pullback to obtain the bottom square, followed by taking the appropriate pullbacks to form the left and right squares (which induces the dotted arrow by the universal property of pullbacks). Since the left, back and right squares are pullbacks, by pullback-pullback decomposition (see Lemma~\ref{lem:Main}\ref{lem:PBPBdec}) so is the front square. Then by stability of isomorphisms under pullbacks (see Lemma~\ref{lem:Main}\ref{lem:PBiso}), $A\cong B$. The right diagram is constructed in precisely the same fashion, thus proving that $B\cong F$.
\end{proof}

The notion of \emph{strict} $\cM$-initiality also plays an interesting role in our framework of conditions on objects and rewriting rules studied in the main part of this paper due to the following result:

\begin{lemma}[$\cM$-morphisms into binary coproducts]\label{lem:MintoCop}
  Let $\bfC$ be an $\cM$-adhesive category (for $\cM$ a class of monomorphisms) that possesses a strict $\cM$-initial object $\mIO\in \obj{\bfC}$. Then for all objects $X,Y,Z\in \obj{\bfC}$, if there exists an $\cM$-morphism $X+Y\xhookleftarrow{f}Z$, then $Z\cong V+W$ with
  \begin{equation}\label{eq:lemMintoCopVW}
    V=\pB{X\hookrightarrow X+Y\xhookleftarrow{f} Z}\quad \text{and}\quad 
    W=\pB{Z\xhookrightarrow{f} X+Y\hookleftarrow Y}\,,
  \end{equation}
  and consequently (by the universal property of binary coproducts) $f=[v,w]$ (with $v:V\hookrightarrow X$ and $w:W\hookrightarrow Y$ both in $\cM$).
\end{lemma}
\begin{proof}
Construct the following commutative cube, with the bottom square a pushout, and where the front and left faces are formed by taking pullbacks as described in~\eqref{eq:lemMintoCopVW} in order to obtain $V$ and $W$, followed by forming $Z'=\pB{V\hookrightarrow Z\hookleftarrow W}$, which provides an arrow $Z'\rightarrow \mIO$ by the universal property of pushouts (since the bottom square is a PO):
\begin{equation}
\vcenter{\hbox{\includegraphics[scale=1]{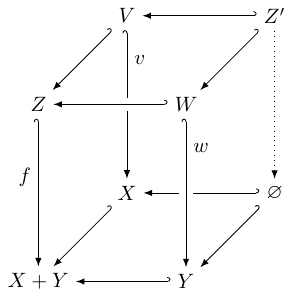}}}
\end{equation}
By virtue of strict $\cM$-initiality of $\mIO$, the existence of the arrow $Z'\rightarrow \mIO$ entails that $Z'\cong \mIO$. Since the bottom square is a pushout, all vertical squares are pullbacks and all vertical morphisms are in $\cM$, by the $\cM$-van Kampen property the top square is a pushout. Thus $Z\cong V+W$ by uniqueness of pushouts up to isomorphisms, and the claim follows.
\end{proof}

\paragraph{Acknowledgments.}

The work of NB was supported by a \emph{Marie Sk\l{}odowska-Curie Individual fellowship} (Grant Agreement No.~753750 -- RaSiR). We would like to thank P.A.M.~Melli\`{e}s, P.~Soboci\'{n}ski and N.~Zeilberger for fruitful discussions.


\begin{thebibliography}{64}
\providecommand{\natexlab}[1]{#1}
\providecommand{\url}[1]{\texttt{#1}}
\expandafter\ifx\csname urlstyle\endcsname\relax
  \providecommand{\doi}[1]{doi: #1}\else
  \providecommand{\doi}{doi: \begingroup \urlstyle{rm}\Url}\fi

\bibitem[Abou-Jaoud{\'{e}} et~al.(2016)Abou-Jaoud{\'{e}}, Thieffry, and
  Feret]{Abou_Jaoud__2016}
Wassim Abou-Jaoud{\'{e}}, Denis Thieffry, and J{\'{e}}r{\^{o}}me Feret.
\newblock Formal derivation of qualitative dynamical models from biochemical
  networks.
\newblock \emph{Biosystems}, 149:\penalty0 70--112, 2016.
\newblock \doi{https://doi.org/10.1016/j.biosystems.2016.09.001}.

\bibitem[Ad{\'a}mek et~al.(1990)Ad{\'a}mek, Herrlich, and
  Strecker]{adamek1990abstract}
Ji{\v{r}}{\'\i} Ad{\'a}mek, Horst Herrlich, and George Strecker.
\newblock Abstract and concrete categories, 1990.

\bibitem[Andersen et~al.(2016)Andersen, Flamm, Merkle, and
  Stadler]{Andersen_2016}
Jakob~L. Andersen, Christoph Flamm, Daniel Merkle, and Peter~F. Stadler.
\newblock {A Software Package for Chemically Inspired Graph Transformation}.
\newblock In R.~Echahed and M.~Minas, editors, \emph{Graph Transformation (ICGT
  2016)}, volume 9761 of \emph{Lecture Notes in Computer Science,}, pages
  73--88, Cham, 2016. Springer International Publishing.
\newblock \doi{https://doi.org/10.1007/978-3-319-40530-8_5}.

\bibitem[Andersen et~al.(2019)Andersen, Flamm, Merkle, and
  Stadler]{Andersen_2019}
Jakob~L. Andersen, Christoph Flamm, Daniel Merkle, and Peter~F. Stadler.
\newblock {Chemical Transformation Motifs --- Modelling Pathways as Integer
  Hyperflows}.
\newblock \emph{{IEEE}/{ACM} Transactions on Computational Biology and
  Bioinformatics}, 16\penalty0 (2):\penalty0 510--523, 2019.
\newblock \doi{https://doi.org/10.1109/tcbb.2017.2781724}.

\bibitem[Andersen et~al.(2018{\natexlab{a}})Andersen, Fagerberg, Flamm,
  Kianian, Merkle, and Stadler]{andersen2018towards}
Jakob~Lykke Andersen, Rolf Fagerberg, Christoph Flamm, Rojin Kianian, Daniel
  Merkle, and Peter~F Stadler.
\newblock Towards mechanistic prediction of mass spectra using graph
  transformation.
\newblock \emph{MATCH Commun. Math. Comput. Chem.}, 80:\penalty0 705--731,
  2018{\natexlab{a}}.

\bibitem[Andersen et~al.(2018{\natexlab{b}})Andersen, Flamm, Merkle, and
  Stadler]{andersen2018rule}
Jakob~Lykke Andersen, Christoph Flamm, Daniel Merkle, and Peter~F Stadler.
\newblock Rule composition in graph transformation models of chemical
  reactions.
\newblock \emph{MATCH Commun. Math. Comput. Chem.}, 80\penalty0
  (661-704):\penalty0 45, 2018{\natexlab{b}}.

\bibitem[Banzhaf et~al.(2015)Banzhaf, Flamm, Merkle, and
  Stadler]{banzhaf_et_al:DR:2015:4968}
Wolfgang Banzhaf, Christoph Flamm, Daniel Merkle, and Peter~F. Stadler.
\newblock {Algorithmic Cheminformatics (Dagstuhl Seminar 14452)}.
\newblock \emph{Dagstuhl Reports}, 4\penalty0 (11):\penalty0 22--39, 2015.
\newblock \doi{https://doi.org/10.4230/DagRep.4.11.22}.

\bibitem[Behr(2019)]{nbSqPO2019}
Nicolas Behr.
\newblock {Sesqui-Pushout Rewriting: Concurrency, Associativity and Rule
  Algebra Framework}.
\newblock In Rachid Echahed and Detlef Plump, editors, \emph{{Proceedings of
  theTenth International Workshop on Graph Computation Models (GCM 2019) in
  Eindhoven, The Netherlands}}, volume 309 of \emph{Electronic Proceedings in
  Theoretical Computer Science}, pages 23--52. Open Publishing Association,
  2019.
\newblock \doi{https://doi.org/10.4204/eptcs.309.2}.

\bibitem[Behr(2020)]{behr2019tracelets}
Nicolas Behr.
\newblock Tracelets and tracelet analysis of compositional rewriting systems.
\newblock In John Baez and Bob Coecke, editors, \emph{{\rm Proceedings} Applied
  Category Theory 2019, {\rm University of Oxford, UK, 15-19 July 2019}},
  volume 323 of \emph{Electronic Proceedings in Theoretical Computer Science},
  pages 44--71. Open Publishing Association, 2020.
\newblock \doi{https://doi.org/10.4204/EPTCS.323.4}.

\bibitem[Behr and Krivine(2020)]{BK2020}
Nicolas Behr and Jean Krivine.
\newblock {Rewriting theory for the life sciences: A unifying framework for
  CTMC semantics}.
\newblock In Fabio Gadducci and Timo Kehrer, editors, \emph{Graph
  Transformation, 13th International Conference, ICGT 2020, Held as Part of
  STAF 2020, Bergen, Norway, June 25--26, 2020, Proceedings}, volume 12150 of
  \emph{Theoretical Computer Science and General Issues}. Springer
  International Publishing, 2020.
\newblock \doi{https://doi.org/10.1007/978-3-030-51372-6}.

\bibitem[Behr and Sobocinski(2018)]{bp2018}
Nicolas Behr and Pawel Sobocinski.
\newblock {Rule Algebras for Adhesive Categories}.
\newblock In Dan Ghica and Achim Jung, editors, \emph{27th EACSL Annual
  Conference on Computer Science Logic (CSL 2018)}, volume 119 of \emph{Leibniz
  International Proceedings in Informatics (LIPIcs)}, pages 11:1--11:21,
  Dagstuhl, Germany, 2018. Schloss Dagstuhl--Leibniz-Zentrum fuer Informatik.
\newblock \doi{https://doi.org/10.4230/LIPIcs.CSL.2018.11}.

\bibitem[Behr and Sobocinski(2020)]{bp2019-ext}
Nicolas Behr and Pawel Sobocinski.
\newblock {Rule Algebras for Adhesive Categories (extended journal version)}.
\newblock \emph{{Logical Methods in Computer Science}}, {Volume 16, Issue 3},
  July 2020.
\newblock URL \url{https://lmcs.episciences.org/6615}.

\bibitem[Behr et~al.(2016{\natexlab{a}})Behr, Danos, and Garnier]{bdg2016}
Nicolas Behr, Vincent Danos, and Ilias Garnier.
\newblock Stochastic mechanics of graph rewriting.
\newblock In \emph{Proceedings of the 31st Annual {ACM}/{IEEE} Symposium on
  Logic in Computer Science - {LICS} {'}16}. {ACM} Press, 2016{\natexlab{a}}.
\newblock \doi{https://doi.org/10.1145/2933575.2934537}.

\bibitem[Behr et~al.(2016{\natexlab{b}})Behr, Danos, Garnier, and
  Heindel]{bdgh2016}
Nicolas Behr, Vincent Danos, Ilias Garnier, and Tobias Heindel.
\newblock {The algebras of graph rewriting}.
\newblock \emph{\href{https://arxiv.org/abs/1612.06240}{arXiv preprint
  arXiv:1612.06240}}, 2016{\natexlab{b}}.

\bibitem[Behr et~al.(2020{\natexlab{a}})Behr, Danos, and Garnier]{bdg2019}
Nicolas Behr, Vincent Danos, and Ilias Garnier.
\newblock {Combinatorial Conversion and Moment Bisimulation for Stochastic
  Rewriting Systems}.
\newblock \emph{{Logical Methods in Computer Science}}, {Volume 16, Issue 3},
  July 2020{\natexlab{a}}.
\newblock URL \url{https://lmcs.episciences.org/6628}.

\bibitem[Behr et~al.(2020{\natexlab{b}})Behr, Heckel, and
  Ghaffari~Saadat]{behr2020commutators}
Nicolas Behr, Reiko Heckel, and Maryam Ghaffari~Saadat.
\newblock {Efficient Computation of Graph Overlaps for Rule Composition: Theory
  and Z3 Prototyping}.
\newblock In Berthold Hoffmann and Mark Minas, editors, \emph{{\rm Proceedings
  of the Eleventh International Workshop on} Graph Computation Models, {\rm
  Online-Workshop, 24th June 2020}}, volume 330 of \emph{Electronic Proceedings
  in Theoretical Computer Science}, pages 126--144. Open Publishing
  Association, 2020{\natexlab{b}}.
\newblock \doi{https://doi.org/10.4204/EPTCS.330.8}.

\bibitem[Benk\"{o} et~al.(2003)Benk\"{o}, Flamm, and Stadler]{Benk__2003}
Gil Benk\"{o}, Christoph Flamm, and Peter~F. Stadler.
\newblock {A Graph-Based Toy Model of Chemistry}.
\newblock \emph{Journal of Chemical Information and Computer Sciences},
  43\penalty0 (4):\penalty0 1085--1093, 2003.
\newblock \doi{https://doi.org/10.1021/ci0200570}.

\bibitem[Blinov et~al.(2004)Blinov, Faeder, Goldstein, and
  Hlavacek]{Blinov:2004aa}
M.~L. Blinov, J.~R. Faeder, B.~Goldstein, and W.~S. Hlavacek.
\newblock {BioNetGen: software for rule-based modeling of signal transduction
  based on the interactions of molecular domains}.
\newblock \emph{Bioinformatics}, 20\penalty0 (17):\penalty0 3289--3291, 2004.
\newblock \doi{https://doi.org/10.1093/bioinformatics/bth378}.

\bibitem[Boehm et~al.(1987)Boehm, Fonio, and Habel]{Boehm:1987aa}
Paul Boehm, Harald-Reto Fonio, and Annegret Habel.
\newblock {Amalgamation of graph transformations: A synchronization mechanism}.
\newblock \emph{Journal of Computer and System Sciences}, 34\penalty0
  (2-3):\penalty0 377--408, 1987.
\newblock \doi{https://doi.org/10.1016/0022-0000(87)90030-4}.

\bibitem[Boutillier et~al.(2018)Boutillier, Maasha, Li, Medina-Abarca, Krivine,
  Feret, Cristescu, Forbes, and Fontana]{Boutillier:2018aa}
Pierre Boutillier, Mutaamba Maasha, Xing Li, H{\'{e}}ctor~F Medina-Abarca, Jean
  Krivine, J{\'{e}}r{\^{o}}me Feret, Ioana Cristescu, Angus~G Forbes, and
  Walter Fontana.
\newblock {The Kappa platform for rule-based modeling}.
\newblock \emph{Bioinformatics}, 34\penalty0 (13):\penalty0 i583--i592, 2018.
\newblock \doi{https://doi.org/10.1093/bioinformatics/bty272}.

\bibitem[Braatz and Brandt(2008)]{braatz2008graph}
Benjamin Braatz and Christoph Brandt.
\newblock Graph transformations for the resource description framework.
\newblock \emph{Electronic Communications of the EASST}, 10, 2008.
\newblock \doi{https://doi.org/10.14279/tuj.eceasst.10.158}.

\bibitem[Braatz et~al.(2014)Braatz, Ehrig, Gabriel, and Golas]{Braatz:2010aa}
Benjamin Braatz, Hartmut Ehrig, Karsten Gabriel, and Ulrike Golas.
\newblock {Finitary $\mathcal{M}$ -adhesive categories}.
\newblock \emph{Mathematical Structures in Computer Science}, 24\penalty0
  (4):\penalty0 240403--240443, 2014.
\newblock \doi{https://doi.org/10.1017/S0960129512000321}.

\bibitem[Camporesi et~al.(2017)Camporesi, Feret, and L{\'{y}}]{Camporesi_2017}
Ferdinanda Camporesi, J{\'{e}}r{\^{o}}me Feret, and Kim~Quy{\^{e}}n L{\'{y}}.
\newblock {KaDE: A Tool to Compile Kappa Rules into (Reduced) {ODE} Models}.
\newblock In \emph{Computational Methods in Systems Biology}, pages 291--299.
  Springer International Publishing, 2017.
\newblock \doi{https://doi.org/10.1007/978-3-319-67471-1_18}.

\bibitem[Cockett and Lack(2003)]{Cockett_2003}
J.R.B. Cockett and Stephen Lack.
\newblock {Restriction categories {II}: partial map classification}.
\newblock \emph{Theoretical Computer Science}, 294\penalty0 (1-2):\penalty0
  61--102, 2003.
\newblock \doi{https://doi.org/10.1016/s0304-3975(01)00245-6}.

\bibitem[Corradini et~al.(1997)Corradini, Montanari, Rossi, Ehrig, Heckel, and
  L{\"{o}}we]{CorradiniMREHL97}
Andrea Corradini, Ugo Montanari, Francesca Rossi, Hartmut Ehrig, Reiko Heckel,
  and Michael L{\"{o}}we.
\newblock {Algebraic Approaches to Graph Transformation - Part I: Basic
  Concepts and Double Pushout Approach}.
\newblock In \emph{Handbook of Graph Grammars and Computing by Graph
  Transformations, Volume 1: Foundations}, pages 163--246, 1997.

\bibitem[Corradini et~al.(2006)Corradini, Heindel, Hermann, and
  K\"{o}nig]{Corradini_2006}
Andrea Corradini, Tobias Heindel, Frank Hermann, and Barbara K\"{o}nig.
\newblock {Sesqui-Pushout Rewriting}.
\newblock In A.~Corradini, H.~Ehrig, U.~Montanari, L.~Ribeiro, and
  G.~Rozenberg, editors, \emph{Graph Transformations (ICGT 2006)}, volume 4178
  of \emph{Lecture Notes in Computer Science}, pages 30--45. Springer Berlin
  Heidelberg, 2006.
\newblock \doi{https://doi.org/10.1007/11841883_4}.

\bibitem[Corradini et~al.(2015)Corradini, Duval, Echahed, Prost, and
  Ribeiro]{Corradini_2015}
Andrea Corradini, Dominique Duval, Rachid Echahed, Frederic Prost, and Leila
  Ribeiro.
\newblock {AGREE {\textendash} Algebraic Graph Rewriting with Controlled
  Embedding}.
\newblock In F.~Parisi-Presicce and B.~Westfechtel, editors, \emph{Graph
  Transformation (ICGT 2015)}, volume 9151 of \emph{Lecture Notes in Computer
  Science}, pages 35--51, Cham, 2015. Springer International Publishing.
\newblock \doi{https://doi.org/10.1007/978-3-319-21145-9_3}.

\bibitem[Danos and Laneve(2003{\natexlab{a}})]{Danos:2003aa}
Vincent Danos and Cosimo Laneve.
\newblock {Graphs for Core Molecular Biology}.
\newblock In C.~Priami, editor, \emph{Computational Methods in Systems Biology
  (CMSB 2003)}, volume 2602 of \emph{Lecture Notes in Computer Science}, pages
  34--46. Springer Berlin Heidelberg, 2003{\natexlab{a}}.
\newblock \doi{https://doi.org/10.1007/3-540-36481-1_4}.

\bibitem[Danos and Laneve(2003{\natexlab{b}})]{Danos:2003ab}
Vincent Danos and Cosimo Laneve.
\newblock {Core Formal Molecular Biology}.
\newblock In P.~Degano, editor, \emph{Programming Languages and Systems (ESOP
  2003)}, volume 2618 of \emph{Lecture Notes in Computer Science}, pages
  302--318. Springer Berlin Heidelberg, 2003{\natexlab{b}}.
\newblock \doi{https://doi.org/10.1007/3-540-36575-3_21}.

\bibitem[Danos and Laneve(2004)]{Danos:2004aa}
Vincent Danos and Cosimo Laneve.
\newblock {Formal molecular biology}.
\newblock \emph{Theoretical Computer Science}, 325\penalty0 (1):\penalty0 69 --
  110, 2004.
\newblock \doi{https://doi.org/10.1016/j.tcs.2004.03.065}.

\bibitem[Danos et~al.(2007)Danos, Feret, Fontana, Harmer, and
  Krivine]{Danos:ab}
Vincent Danos, J\'{e}r\^{o}me Feret, Walter Fontana, Russell Harmer, and Jean
  Krivine.
\newblock {Rule-Based Modelling of Cellular Signalling}.
\newblock In L.~Caires and V.T. Vasconcelos, editors, \emph{Concurrency Theory
  (CONCUR 2007)}, volume 4703 of \emph{Lecture Notes in Computer Science},
  pages 17--41. Springer Berlin Heidelberg, 2007.
\newblock \doi{https://doi.org/10.1007/978-3-540-74407-8_3}.

\bibitem[Danos et~al.(2008)Danos, Feret, Fontana, Harmer, and
  Krivine]{danos2008rule}
Vincent Danos, J{\'e}r{\^o}me Feret, Walter Fontana, Russell Harmer, and Jean
  Krivine.
\newblock Rule-based modelling, symmetries, refinements.
\newblock In Jasmin Fisher, editor, \emph{Formal Methods in Systems Biology
  (FMSB 2008)}, volume 5054 of \emph{Lecture Notes in Computer Science}, pages
  103--122. Springer Berlin Heidelberg, 2008.
\newblock \doi{https://doi.org/10.1007/978-3-540-68413-8_8}.

\bibitem[Danos et~al.(2012)Danos, Feret, Fontana, Harmer, Hayman, Krivine,
  Thompson-Walsh, and Winskel]{danos2012graphs}
Vincent Danos, Jerome Feret, Walter Fontana, Russell Harmer, Jonathan Hayman,
  Jean Krivine, Chris Thompson-Walsh, and Glynn Winskel.
\newblock {Graphs, Rewriting and Pathway Reconstruction for Rule-Based Models}.
\newblock In Deepak D'Souza, Telikepalli Kavitha, and Jaikumar Radhakrishnan,
  editors, \emph{IARCS Annual Conference on Foundations of Software Technology
  and Theoretical Computer Science (FSTTCS 2012)}, volume~18 of \emph{Leibniz
  International Proceedings in Informatics (LIPIcs)}, pages 276--288, Dagstuhl,
  Germany, 2012. Schloss Dagstuhl--Leibniz-Zentrum fuer Informatik.
\newblock \doi{https://doi.org/10.4230/LIPIcs.FSTTCS.2012.276}.

\bibitem[Danos et~al.(2014)Danos, Heindel, Honorato-Zimmer, and
  Stucki]{reversibleSqPO}
Vincent Danos, Tobias Heindel, Ricardo Honorato-Zimmer, and Sandro Stucki.
\newblock {Reversible Sesqui-Pushout Rewriting}.
\newblock In Holger Giese and Barbara K{\"o}nig, editors, \emph{Graph
  Transformation (ICGT 2014)}, volume 8571 of \emph{Lecture Notes in Computer
  Science}, pages 161--176, Cham, 2014. Springer International Publishing.
\newblock \doi{https://doi.org/10.1007/978-3-319-09108-2_11}.

\bibitem[Ehrig and Habel(1986)]{EhrigHabel1986}
H.~Ehrig and A.~Habel.
\newblock \emph{{Graph Grammars with Application Conditions}}.
\newblock Springer Berlin Heidelberg, Berlin, Heidelberg, 1986.
\newblock \doi{https://doi.org/10.1007/978-3-642-95486-3_7}.

\bibitem[{Ehrig} et~al.(1973){Ehrig}, {Pfender}, and {Schneider}]{ehrig1973}
H.~{Ehrig}, M.~{Pfender}, and H.~J. {Schneider}.
\newblock {Graph-grammars: An algebraic approach}.
\newblock In \emph{{14th Annual Symposium on Switching and Automata Theory
  (SWAT 1973)}}, pages 167--180, 1973.
\newblock \doi{https://doi.org/10.1109/SWAT.1973.11}.

\bibitem[Ehrig et~al.(2006{\natexlab{a}})Ehrig, Ehrig, Prange, and
  Taentzer]{ehrig:2006aa}
H.~Ehrig, K.~Ehrig, U.~Prange, and G.~Taentzer.
\newblock Fundamentals of algebraic graph transformation.
\newblock \emph{Monographs in Theoretical Computer Science. An EATCS Series},
  2006{\natexlab{a}}.
\newblock \doi{https://doi.org/10.1007/3-540-31188-2}.

\bibitem[Ehrig et~al.(1991)Ehrig, Habel, Kreowski, and
  Parisi-Presicce]{ehrig1991parallelism}
Hartmut Ehrig, Annegret Habel, Hans-J{\"o}rg Kreowski, and Francesco
  Parisi-Presicce.
\newblock Parallelism and concurrency in high-level replacement systems.
\newblock \emph{Mathematical Structures in Computer Science}, 1\penalty0
  (03):\penalty0 361, 1991.
\newblock \doi{https://doi.org/10.1017/s0960129500001353}.

\bibitem[Ehrig et~al.(2006{\natexlab{b}})Ehrig, Padberg, Prange, and
  Habel]{ehrig2006adhesive}
Hartmut Ehrig, Julia Padberg, Ulrike Prange, and Annegret Habel.
\newblock {Adhesive High-Level Replacement Systems: A New Categorical Framework
  for Graph Transformation}.
\newblock \emph{Fundamenta Informaticae}, 74\penalty0 (1):\penalty0 1--29,
  2006{\natexlab{b}}.
\newblock \doi{https://doi.org/10.5555/1231199.1231201}.

\bibitem[Ehrig et~al.(2008)Ehrig, Heckel, Rozenberg, and
  Taentzer]{ehrig2008graph}
Hartmut Ehrig, Reiko Heckel, Grzegorz Rozenberg, and Gabriele Taentzer,
  editors.
\newblock \emph{{Graph Transformations (ICGT 2008)}}, volume 5214 of
  \emph{Lecture Notes in Computer Science}.
\newblock Springer Berlin Heidelberg, 2008.
\newblock \doi{https://doi.org/10.1007/978-3-540-87405-8}.

\bibitem[Ehrig et~al.(2010)Ehrig, Golas, Hermann, et~al.]{ehrig2010categorical}
Hartmut Ehrig, Ulrike Golas, Frank Hermann, et~al.
\newblock {Categorical frameworks for graph transformation and HLR systems
  based on the DPO approach}.
\newblock \emph{Bulletin of the EATCS}, \penalty0 (102):\penalty0 111--121,
  2010.

\bibitem[Ehrig et~al.(2014)Ehrig, Golas, Habel, Lambers, and
  Orejas]{ehrig2014mathcal}
Hartmut Ehrig, Ulrike Golas, Annegret Habel, Leen Lambers, and Fernando Orejas.
\newblock {$\mathcal{M}$-adhesive transformation systems with nested
  application conditions. Part 1: parallelism, concurrency and amalgamation}.
\newblock \emph{Mathematical Structures in Computer Science}, 24\penalty0 (04),
  2014.
\newblock \doi{https://doi.org/10.1017/s0960129512000357}.

\bibitem[Fagerberg et~al.(2018)Fagerberg, Flamm, Kianian, Merkle, and
  Stadler]{Fagerberg_2018}
Rolf Fagerberg, Christoph Flamm, Rojin Kianian, Daniel Merkle, and Peter~F.
  Stadler.
\newblock {Finding the K best synthesis plans}.
\newblock \emph{Journal of Cheminformatics}, 10\penalty0 (1), 2018.
\newblock \doi{https://doi.org/10.1186/s13321-018-0273-z}.

\bibitem[Feret et~al.(2012)Feret, Henzinger, Koeppl, and Petrov]{Feret2012137}
Jerome Feret, Thomas Henzinger, Heinz Koeppl, and Tatjana Petrov.
\newblock Lumpability abstractions of rule-based systems.
\newblock \emph{Theoretical Computer Science}, 431\penalty0 (0):\penalty0 137
  -- 164, 2012.
\newblock \doi{https://doi.org/10.1016/j.tcs.2011.12.059}.

\bibitem[Golas et~al.(2014)Golas, Habel, and Ehrig]{GOLAS2014}
Ulrike Golas, Annegret Habel, and Hartmut Ehrig.
\newblock {Multi-amalgamation of rules with application conditions in
  $\mathcal{M}$-adhesive categories}.
\newblock \emph{Mathematical Structures in Computer Science}, 24\penalty0 (04),
  2014.
\newblock \doi{https://doi.org/10.1017/s0960129512000345}.

\bibitem[Habel and Pennemann(2009)]{habel2009correctness}
Annegret Habel and Karl-Heinz Pennemann.
\newblock Correctness of high-level transformation systems relative to nested
  conditions.
\newblock \emph{Mathematical Structures in Computer Science}, 19\penalty0
  (02):\penalty0 245, 2009.
\newblock \doi{https://doi.org/10.1017/s0960129508007202}.

\bibitem[Habel and Plump(2012{\natexlab{a}})]{Habel:2012aa}
Annegret Habel and Detlef Plump.
\newblock {$\mathcal{M}, \mathcal{N}$ -Adhesive Transformation Systems}.
\newblock In H.~Ehrig, G.~Engels, H.J. Kreowski, and G.~Rozenberg, editors,
  \emph{Graph Transformations (ICGT 2012)}, volume 7562 of \emph{Lecture Notes
  in Computer Science}, pages 218--233. Springer Berlin Heidelberg,
  2012{\natexlab{a}}.
\newblock \doi{https://doi.org/10.1007/978-3-642-33654-6_15}.

\bibitem[Habel and Plump(2012{\natexlab{b}})]{Habel:2012aa-ext}
Annegret Habel and Detlef Plump.
\newblock {$\mathcal M, \mathcal N$ -Adhesive Transformation Systems}.
\newblock (Long Version), 2012{\natexlab{b}}.
\newblock URL
  \url{http://formale-sprachen.informatik.uni-oldenburg.de/~skript/fs-pub/HaPl12b.pdf}.

\bibitem[Habel et~al.(1996)Habel, Heckel, and Taentzer]{Habel1996}
Annegret Habel, Reiko Heckel, and Gabriele Taentzer.
\newblock Graph grammars with negative application conditions.
\newblock \emph{Fundam. Inf.}, 26\penalty0 (3,4):\penalty0 287--313, 1996.
\newblock \doi{https://doi.org/10.3233/FI-1996-263404}.

\bibitem[Harris et~al.(2016)Harris, Hogg, Tapia, Sekar, Gupta, Korsunsky,
  Arora, Barua, Sheehan, and Faeder]{Harris:2016aa}
Leonard~A. Harris, Justin~S. Hogg, Jos{\'e}-Juan Tapia, John A.~P. Sekar,
  Sanjana Gupta, Ilya Korsunsky, Arshi Arora, Dipak Barua, Robert~P. Sheehan,
  and James~R. Faeder.
\newblock {BioNetGen 2.2: advances in rule-based modeling}.
\newblock \emph{Bioinformatics}, 32\penalty0 (21):\penalty0 3366--3368, 2016.
\newblock \doi{https://doi.org/10.1093/bioinformatics/btw469}.

\bibitem[Hermann et~al.(2014)Hermann, Corradini, and
  Ehrig]{hermann2014analysis}
Frank Hermann, Andrea Corradini, and Hartmut Ehrig.
\newblock {Analysis of permutation equivalence in $\mathcal{M}$-adhesive
  transformation systems with negative application conditions}.
\newblock \emph{Mathematical Structures in Computer Science}, 24\penalty0 (4),
  2014.

\bibitem[Lack and Soboci{\'{n}}ski(2004)]{ls2004adhesive}
Stephen Lack and Pawe{\l} Soboci{\'{n}}ski.
\newblock {Adhesive Categories}.
\newblock In Igor Walukiewicz, editor, \emph{Foundations of Software Science
  and Computation Structures (FoSSaCS 2004)}, volume 2987 of \emph{Lecture
  Notes in Computer Science}, pages 273--288. Springer Berlin Heidelberg, 2004.
\newblock \doi{https://doi.org/10.1007/978-3-540-24727-2_20}.

\bibitem[Lack and Soboci{\'{n}}ski(2005)]{lack2005adhesive}
Stephen Lack and Pawe{\l} Soboci{\'{n}}ski.
\newblock Adhesive and quasiadhesive categories.
\newblock \emph{{RAIRO} - Theoretical Informatics and Applications},
  39\penalty0 (3):\penalty0 511--545, 2005.
\newblock \doi{https://doi.org/10.1051/ita:2005028}.

\bibitem[L\"{o}we(1993)]{loeweSPO}
Michael L\"{o}we.
\newblock {Algebraic Approach to Single-Pushout Graph Transformation}.
\newblock \emph{Theor. Comput. Sci.}, 109\penalty0 (1--2):\penalty0 181--224,
  March 1993.
\newblock ISSN 0304-3975.
\newblock \doi{https://doi.org/10.1016/0304-3975(93)90068-5}.

\bibitem[L\"{o}we(2015)]{Loewe_2015}
Michael L\"{o}we.
\newblock {Polymorphic Sesqui-Pushout Graph Rewriting}.
\newblock In F.~Parisi-Presicce and B.~Westfechtel, editors, \emph{Graph
  Transformation}, volume 9151, pages 3--18, Cham, 2015. Springer International
  Publishing.
\newblock \doi{https://doi.org/10.1007/978-3-319-21145-9_1}.

\bibitem[Murphy et~al.(2010)Murphy, Danos, F{\'{e}}ret, Krivine, and
  Harmer]{Murphy_2010}
Elaine Murphy, Vincent Danos, J{\'{e}}r{\^{o}}me F{\'{e}}ret, Jean Krivine, and
  Russell Harmer.
\newblock Rule-based modeling and model refinement.
\newblock In \emph{Elements of Computational Systems Biology}, pages 83--114.
  John Wiley {\&} Sons, Inc., 2010.
\newblock \doi{https://doi.org/10.1002/9780470556757.ch4}.

\bibitem[Navarro~Gomez et~al.(2016)Navarro~Gomez, Orejas~Vald{\'e}s,
  Pino~Blanco, and Lambers]{navarro2016logic}
Marisa Navarro~Gomez, Fernando Orejas~Vald{\'e}s, Elvira Pino~Blanco, and Leen
  Lambers.
\newblock A logic of graph conditions extended with paths.
\newblock In \emph{Actas de las XVI Jornadas de Programaci{\'o}n y Lenguajes
  (PROLE 2016): Salamanca, septiembre de 2016}, pages 1--15, 2016.

\bibitem[Norris(1998)]{norris}
James~R. Norris.
\newblock \emph{{Markov Chains}}.
\newblock Cambridge Series in Statistical and Probabilistic Mathematics.
  Cambridge University Press, 1998.

\bibitem[Padberg(2017)]{padberg2017towards}
Julia Padberg.
\newblock {Towards M-Adhesive Categories based on Coalgebras and Comma
  Categories}.
\newblock \emph{\href{http://arxiv.org/abs/1702.04650}{arXiv:1702.04650}},
  2017.

\bibitem[Pennemann(2008)]{Pennemann:aa}
Karl-Heinz Pennemann.
\newblock {Resolution-Like Theorem Proving for High-Level Conditions}.
\newblock In H.~Ehrig, R.~Heckel, G.~Rozenberg, and G.~Taentzer, editors,
  \emph{Graph Transformations (ICGT 2008)}, volume 5214 of \emph{Lecture Notes
  in Computer Science}, pages 289--304. Springer Berlin Heidelberg, 2008.
\newblock \doi{https://doi.org/10.1007/978-3-540-87405-8_20}.

\bibitem[Petrov et~al.(2012)Petrov, Feret, and Koeppl]{Petrov_2012}
Tatjana Petrov, Jerome Feret, and Heinz Koeppl.
\newblock Reconstructing species-based dynamics from reduced stochastic
  rule-based models.
\newblock In \emph{Proceedings Title: Proceedings of the 2012 Winter Simulation
  Conference ({WSC})}. {IEEE}, 2012.
\newblock \doi{https://doi.org/10.1109/wsc.2012.6465241}.

\bibitem[Radke(2016)]{radke2016theory}
Hendrik Radke.
\newblock \emph{{A Theory of HR* Graph Conditions and their Application to
  Meta-Modeling}}.
\newblock PhD thesis, Carl von Ossietzky Universit{\"{a}}t Oldenburg,
  Fakult{\"{a}}t II, Department f{\"u}r Informatik, 2016.

\bibitem[Rozenberg(1997)]{DBLP:conf/gg/1997handbook}
Grzegorz Rozenberg, editor.
\newblock \emph{Handbook of Graph Grammars and Computing by Graph
  Transformations, Volume 1: Foundations}.
\newblock World Scientific, 1997.
\newblock \doi{https://doi.org/10.1142/3303}.

\bibitem[Schneider et~al.(2017)Schneider, Lambers, and
  Orejas]{schneider2017symbolic}
Sven Schneider, Leen Lambers, and Fernando Orejas.
\newblock {Symbolic Model Generation for Graph Properties}.
\newblock In M.~Huisman and J.~Rubin, editors, \emph{Fundamental Approaches to
  Software Engineering (FASE 2017)}, volume 10202 of \emph{Lecture Notes in
  Computer Science}, pages 226--243. Springer Berlin Heidelberg, 2017.
\newblock \doi{https://doi.org/10.1007/978-3-662-54494-5_13}.

\end{thebibliography}


\end{document}